\documentclass[11pt]{article}

\usepackage{amsthm,amsmath,amssymb}
\usepackage{natbib}
\usepackage{multirow}
\usepackage[pdftex]{graphicx}
\usepackage{subfigure}
\usepackage{makecell}
\usepackage{booktabs}
\usepackage{array}
\usepackage{fullpage}
\usepackage{url}
\usepackage{algorithm}
\usepackage{algorithmic}
\usepackage{bm}
\usepackage{wrapfig}
\usepackage{lipsum}
\usepackage{mathrsfs}
\usepackage{dsfont}
\usepackage{titling}
\usepackage{epstopdf}
\usepackage{bbm}
\usepackage{relsize}
\usepackage{smile}
\usepackage{natbib}
\usepackage{multirow}
\usepackage{color}
\usepackage{mathtools}
\usepackage{fixltx2e}
\usepackage{caption}
\usepackage{color}

\usepackage{natbib}
\usepackage{multirow}

\usepackage{dcolumn}
\newcolumntype{.}{D{.}{.}{-1}}
\newcolumntype{d}[1]{D{.}{.}{#1}}

\usepackage[usenames,dvipsnames,svgnames,table]{xcolor}
\usepackage[colorlinks=true,
            linkcolor=blue,
            urlcolor=blue,
            citecolor=blue]{hyperref}

\newcommand{\indep}{\mbox{$\perp\!\!\!\perp$}}

\begin{document}

\title{\LARGE Optimal Covariate Balancing Conditions in Propensity
  Score Estimation\thanks{Supported by NSF grants DMS-1854637,
    DMS-1712591, DMS-2053832, CAREER Award DMS-1941945, NIH grant
    R01-GM072611, and Sloan Grant 2020-13946. An earlier version of
    this paper was entitled, ``Improving Covariate Balancing
    Propensity Score: A Doubly Robust and Efficient Approach.''
   }} \author{\normalsize Jianqing Fan\thanks{
    Department of Operations Research and Financial Engineering,
    Princeton University}
  \and Kosuke Imai\thanks{Department of Government
    and Department of Statistics,
    Harvard University}
 \and Inbeom Lee\thanks{Department of Statistics and Data Science,
    Cornell University}    
  \and  Han Liu\thanks{
    Department of Electrical Engineering and Computer Science,  Northwestern University.}
  \and Yang Ning\thanks{Department of Statistics and Data Science,
    Cornell University}
  \and Xiaolin Yang\thanks{Amazon}
  }

  \maketitle


\begin{abstract}

  Inverse probability of treatment weighting (IPTW) is a popular
  method for estimating the average treatment effect (ATE).  However,
  empirical studies show that the IPTW estimators can be sensitive to
  the misspecification of the propensity score model.  To address this
  problem, researchers have proposed to estimate propensity score by
  directly optimizing the balance of pre-treatment covariates.  While
  these methods appear to empirically perform well, little is known
  about how the choice of balancing conditions affects their
  theoretical properties.  To fill this gap, we first characterize the
  asymptotic bias and efficiency of the IPTW estimator based on the
  Covariate Balancing Propensity Score (CBPS) methodology under local
  model misspecification.  Based on this analysis, we show how to
  optimally choose the covariate balancing functions and propose an
  optimal CBPS-based IPTW estimator. This estimator is doubly robust;
  it is consistent for the ATE if either the propensity score model or
  the outcome model is correct.  In addition, the proposed estimator
  is locally semiparametric efficient when both models are correctly
  specified.  To further relax the parametric assumptions, we extend
  our method by using a sieve estimation approach. We show that the
  resulting estimator is globally efficient under a set of much weaker
  assumptions and has a smaller asymptotic bias than the existing
  estimators. Finally, we evaluate the finite sample performance of
  the proposed estimators via simulation and empirical studies.  An
  open-source software package is available for implementing the
  proposed methods.

  \bigskip


\end{abstract}


\section{Introduction}

Suppose that we have a random sample of $n$ units from a population of
interest.  For each unit $i$, we observe $(T_i,Y_i,\bX_i)$, where
$\bX_i\in \RR^d$ is a $d$-dimensional vector of pre-treatment
covariates, $T_i$ is a binary treatment variable, and $Y_i$ is an
outcome variable.  In particular, $T_i$ takes $1$ if unit $i$ receives
the treatment and is equal to $0$ if unit $i$ belongs to the control
group.  The observed outcome can be written as
$Y_i = Y_i(1)T_i+Y_i(0)(1-T_i)$, where $Y_i(1)$ and $Y_i(0)$ are the
potential outcomes under the treatment and control conditions,
respectively.  This notation implicitly requires the stable unit
treatment value assumption \citep{rubi:90}.  In addition, throughout
this paper, we assume the strong ignorability of the treatment
assignment \citep{rose:rubi:83},
\begin{equation}
  \{Y_i(1), Y_i(0)\} \ \indep \ T_i \mid \bX_i \quad {\rm and} \quad 0 \ < \ \PP(T_i = 1 \mid \bX_i) \ < \ 1.
\end{equation}

Next, we assume that the conditional mean functions of potential
outcomes exist and denote them by,
\begin{eqnarray}
  \EE(Y_i(0)\mid \bX_i) &=& K(\bX_i) \quad {\rm and} \quad
  \EE(Y_i(1)\mid \bX_i) \ = \ K(\bX_i)+L(\bX_i), \label{eq:CEFy}
\end{eqnarray}
for some functions $K(\cdot)$ and $L(\cdot)$, which represent the
conditional mean of the potential outcome under the control condition and
the conditional average treatment effect, respectively.  Under this
setting, we are interested in estimating the average treatment effect
(ATE),
\begin{equation}
  \mu \ = \ \EE(Y_i(1) - Y_i(0)) \ = \ \EE(L(\bX_i)). \label{eq:ATE}
\end{equation}

The propensity score is defined as the conditional probability of
treatment assignment \citep{rose:rubi:83},
\begin{eqnarray}
  \pi(\bX_i) & = & \PP(T_i = 1 \mid \bX_i).
\end{eqnarray}
In practice, since $\bX_i$ can be high dimensional, the propensity
score is usually parameterized by a model
$\pi_{\bbeta}(\bX_i)$ where $\bbeta$ is a $q$-dimensional vector of
parameters.  A popular choice  is the logistic regression model, i.e.,
$\pi_{\bbeta}(\bX_i) = \exp(\bX_i^\top \bbeta)/\{1 + \exp(\bX_i^\top
\bbeta)\}$.
Once the parameter $\bbeta$ is estimated (e.g., by the maximum
likelihood estimator $\hat{\bbeta}$), the Horvitz-Thompson estimator \citep{horv:thom:52},
which is based on the inverse probability of treatment weighting
(IPTW), can be used to obtain an estimate of the ATE,
\begin{eqnarray}
\hat{\mu}_{\hat\bbeta} & = & \frac{1}{n} \sum_{i=1}^{n}\left( \frac{T_iY_i}{\pi_{\hat\bbeta}(\bX_i)}-\frac{(1-T_i)Y_i}{1-\pi_{\hat\bbeta}(\bX_i)}\right).
\label{eq:HT}
\end{eqnarray}

However, it has been shown that the IPTW estimator with the known
propensity score does not attain the semiparametric efficiency bound
\citep{hahn1998role}.  A variety of efficient ATE estimators have been
proposed \citep[see e.g.,][among many others]{robins1994estimation,
  bang2005doubly, tan2006distributional,qin2007empirical,robins2007comment,
  cao2009improving,tan:10,van2010targeted,rotnitzky2012improved,han2013estimation,vermeulen2015bias}.
Despite the popularity of these methods, researchers have found that
in practice the estimators can be sensitive to the misspecification of
the propensity score model and the outcome model
\citep[e.g.,][]{MR2420462}. To overcome this problem, several
researchers have recently considered the estimation of the propensity
score by optimizing covariate balance rather than maximizing the
accuracy of predicting treatment assignment
\citep[e.g.,][]{hain:12,grah:pint:egel:12,imai:ratk:14,chan:yam:zhan:16,zubizarreta2015stable,zhao2017entropy,zhao2016covariate}.
In this paper, we focus on the Covariate Balancing Propensity Score
(CBPS) methodology \citep{imai:ratk:14}.  In spite of its simplicity,
several scholars independently found that the CBPS performs well in
practice \citep[e.g.,][]{wyss:etal:14,frol:hube:wies:15}.  The method
can also be extended for the analysis of longitudinal data
\citep{imai:ratk:15}, general treatment regimes
\citep{fong:hazl:imai:18} and high-dimensional propensity score
\citep{ning2018robust}.  In this paper, we conduct a theoretical
investigation of the CBPS.  Given the similarity between the CBPS and
some other methods, our theoretical analysis may also provide new
insights for understanding other covariate balancing methods.  

The CBPS method estimates the parameters of the propensity score
model, $\bbeta$, by solving the following $m$-dimensional estimating equation,
\begin{equation}
  \bar \bg_{\bbeta}(\bT,\bX) \ = \ \frac{1}{n}\sum_{i=1}^n \bg_{\bbeta}(T_i,\bX_i) \ = \ 0
\quad {\rm where} \quad \bg_{\bbeta}(T_i,\bX_i) \ = \ \left(\frac{T_i}{\pi_{\bbeta}(\bX_i)}-\frac{1-T_i}{1-\pi_{\bbeta}(\bX_i)} \right)\fb(\bX_i),
\label{eq:cvb}
\end{equation}
for some covariate balancing function
$\fb(\cdot):\RR^d\rightarrow \RR^m$ when the number of equations $m$
is equal to the number of parameters $q$. \cite{imai:ratk:14} point
out that the common practice of fitting a logistic model is equivalent
to balancing the score function with
$\fb(\bX_i) = \pi_{\bbeta}^{'}(\bX_i) = \partial
\pi_{\bbeta}(\bX_i)/\partial \bbeta$.  They find that choosing
$\fb(\bX_i) = \bX_i$, which balances the first moment between the
treatment and control groups, significantly reduces the bias of the
estimated ATE.  Some researchers also include higher moments and/or
interactions, e.g., $\fb(\bX_i) = (\bX_i\ \bX_i^2)$, in their
applications. This guarantees that the treatment and control groups
have an identical sample mean of $\fb(\bX_i)$ after weighting by the
estimated propensity score.

When $m > q$, then $\hat{\bbeta}$ can be estimated by optimizing the
covariate balance by the generalized method of moments (GMM) method
\citep{hans:82}:
\begin{eqnarray}
\hat \bbeta & = & \argmin_{\bbeta\in\Theta} \ \bar \bg_{\bbeta}(\bT,\bX)^\top \ \widehat\Wb \ \bar \bg_{\bbeta}(\bT,\bX),
\label{eq:gmm}
\end{eqnarray}
where $\Theta$ is the parameter space for $\bbeta$ in $\RR^q$ and
$\widehat\Wb$ is an $(m\times m)$ positive definite weighting matrix,
which we assume in this paper does not depend on $\bbeta$.
Alternatively, the empirical likelihood method can be used
\citep{owen:01}.  Once the estimate of $\bbeta$ is obtained, we can
estimate the ATE using the IPTW estimator in equation~\eqref{eq:HT}.

The main idea of the CBPS and other related methods is to directly
optimize the balance of covariates between the treatment and control
groups so that even when the propensity score model is misspecified we
still obtain a reasonable balance of the covariates between the
treatment and control groups.  However, one open question remains in
this literature: How shall we choose the covariate balancing function
$\fb(\bX_i)$? In particular, if the propensity score model is
misspecified, this problem becomes even more important.  

This paper makes two main contributions. First, we conduct a thorough
theoretical study of the CBPS-based IPTW estimator with an arbitrary
covariate balancing function $\fb(\cdot)$. We characterize the
asymptotic bias and efficiency of this estimator under locally
misspecified propensity score models.  Based on these findings, we
show how to optimally choose the covariate balancing function
$\fb(\bX_i)$ for the CBPS methodology
(Section~\ref{sec:model.misspecification}).

However, the optimal choice of $\fb(\bX_i)$ requires some initial
estimators for the unknown propensity score model and the outcome models. This limits the application of
the CBPS method with the optimal $\fb(\bX_i)$ in practice.  Our second contribution is to 
overcome this problem by developing an optimal CBPS method that does
not require an initial estimator. We show that the IPTW estimator
based on the optimal CBPS (oCBPS) method retains the double
robustness property. The proposed estimator is
semiparametrically efficient when both the propensity score and
outcome models are correctly specified. More importantly, we show that the rate of convergence of the 
proposed oCBPS estimator is faster than the augmented inverse probability weighted (AIPW) estimator \citep{robins1994estimation} under locally misspecified models.  (Section \ref{sec:proposed.methodology}).


To relax the parametric assumptions on the propensity score model and
the outcome model, we further extend the proposed oCBPS method to the
nonparametric settings, by using a sieve estimation
approach \citep{newey1997convergence,chen2007large}. In
Section~\ref{secnon}, we establish the semiparametric efficiency
result for the IPTW estimator under the nonparametric setting.  Compared to the existing nonparametric
propensity score methods
\citep[e.g.,][]{hira:imbe:ridd:03,chan:yam:zhan:16}, our theoretical
results require weaker smoothness assumptions. For instance, the theories in \cite{hira:imbe:ridd:03},
\cite{imbens2005mean} and \cite{chan:yam:zhan:16} require $s/d>7$,
$s/d>9$ and $s/d>13$, respectively, where $s$ is the smoothness
parameter of the corresponding function class and $d=\dim(\bX_i)$.  In
comparison, we only require $s/d>3/4$, which is significantly weaker
than the existing conditions. To prove this result, we exploit the
matrix Bernstein's concentration inequalities
\citep{tropp2015introduction} and a Bernstein-type concentration
inequality for U-statistics \citep{arcones1995bernstein}.
Moreover, we show that our estimator has smaller asymptotic bias than
the usual nonparametric method
\citep[e.g.,][]{hira:imbe:ridd:03}. Therefore, the asymptotic
normality result is expected to be more accurate in practice (Section \ref{secnon}). The proof of the theoretical results are deferred to the supplementary material. 

An open-source {\sf R} software package {\sf CBPS} is available for
implementing the proposed estimators \citep{fong:etal:18}.  In Section \ref{sec_empirical}, we conduct
simulation studies to evaluate the performance of the proposed
methodology and show that the oCBPS methodology indeed performs
better than the standard CBPS methodology in a variety of
settings. Finally, we conduct an empirical study using a canonical
application in labor economics.  We show that the oCBPS method is
able to yield estimates closer to the experimental benchmark when
compared to the standard CBPS method. 

\section{CBPS under Locally Misspecified Propensity Score Models}
\label{sec:model.misspecification}

Our theoretical investigation starts by examining the consequences of
model misspecification for the CBPS-based IPTW estimator.  While
researchers can avoid gross model misspecification through careful
model fitting, in practice it is often difficult to nail down the
exact specification.  The prominent simulation study of
\cite{MR2420462}, for example, is designed to illustrate this
phenomenon.  We therefore consider the consequences of local
misspecification of propensity score model in the general framework of
\cite{copa:eguc:05}.  In particular, we assume that the true
propensity score $\pi(\bX_i)$ is related to the working model
$\pi_{\bbeta}(\bX_i)$ through the exponential tilt for some
$\bbeta^*$,
\begin{eqnarray}
  \pi(\bX_i)&=& \pi_{\bbeta^*}(\bX_i)\exp(\xi\ u(\bX_i;\bbeta^*)),
\label{eq:misp}
\end{eqnarray}
where $u(\bX_i;\bbeta^*)$ is a function determining the direction of
misspecification and $\xi\in\RR$ represents the magnitude of
misspecification. We assume $\xi=o(1)$ as $n\rightarrow\infty$ so that
the true propensity score $\pi(\bX_i)$ is in a local neighborhood of
the working model $\pi_{\bbeta^*}(\bX_i)$.  Intuitively, since
$\pi(\bX_i)\approx \pi_{\bbeta^*}(\bX_i)$ holds, we can interpret
$\bbeta^*$ as the approximate true value of $\bbeta$. The main advantage of this exponential tilt approach is that $\pi(\bX)$ is always nonnegative, while it does not guarantee $\pi(\bX)\leq 1$. However, with $\xi=o(1)$ and Assumption \ref{reg_mis} (i.e., $|u(\bX; \bbeta^*)|\leq C$ almost surely for some constant $C>0$), we can show that $\pi(\bX)\leq 1$ holds with probability tending to 1. Finally, we note that under suitable regularity conditions model (\ref{eq:misp}) can be approximated by $\pi(\bX)=\pi_{\bbeta^*}(\bX)+\xi \bar u(\bX;\bbeta^*)+O_p(\xi^2)$, for some $\bar u(\bX;\bbeta^*)$. This provides an asymptotically equivalent specification of the locally missepecified model. To keep our presentation focused, in this section we assume model (\ref{eq:misp}) holds.

In the following, we will establish the asymptotic normality of the
CBPS-based IPTW estimator in \eqref{eq:HT} under this local model
misspecification framework.

To derive the asymptotic bias and variance, let us define some necessary quantities,
\begin{eqnarray}
B&=& \left\{\EE\left[\frac{u(\bX_i;\bbeta^*)\{K(\bX_i)+L(\bX_i)(1-\pi_{\bbeta^*}(\bX_i))\}}{1-\pi_{\bbeta^*}(\bX_i)}\right]\right. \nonumber \\
&& \hspace{.5in} \left. +\bH_y^*(\bH_{\fb}^{*\top}\Wb^*\bH_{\fb}^*)^{-1}\bH_{\fb}^{*\top}\Wb^*\EE\left(\frac{u(\bX_i;\bbeta^*)\fb(\bX_i)}{1-\pi_{\bbeta^*}(\bX_i)}\right)\right\},
\label{eq:biasmis}
\end{eqnarray}
where $K(\bX_i)$ and $L(\bX_i)$ are defined in (\ref{eq:CEFy}), $\Wb^*$ is the limiting value of $\widehat\Wb$ in \eqref{eq:gmm}, and
\begin{eqnarray*}
\bH_y^*  & = &  -\EE\left(\frac{K(\bX_i)+(1-\pi_{\bbeta^*}(\bX_i))L(\bX_i)}{\pi_{\bbeta^*}(\bX_i) (1-\pi_{\bbeta^*}(\bX_i))}\cdot \frac{\partial \pi_{\bbeta^*}(\bX_i)}{\partial\bbeta}\right), \\
\bH_{\fb}^* & = & -\EE \left(\frac{\fb(\bX_i)}{\pi_{\bbeta^*}(\bX_i)(1-\pi_{\bbeta^*}(\bX_i))}\left(\frac{\partial \pi_{\bbeta^*}(\bX_i)}{\partial \bbeta}\right)^{\top}\right).
\end{eqnarray*}
Furthermore, denote $\mu_{\bbeta^*}(T_i, Y_i, \bX_i)= \frac{T_iY_i}{\pi_{\bbeta^*}(\bX_i)}-\frac{(1-T_i)Y_i}{1-\pi_{\bbeta^*}(\bX_i)}$, 
\begin{equation}\label{eqHSigma}
\bar \bH^* = (1, \bH_y^{*\top})~~\textrm{and}~~\bSigma = \left(\begin{array}{clcr} \Sigma_{\mu} & \bSigma_{\mu\bbeta}^\top \\ \bSigma_{\mu \bbeta} & \bSigma_{\bbeta}\end{array}\right),
\end{equation}
where
\begin{align*}
& \Sigma_{\mu}  =   \Var\bigl(\mu_{\bbeta^*}(T_i, Y_i, \bX_i) \bigr) \ = \ \mathbb{E}\biggl(\frac{Y_i(1)^2}{\pi_{\bbeta^*}(\bX_i)} + \frac{Y_i(0)^2}{1-\pi_{\bbeta}^*(\bX_i)}-(\mathbb{E}(Y_i(1))-\mathbb{E}(Y_i(0)))^2\biggr), \\
& \bSigma_{\bbeta}  =  (\bH_{\fb}^{*\top}\Wb^*\bH_{\fb}^*)^{-1}\bH_{\fb}^{*\top}\Wb^*
\Var(\bg_{\bbeta^*}(T_i, \bX_i)) \Wb^* \bH_{\fb}^{*} (\bH_{\fb}^{*\top}\Wb^*\bH_{\fb}^*)^{-1}, \\
& \bSigma_{\mu \bbeta}   =
-(\bH_{\fb}^{*\top}\Wb^*\bH_{\fb}^*)^{-1}\bH_{\fb}^{*\top}\Wb^*
\Cov(\mu_{\bbeta^*}(T_i, Y_i,\bX_i),\bg_{\bbeta^*}(T_i, \bX_i)),
\end{align*}
in which $\bg_{\bbeta^*}(T_i, \bX_i)$ is defined in \eqref{eq:cvb}.
Under the model in equation \eqref{eq:CEFy}, we have
\begin{align*}
 & \Var(\bg_{\bbeta^*}(T_i, \bX_i))  =  \mathbb{E}\left(\frac{\fb(\bX_i)\fb(\bX_i)^\top}{\pi_{\bbeta^*}(\bX_i)(1-\pi_{\bbeta^*}(\bX_i))}\right), \\
&  \Cov(\mu_{\bbeta^*}(T_i, Y_i, \bX_i),\bg_{\bbeta^*}(T_i, \bX_i))  =  \mathbb{E}\left[\frac{\{K(\bX_i) + (1-\pi_{\bbeta^*}(\bX_i)) L(\bX_i)\}\fb(\bX_i)}{\pi_{\bbeta^*}(\bX_i) (1-\pi_{\bbeta^*}(\bX_i))}\right].
\end{align*}
The following theorem establishes the asymptotic normality of the
CBPS-based IPTW estimator under the local misspecification of the
propensity score model.
\begin{theorem}[\em Asymptotic Distribution under Local
  Misspecification of the Propensity Score Model]
	\label{th:asymnormal_mis} 
	If the propensity score model is locally misspecified as
in \eqref{eq:misp} with $\xi=n^{-1/2}$ and Assumption \ref{reg_mis}
	in Appendix \ref{appendix::A1} holds, the estimator $\hat\mu_{\hat\bbeta}$ in \eqref{eq:HT},
	where $\hat\bbeta$ is obtained by GMM \eqref{eq:gmm}, has the following
	asymptotic distribution 
	\begin{eqnarray}
		\sqrt{n}(\hat \mu_{\hat\bbeta} - \mu) & \stackrel{d}{\longrightarrow} &{N}(B,  \  \bar \bH^{*\top} \bSigma \bar \bH^*),
	\end{eqnarray}
	where $B$ is the asymptotic bias given in equation \eqref{eq:biasmis} and the asymptotic variance $\bar \bH^{*\top} \bSigma \bar \bH^*$ is obtained from (\ref{eqHSigma}). 
\end{theorem}

The theorem shows that the first order asymptotic bias of
$\hat{\mu}_{\hat\bbeta}$ is given by $B$ under local model
misspecification.  In particular, this bias term implicitly depends on
the covariate balancing function $\fb(\cdot)$.  Thus, we consider how
to choose $\fb(\cdot)$ such that the first order bias $|B|$ is
minimized.  While at the first glance the expression of $B$ appears to
be mathematically intractable, the next corollary shows that any
$\fb(\bX)$ satisfying (\ref{eqoptimal}) can eliminate the first order
bias, $B=0$.

\begin{corollary}
\label{th:localbias}
Suppose that the covariate balancing function $\fb(\bX)$ satisfies the following condition: there exits some  $\balpha \in \RR^m$ such that
\begin{equation}\label{eqoptimal}
\balpha^\top \fb(\bX_i) = \pi_{\bbeta^*}(\bX_i) \EE(Y_i(0) \mid
\bX_i) + (1-\pi_{\bbeta^*}(\bX_i)) \EE(Y_i(1) \mid \bX_i).
\end{equation}
In addition, assume that the dimension of $\fb(\bX_i)$ is equal to the
number of parameters, i.e., $m=q$.  Then, under the conditions in
Theorem~\ref{th:asymnormal_mis}, the asymptotic bias of the IPTW
estimator $\hat\mu_{\hat \bbeta}$ is 0, i.e., $B=0$.
\end{corollary}

Intuitively, the above result can be viewed as a ``local" version of robustness of IPTW with
respect to the misspecification of the propensity score model. The form of $\fb(\bX_i)$ in (\ref{eqoptimal}) implies that when
balancing covariates, for any given unit we should give a greater
weight to the determinants of the mean potential outcome that is less
likely to be realized.  For example, if a unit is less likely to be
treated, then it is more important to balance the covariates that
influence the mean potential outcome under the treatment condition.
In the following, we focus on the asymptotic variance of
$\hat \mu_{\hat\bbeta}$ in
Theorem~\ref{th:asymnormal_mis}. Interestingly, we can show that the
same choice of $\fb(\bX_i)$ in (\ref{eqoptimal}) minimizes the
asymptotic variance.

\begin{corollary}
\label{th:efficiency}
Under the same conditions in Corollary~\ref{th:localbias}, the
asymptotic variance of $\hat\mu_{\hat \bbeta}$ is minimized by any
covariate balancing function $\fb(\bX_i)$ which satisfies
(\ref{eqoptimal}). In this case, the CBPS-based IPTW estimator
$\hat\mu_{\hat \bbeta}$ attains the semiparametric asymptotic variance
bound in Theorem 1 of \cite{hahn1998role}, i.e.,
\begin{equation}\label{eqinforbound1}
V_{\textrm{opt}}=\EE\left[\frac{\Var(Y_i(1)\mid\bX_i)}{\pi(\bX_i)}+\frac{\Var(Y_i(0)\mid\bX_i)}{1-\pi(\bX_i)}+\{L(\bX_i)-\mu\}^2\right].
\end{equation}
\end{corollary}

Based on Theorem~\ref{th:asymnormal_mis}, we can define the asymptotic
mean squared error (AMSE) of $\hat\mu_{\hat \bbeta}$ as
$AMSE=B^2+ \bar \bH^{*\top} \bSigma \bar
\bH^*$. Corollaries~\ref{th:localbias}~and~\ref{th:efficiency}
together imply that $\hat\mu_{\hat \bbeta}$ with $\fb(\bX)$ satisfying
(\ref{eqoptimal}) attains the minimum AMSE over all possible covariate
balancing estimators. Thus, we refer to (\ref{eqoptimal}) as the
optimality condition for the covariate balancing function. We note
that there may exist many choices of $\fb(\bX)$ which satisfy
(\ref{eqoptimal}). For instance, we can choose
$f_1(\bX)=\pi_{\bbeta^*}(\bX_i) \EE(Y_i(0) \mid \bX_i) +
(1-\pi_{\bbeta^*}(\bX_i)) \EE(Y_i(1) \mid \bX_i)$ and $f_2, ..., f_m$
in an arbitrary way, as long as the estimating equation
$\bar \bg_{\bbeta}(\bT,\bX) \ = \ 0$ is not degenerate. In this case,
to implement $f_1(\bX)$, we need to further estimate $\bbeta^*$ by
some initial estimator, e.g., the maximum likelihood estimator, and
estimate the conditional mean $\EE(Y_i(0) \mid \bX_i)$ and
$\EE(Y_i(1) \mid \bX_i)$ by some parametric/nonparametric
models. While Corollaries~\ref{th:localbias}~and~\ref{th:efficiency}
hold with this choice of $\fb(\bX)$, the empirical performance of the
resulting estimator $\hat\mu_{\hat \bbeta}$ is often unstable due to
the estimation error of the initial estimators. To overcome this
problem, we will next construct the optimal CBPS estimator that does
not require any initial estimator.



\section{The Optimal CBPS Methodology}
\label{sec:proposed.methodology}

Recall that the optimal covariate balancing function $\fb(\bX)$ is
given by (\ref{eqoptimal}). Plugging $\fb(\bX)$ into the estimating
function $\bg_{\bbeta}(T_i,\bX_i)$ in (\ref{eq:cvb}), we obtain that
\begin{align}
\balpha^\top \bg_{\bbeta^*}(T_i,\bX_i)&= \left(\frac{T_i}{\pi_{\bbeta^*}(\bX_i)}-\frac{1-T_i}{1-\pi_{\bbeta^*}(\bX_i)} \right)\Big[\pi_{\bbeta^*}(\bX_i) K(\bX_i) + (1-\pi_{\bbeta^*}(\bX_i)) (K(\bX_i)+L(\bX_i))\Big]\nonumber\\
&=\left(\frac{T_i}{\pi_{\bbeta^*}(\bX_i)}-\frac{1-T_i}{1-\pi_{\bbeta^*}(\bX_i)} \right)K(\bX_i) +\left(\frac{T_i}{\pi_{\bbeta^*}(\bX_i)}-1\right)L(\bX_i).\label{eqbias}
\end{align}
In other words, the optimality condition (\ref{eqoptimal}) holds if
and only if some linear combination of estimating function
$ \bg_{\bbeta}(T_i,\bX_i)$ satisfies (\ref{eqbias}). Motivated by this
observation, we construct the following set of estimating functions,
\begin{eqnarray}
\bar\bg_{\bbeta}(\bT, \bX)=\left(\begin{array}{c}
                            \bar\bg_{1\bbeta}(\bT, \bX) \\
                            \bar\bg_{2\bbeta}(\bT, \bX)\end{array}\right), \label{eq:cvb2}
\end{eqnarray}
where $\bar\bg_{1\bbeta}(\bT, \bX)={n}^{-1}\sum_{i=1}^n\bg_{1\bbeta}(T_i, \bX_i)$ and $\bar\bg_{2\bbeta}(\bT, \bX)={n}^{-1}\sum_{i=1}^n\bg_{2\bbeta}(T_i, \bX_i)$ with
\begin{eqnarray}\label{eq_hh}
\bg_{1\bbeta}(T_i, \bX_i) & = & \left(\frac{T_i}{\pi_{\bbeta}(\bX_i)}-\frac{1-T_i}{1-\pi_{\bbeta}(\bX_i)}\right)\bh_1(\bX_i),
\bg_{2\bbeta}(T_i, \bX_i) \ = \ \left(\frac{T_i}{\pi_{\bbeta}(\bX_i)}-1\right)\bh_2(\bX_i),~~
\end{eqnarray}
for some pre-specified functions
$\bh_1(\cdot): \RR^{d}\rightarrow \RR^{m_1}$ and
$\bh_2(\cdot): \RR^{d}\rightarrow \RR^{m_2}$ with $m_1+m_2=m$.  It is
easy to see that if the functions $K(\cdot)$ and $L(\cdot)$ lie in the
linear space spanned by the functions $\bh_1(\cdot)$ and
$\bh_2(\cdot)$ respectively, then there exists a vector
$\balpha\in\RR^m$ such that (\ref{eqbias}) holds for
$(\bg_{1\bbeta}(T_i, \bX_i),\bg_{2\bbeta}(T_i, \bX_i))$, further
implying that the optimality condition (\ref{eqoptimal}) is met. 

As
discussed in Section \ref{sec:model.misspecification}, the choice of
the optimal covariate balancing function is not unique. Unlike the one mentioned after Corollary \ref{th:efficiency}, the
estimating function in (\ref{eq:cvb2}) does not require any initial
estimators for $\bbeta$ or the conditional mean models, and is more
convenient for implementation.  Given the estimating functions in
\eqref{eq:cvb2}, we can estimate $\bbeta$ by the GMM estimator
$\hat{\bbeta}$ in \eqref{eq:gmm}. We call this method as the optimal
CBPS method ({oCBPS}). Similarly, the ATE is estimated by the IPTW
estimator $\hat{\mu}_{\hat\bbeta}$ in \eqref{eq:HT}.  The
implementation of the proposed {oCBPS} method (e.g., the choice of
$\bh_1(\cdot)$ and $\bh_2(\cdot)$) will be discussed in later
sections.
  
It is worthwhile to note that $\bar\bg_{\bbeta}(\bT, \bX)$ has the
following interpretation.  The first set of functions
$\bar\bg_{1\bbeta}(\bT, \bX)$ is the same as the existing covariate
balancing moment function in \eqref{eq:cvb}, which balances the
covariates $\bh_1(\bX_i)$ between the treatment and control
groups. However, unlike the original CBPS method, we introduce another
set of functions $\bar\bg_{2\bbeta}(\bT, \bX)$ which matches the
weighted covariates $\bh_2(\bX_i)$ in the treatment group to the
unweighted covariates $\bh_2(\bX_i)$ in the control group, because
$\bar\bg_{2\bbeta}(\bT, \bX)=0$ can be rewritten as
\begin{equation*}
\sum_{T_i=1}\frac{1-\pi_{\bbeta}(\bX_i)}{\pi_{\bbeta}(\bX_i)}\bh_2(\bX_i)\
= \ \sum_{T_i=0}\bh_2(\bX_i).
\end{equation*}
As seen in the derivation of \eqref{eqbias}, the auxiliary estimating
function $\bar\bg_{2\bbeta}(\bT, \bX)$ is required in order to satisfy
the optimality condition.

\subsection{Theoretical Properties}

We now derive the theoretical properties of the IPTW estimator \eqref{eq:HT} based on the proposed {oCBPS} method.  In
particular, we will show that the proposed estimator is doubly robust
and locally efficient.  The following set of assumptions
are imposed for the establishment of double robustness.

\begin{assumption}\label{ass1} The following regularity
  conditions are assumed.
  \begin{enumerate}
  \item There exists a positive definite matrix $\Wb^*$ such that
    $\hat{\Wb}\stackrel{p}{\longrightarrow} \Wb^*$.
  \item
  For any $\bh_1(\cdot)$ and $\bh_2(\cdot)$ in (\ref{eq_hh}), the minimizer
    $\bbeta^o=\argmin_{\bbeta\in\Theta} \EE(\bar\bg_{\bbeta}(\bT,\bX))^\top
    \Wb^*\EE(\bar\bg_{\bbeta}(\bT,\bX))$ is unique. 
  \item  $\bbeta^o\in\textrm{int}(\Theta)$, where $\Theta$ is a
    compact set.
  \item $\pi_{\bbeta}(\bX)$ is continuous in $\bbeta$.
  \item There exists a constant $0<c_0<1/2$ such that with
    probability tending to one, $c_0\leq \pi_{\bbeta}(\bX)\leq 1-c_0$,
    for any $\bbeta\in \textrm{int}(\Theta)$.
  \item     $\EE|Y(1)|^2<\infty$ and $\EE|Y(0)|^2<\infty$.
  \item For any $\bh_1(\cdot)$ and $\bh_2(\cdot)$ in (\ref{eq_hh}) and $\Wb^*$ in part 1, $\Gb^*:=\EE(\partial \bg({\bbeta^o})/\partial\bbeta)$
    exists where $\bg(\bbeta) = (\bg_{1\bbeta}(\bT,\bX)^\top, \bg_{2\bbeta}(\bT, \bX)^\top)^\top$ and there is a $q$-dimensional function $C(\bX)$ and a
    small constant $r>0$ such that
    $\sup_{\bbeta\in\mathbb{B}_r(\bbeta^o)}|\partial\pi_{\bbeta}(\bX)/\partial\beta_k|\leq
    C_k(\bX)$ for $1\leq k\leq q$,
    and $\EE(|h_{1j}(\bX)|C_k(\bX))<\infty$ for $1\leq j\leq m_1$, $1\leq k\leq q$ and
    $\EE(|h_{2j}(\bX)|C_k(\bX))<\infty$ for $1\leq j\leq m_2$, $1\leq k\leq q$, where
    $\mathbb{B}_r(\bbeta^o)$ is a ball in $\RR^q$ with radius $r$ and
    center $\bbeta^o$.
  \end{enumerate}
\end{assumption}
Conditions~1--4 of Assumption~\ref{ass1} are the standard conditions
for consistency of the GMM estimator \citep{newey1994large}. Note that
we allow the propensity score model to be misspecified, so that we use
the notation $\bbeta^o$ in Condition 2 to distinguish it from
$\bbeta^*$ used in the previous section.  Condition~5 is the
positivity assumption commonly used in the causal inference literature
\citep{robins1994estimation,robins1995analysis}. Conditions~6~and~7
are technical conditions that enable us to apply the dominated
convergence theorem. Note that,
$\sup_{\bbeta\in\mathbb{B}_r(\bbeta^o)}|\partial\pi_{\bbeta}(\bX)/\partial\beta_k|\leq
C_k(\bX)$ in Condition~7 is a local condition in the sense that it
only requires the existence of an envelop function $C_k(\bX)$ around a
small neighborhood of $\bbeta^o$.

We now establish the double robustness of the proposed estimator under
Assumption~\ref{ass1}. 
\begin{theorem}[\em Double Robustness]\label{thmdr}
  Under Assumption~\ref{ass1}, the proposed oCBPS-based IPTW estimator
  $\hat{\mu}_{\hat\beta}$ is doubly robust.  That is,
  $\hat{\mu}_{\hat\beta}\stackrel{p}{\longrightarrow} \mu$ if at least
  one of the following two conditions holds:
\begin{enumerate}
\item The propensity score model is correctly specified, i.e.,
  $\PP(T_i=1\mid \bX_i)=\pi_{\bbeta^o}(\bX_i)$;
\item The functions $\bh_1(\cdot)$ and $\bh_2(\cdot)$ in (\ref{eq_hh}) and $\Wb^*$ in Assumption \ref{ass1} satisfy the following condition. There exist some vectors
$\balpha_1, \balpha_2\in\RR^{q}$ such that $K(\bX_i)=\balpha_1^\top\Mb_1\bh_1(\bX_i)$ and
$L(\bX_i)=\balpha_2^\top\Mb_2\bh_2(\bX_i)$, where $\Mb_1\in\RR^{q\times m_1}$  and $\Mb_2\in\RR^{q\times m_2}$ are the partitions of  $\Gb^{*\top}\Wb^*=(\Mb_1,\Mb_2)$.
\end{enumerate}	
\end{theorem}



Next, we establish the asymptotic normality of the proposed estimator
if either the propensity score model (Condition 1 in Theorem~\ref{thmdr}) or the outcome model is correctly
specified (Condition 2 in Theorem~\ref{thmdr}) .  For this result, we require an additional set of
regularity conditions.
\begin{assumption}\label{ass2} The following regularity conditions are
  assumed.
  \begin{enumerate}
  \item For any $\bh_1(\cdot)$ and $\bh_2(\cdot)$ in (\ref{eq_hh}) and $\Wb^*$ in Assumption \ref{ass1}, $\Gb^{*\top}\Wb^*\Gb^*$ and
    $\bOmega=\EE(\bg_{\bbeta^o}(T_i,\bX_i)\bg_{\bbeta^o}(T_i,\bX_i)^\top)$
    are nonsingular.
\item The function $C_k(\bX)$ defined in Condition~7 of
  Assumption~\ref{ass1} satisfies $\EE(|Y(0)|C_k(\bX))<\infty$ and $\EE(|Y(1)|C_k(\bX))<\infty$ for $1\leq k\leq q$.
\end{enumerate}	
\end{assumption}
Condition~1 of Assumption~\ref{ass2} ensures the non-singularity of
the asymptotic variance matrix and Condition~2 is a mild technical
condition required for the dominated convergence
theorem.

\begin{theorem}[\em Asymptotic Normality]\label{thmAN}
  Suppose that Assumptions~\ref{ass1}~and~\ref{ass2} hold.
\begin{enumerate}
\item If Condition~1 of Theorem~\ref{thmdr} holds, then the proposed
   oCBPS-based IPTW estimator $\hat{\mu}_{\hat\bbeta}$  has the following
  asymptotic distribution:
\begin{equation}\label{eqdeltaAN1}
  \sqrt{n}(\hat{\mu}_{\hat\bbeta}-\mu) \ \stackrel{d}{\longrightarrow}  N\left(0,\ \bar\Hb^{*\top}\bSigma\bar\Hb^*\right),
\end{equation}
where $\bar\Hb^*=(\mathbf{1},\Hb^{*\top})^\top$,
$\bSigma_{\bbeta}=(\Gb^{*\top}\Wb^*\Gb^*)^{-1}\Gb^{*\top}\Wb^*\bOmega\Wb^*\Gb^*(\Gb^{*\top}\Wb^*\Gb^*)^{-1}$
and
\begin{eqnarray}
\Hb^* & = & -\EE\left(\frac{K(\bX_i)+(1-\pi_{\bbeta^o}(\bX_i))L(\bX_i)}{\pi_{\bbeta^o}(\bX_i) (1-\pi_{\bbeta^o}(\bX_i))}\cdot \frac{\partial \pi_{\bbeta^o}(\bX_i)}{\partial\bbeta}\right),\notag \\
\bSigma & = & \left( \begin{array}{cc}
\Sigma_{\mu} & \bSigma^\top_{\mu\bbeta} \label{eqsigmadelta}  \\
\bSigma_{\mu\bbeta} & \bSigma_{\bbeta} \end{array} \right), ~~\textrm{with}~~\Sigma_{\mu}=\EE\left(\frac{Y_i^2(1)}{\pi_{\bbeta^o}(\bX_i)}+\frac{Y_i^2(0)}{1-\pi_{\bbeta^o}(\bX_i)}\right)-\mu^{2}.
\end{eqnarray}
In addition, $\bSigma_{\mu\bbeta}$ is given by
\begin{align*}
\bSigma_{\mu\bbeta} & =  -(\Gb^{*\top}\Wb^*\Gb^*)^{-1}\Gb^{*\top}\Wb^*\Big\{\EE\Big(\frac{K(\bX_i)+(1-\pi_i^o)L(\bX_i)}{(1-\pi_i^o)\pi_i^o}\bh^\top_{1}(\bX_i)\Big), \\
&~~~~~~~~~~~~~~\EE\Big(\frac{K(\bX_i)+(1-\pi_i^o)L(\bX_i)}{\pi_i^o}\bh^\top_{2}(\bX_i)\Big)\Big\}^{\top}.
\end{align*}
\item If Condition~2 of Theorem~\ref{thmdr} holds, then the proposed
oCBPS-based IPTW estimator $\hat{\mu}_{\hat\bbeta}$  has the following
asymptotic distribution:
\begin{equation}\label{eqdeltaAN2}
  \sqrt{n}(\hat{\mu}_{\hat\bbeta}-\mu) \ \stackrel{d}{\longrightarrow} N\left(0,\ \tilde\Hb^{*\top}\tilde\bSigma\tilde\Hb^*\right),
\end{equation}
where  $\tilde\Hb^*=(1,\check\Hb^{*\top})^\top$, 
\begin{align*}
  & \check\Hb^*  =  -\EE\left[\left\{\frac{\pi(\bX_i) (K(\bX_i)+L(\bX_i))}{\pi_{\bbeta^o}(\bX_i)^{2}}+\frac{(1-\pi(\bX_i))K(\bX_i)}{(1-\pi_{\bbeta^o}(\bX_i))^{2}}\right\}\frac{\partial {\pi}_{\bbeta^o}(\bX_i)}{\partial\bbeta^o}\right], \\
& \tilde\bSigma  =  \left( \begin{array}{cc}
                     \tilde\Sigma_{\mu} & \tilde\bSigma^\top_{\mu\bbeta}  \\
                     \tilde\bSigma_{\mu\bbeta} &
                                           \bSigma_{\bbeta} \end{array}
                                           \right)
                                           ~~\textrm{with}~~\tilde\Sigma_{\mu} \ = \ \EE\left(\frac{\pi(\bX_i)Y_i^2(1)}{\pi_{\bbeta^o}(\bX_i)^{2}}+\frac{(1-\pi(\bX_i))Y_i^2(0)}{(1-\pi_{\bbeta^o}(\bX_i))^2}\right)-\mu^{2}.
\end{align*}
In addition, $\tilde\bSigma_{\mu\bbeta}$ is given by
\begin{equation*}
\tilde\bSigma_{\mu\bbeta} \ = \ -(\Gb^{*\top}\Wb^*\Gb^*)^{-1}\Gb^{*\top}\Wb^* \bS,
\end{equation*}
where $\bS = (\bS_1^\top,\bS_2^\top)^\top$ and
\begin{eqnarray*}
  \bS_1 & = & \EE\left[\left\{\frac{\pi(\bX_i) (K(\bX_i)+L(\bX_i)-\pi_{\bbeta^o}(\bX_i)\mu)}{\pi_{\bbeta^o}(\bX_i)^{2}}\right.\right. \\
        & & \hspace{1.75in} \left.\left. +\frac{(1-\pi(\bX_i))(K(\bX_i)+(1-\pi_{\bbeta^o}(\bX_i))\mu)}{(1-\pi_{\bbeta^o}(\bX_i))^2}\right\}\bh_{1}(\bX_i)\right],\\
  \bS_2 & = & \EE\left[\left\{\frac{\pi(\bX_i)[(K(\bX_i)+L(\bX_i))(1-\pi_{\bbeta^o}(\bX_i))-\pi_{\bbeta^o}(\bX_i)\mu]}{\pi_{\bbeta^o}(\bX_i)^{2}}\right.\right. \\
        & & \hspace{1.75in} \left.\left. +\frac{(1-\pi(\bX_i))K(\bX_i)+(1-\pi_{\bbeta^o}(\bX_i))\mu}{1-\pi_{\bbeta^o}(\bX_i)}\right\}\bh_{2}(\bX_i)\right].
\end{eqnarray*}
\item If both Conditions~1~and~2 of Theorem~\ref{thmdr} hold and
  $\Wb^*=\bOmega^{-1}$, then the proposed
  oCBPS-based IPTW estimator $\hat{\mu}_{\hat\bbeta}$  has the following
  asymptotic distribution:
$$\sqrt{n}(\hat{\mu}_{\hat\bbeta}-\mu) \ \stackrel{d}{\longrightarrow} N(0,V),$$
where
\begin{equation}\label{eqdeltavar}
	V=\Sigma_{\mu}-(\balpha_1^\top\Mb_1,\balpha_2^\top\Mb_2)\Gb^*(\Gb^{*\top}\bOmega^{-1}\Gb^{*})^{-1}\Gb^{*\top} \left( \begin{array}{c}
		\Mb_1^\top\balpha_1  \\
		\Mb_2^\top\balpha_2  \end{array} \right)
\end{equation}
and $\Sigma_{\mu}$ is defined in \eqref{eqsigmadelta}.
\end{enumerate}
\end{theorem}

The asymptotic variance $V$ in \eqref{eqdeltavar} contains two terms.
The first term $\Sigma_{\mu}$ represents the variance of each summand
in the estimator defined in equation \eqref{eq:HT} with $\hat{\bbeta}$
replaced by $\bbeta^o$.  The second term can be
interpreted as the effect of estimating $\bbeta$ via covariate
balance conditions. Since this second term is nonnegative, the
proposed estimator is more efficient than the standard IPTW estimator
with the true propensity score model, i.e., $V\leq \Sigma_{\mu}$.  In
particular, \cite{henmi2004paradox} offered a theoretical analysis of
such efficiency gain due to the estimation of nuisance parameters
under a general estimating equation framework.

Since the choice of $\bh_1(\cdot)$ and $\bh_2(\cdot)$ can be arbitrary, it
might be tempting to incorporate more covariate balancing
conditions into $\bh_1(\cdot)$ and $\bh_2(\cdot)$.  However, the following corollary shows that under Conditions~1 and 2 of
  Theorem~\ref{thmdr} one cannot improve
the efficiency of the proposed estimator by increasing the number of functions
$\bh_1(\cdot)$ and $\bh_2(\cdot)$ or equivalently, the dimensionality
of covariate balance conditions, i.e., $\bar\bg_{1\bbeta}(\bT, \bX)$
and $\bar\bg_{2\bbeta}(\bT,\bX)$.

\begin{corollary}\label{coraddmoment}
  Define $\bar\bh_1(\bX)=(\bh^\top_1(\bX),\ba^\top_1(\bX))^\top$ and
  $\bar\bh_2(\bX)=(\bh^\top_2(\bX),\ba^\top_2(\bX))^\top$, where
  $\ba_1(\cdot)$ and $\ba_2(\cdot)$ are some additional covariate
  balancing functions. Similarly, let $\bar\bg_1(\bX)$ and
  $\bar\bg_2(\bX)$ denote the corresponding estimating equations
  defined by $\bar\bh_1(\bX)$ and $\bar\bh_2(\bX)$.  The resulting
  oCBPS-based IPTW estimator is denoted by $\bar{\mu}_{\hat\bbeta}$
  where $\hat\bbeta$ is in \eqref{eq:gmm} and its
  asymptotic variance is denoted by $\bar V$.  Under Conditions~1 and 2 of
  Theorem~\ref{thmdr}, we have $V\leq \bar V$, where $V$ is defined in \eqref{eqdeltavar}.
\end{corollary}

The above corollary shows a potential trade-off between robustness and efficiency
when choosing $\bh_1(\cdot)$ and $\bh_2(\cdot)$. Recall that Condition 2 of
  Theorem~\ref{thmdr} implies $K(\bX_i)=\balpha_1^\top\Mb_1\bh_1(\bX_i)$ and
$L(\bX_i)=\balpha_2^\top\Mb_2\bh_2(\bX_i)$. Therefore, we can make the
proposed estimator more robust by incorporating more
basis functions into $\bh_1(\cdot)$ and $\bh_2(\cdot)$, such that this  condition is more likely to hold. However,
Corollary~\ref{coraddmoment} shows that doing so may inflate the variance  of the proposed estimator.

In the following, we focus on the efficiency of the estimator.
Using the notations in this section, we can rewrite the semiparametric asymptotic variance bound $V_{\textrm{opt}}$ in (\ref{eqinforbound1}) as
\begin{equation}\label{eqinforbound}
V_{\textrm{opt}} \ = \ \Sigma_{\mu}-(\balpha_1^\top\Mb_1,\balpha_2^\top\Mb_2)\bOmega\left( \begin{array}{c}
\Mb_1^\top\balpha_1  \\
\Mb_2^\top\balpha_2  \end{array} \right).
\end{equation}
Comparing this expression with \eqref{eqdeltavar}, we see
that the proposed estimator is semiparametrically efficient if $\Gb^*$
is a square matrix (i.e., $m = q$) and invertible.  
This important result is summarized as the following corollary.

\begin{corollary}\label{cormq}
  Assume $m=q$ and $\Gb^*$ is invertible. Under
  Assumption~\ref{ass1}, the proposed estimator
  $\hat{\mu}_{\hat\bbeta}$  in \eqref{eq:HT} is doubly
  robust in the sense that
  $\hat{\mu}_{\hat\bbeta}\stackrel{p}{\longrightarrow} \mu$ if
  either of the following conditions holds:
\begin{enumerate}
\item The propensity score model is correctly specified. That is
  $\PP(T_i=1\mid \bX_i)=\pi_{\bbeta^o}(\bX_i)$.
\item There exist some vectors
$\balpha_1, \balpha_2\in\RR^{q}$ such that $K(\bX_i)=\balpha_1^\top\bh_1(\bX_i)$ and
$L(\bX_i)=\balpha_2^\top\bh_2(\bX_i)$.
\end{enumerate}		
In addition, under Assumption~\ref{ass2}, if both conditions hold, then the proposed estimator has the asymptotic variance given in \eqref{eqinforbound}. Thus, our estimator is a locally semiparametric efficient estimator in the sense of \cite{robins1994estimation}. 
\end{corollary}
The corollary shows that the proposed oCBPS method has two advantages
over the original CBPS method \citep{imai:ratk:14} with balancing
first and second moments of $\bX_i$ and/or the score function of the
propensity score model. First, the proposed estimator
$\hat{\mu}_{\hat\bbeta}$ is robust to model misspecification, whereas
the original CBPS estimator does not have that property. Second, the
proposed oCBPS estimator can be more efficient than the original CBPS
estimator.

Corollary \ref{cormq} also implies that the asymptotic variance of
$\hat{\mu}_{\hat\bbeta}$ is identical to the semiparametric variance
bound $V_{\textrm{opt}}$, even if we incorporate additional covariate
balancing functions into $\bh_1(\cdot)$ and $\bh_2(\cdot)$. Namely,
under the conditions in Corollary \ref{cormq}, we have
$V=\bar V=V_{\textrm{opt}}$ in the context of Corollary
\ref{coraddmoment}. Thus, in this setting, we can improve the
robustness of the estimator without sacrificing the efficiency by
increasing the number of functions $\bh_1(\cdot)$ and
$\bh_2(\cdot)$. Meanwhile, this also makes the propensity score model
more flexible, since we need to increase the number of parameters
$\bbeta$ to match $m=q$ as required in Corollary \ref{cormq}. This
observation further motivates us to consider a sieve estimation
approach to improve the oCBPS method, as shown in Section
\ref{secnon}.


\begin{remark}[\bf Implementation of the oCBPS method]
  Based on Corollary \ref{cormq}, $\bh_1(\cdot)$ serves as the basis
  functions for the baseline conditional mean function $K(\cdot)$,
  while $\bh_2(\cdot)$ represents the basis functions for the
  conditional average treatment effect function $L(\cdot)$.  Thus, in
  practice, researchers can choose a set of basis functions for the
  baseline conditional mean function and the conditional average
  treatment effect function when determining the specification for
  $\bh_1(\cdot)$ and $\bh_2(\cdot)$.  Once these functions are
  selected, they can over-parameterize the propensity score model by
  including some higher order terms or interactions such that $m=q$
  holds.  The resulting oCBPS-based IPTW estimator may reduce bias
  under model misspecification and attain high
  efficiency. 
\end{remark}

\begin{remark}
  We also extend the oCBPS method to the estimation of the average
  treatment effect for the treated (ATT). Given the space limitation,
  we defer the details to the supplementary material. 
\end{remark}


\subsection{Comparison with Related Estimators}

Next, we compare the proposed estimator with some related estimators
from the literature.  We begin with the following standard AIPW
estimator of \cite{robins1994estimation}, 
$$
\hat\mu^{AIPW}_{\bbeta,\balpha,\bgamma} \ = \ \frac{1}{n}\sum_{i=1}^n\left\{\frac{T_iY_i}{\pi_{\bbeta}(\bX_i)}-\frac{(1-T_i)Y_i}{1-\pi_{\bbeta}(\bX_i)}-(T_i-\pi_{\bbeta}(\bX_i))\left(\frac{K(\bX_i,\balpha)+L(\bX_i,\bgamma)}{\pi_{\bbeta}(\bX_i)}+\frac{K(\bX_i,\balpha)}{1-\pi_{\bbeta}(\bX_i)}\right)\right\},
$$
where $K(\bX_i,\balpha)$ and $L(\bX_i,\bgamma)$ are the conditional
mean models indexed by finite dimensional parameters $\balpha$ and
$\bgamma$.
Assume the linear outcome models:
$K(\bX_i,\balpha)=\balpha^T\bh_1(\bX_i)$ and
$L(\bX_i,\bgamma)=\bgamma^T\bh_2(\bX_i)$.  It is interesting to note
that our IPTW estimator $\hat\mu_{\hat\bbeta}$ in
Corollary~\ref{cormq} can be rewritten as the AIPW estimator
$\hat\mu^{AIPW}_{\hat\bbeta,\balpha,\bgamma}$ (for any $\balpha$ and
$\bgamma$), since we have,
$$
\frac{1}{n}\sum_{i=1}^n(T_i-\pi_{\hat\bbeta}(\bX_i))\left(\frac{K(\bX_i,\balpha)+L(\bX_i,\bgamma)}{\pi_{\hat\bbeta}(\bX_i)}+\frac{K(\bX_i,\balpha)}{1-\pi_{\hat\bbeta}(\bX_i)}\right)
\ = \ 0,
$$	
by the definition of the covariate balancing estimating equations in
(\ref{eq:cvb2}).

It is well known that the AIPW estimator is consistent provided that
either the propensity score model or the outcome model is correctly
specified. Since both the AIPW estimator and our estimator are doubly
robust and locally efficient, in the following we conduct a
theoretical investigation of these two estimators under the scenario
that {\it both} propensity score and outcome models are
misspecified.  Indeed, this scenario corresponds to the simulation
settings used in the influential study of \citet{MR2420462}.

To make the comparison mathematically tractable, we focus on the case
that both of these two models are locally misspecified. Similar to
Section \ref{sec:model.misspecification}, we assume that the true
treatment assignment satisfies,
$\pi(\bX_i)= \pi_{\bbeta^*}(\bX_i)\exp(\xi\ u(\bX_i;\bbeta^*))$ in
\eqref{eq:misp}, while the true regression functions $K(\bX_i)$ and
$L(\bX_i)$ in (\ref{eq:CEFy}) satisfy
\begin{equation}\label{eqmisp2}
K(\bX_i)=\balpha^{*\top}\bh_1(\bX_i)+\delta r_1(\bX_i),~~~L(\bX_i)=\bgamma^{*\top}\bh_2(\bX_i)+\delta r_2(\bX_i),
\end{equation}
where $\balpha^*$ and $\bgamma^*$ can be viewed as the approximate
true values of $\balpha$ and $\bgamma$, the functions $r_1(\bX_i)$ and
$r_2(\bX_i)$ determine the direction of misspecification, and
$\delta\in\RR$ represents the magnitude of misspecification.

Assume further that the models are locally misspecified, i.e.,
$\xi, \delta=o(1)$. Under regularity conditions similar to
Section~\ref{sec:model.misspecification}, we can show that the
proposed estimator satisfies,
\begin{align}
\hat\mu_{\hat\bbeta}-\mu&=\frac{1}{n}\sum_{i=1}^n\Big[\frac{T_i}{\pi(\bX_i)}\{Y_i(1)-K(\bX_i)-L(\bX_i)\}-\frac{1-T_i}{1-\pi(\bX_i)}\{Y_i(0)-K(\bX_i)\}+L(\bX_i)-\mu\Big]\nonumber\\
&~~~~+O_p(\xi^2\delta+\delta n^{-1/2}+\xi n^{-1/2}),\label{eqmis_cbps}
\end{align}
whereas the AIPW estimator satisfies,
\begin{align}
\hat\mu^{AIPW}_{\tilde\bbeta,\tilde\balpha,\tilde\bgamma}-\mu&=\frac{1}{n}\sum_{i=1}^n\Big[\frac{T_i}{\pi(\bX_i)}\{Y_i(1)-K(\bX_i)-L(\bX_i)\}-\frac{1-T_i}{1-\pi(\bX_i)}\{Y_i(0)-K(\bX_i)\}+L(\bX_i)-\mu\Big]\nonumber\\
&~~~~+O_p(\xi\delta+\delta n^{-1/2}+\xi n^{-1/2}),\label{eqmis_aipw}
\end{align}
where $\tilde\bbeta,\tilde\balpha$ and $\tilde\bgamma$ are the
corresponding maximum likelihood and least square estimators. The
derivation of (\ref{eqmis_cbps}) and (\ref{eqmis_aipw}) is shown in
Appendix~\ref{sec_mis_cbps_aipw}.

The leading terms in the asymptotic expansions of
$\hat\mu_{\hat\bbeta}-\mu$ and
$\hat\mu^{AIPW}_{\tilde\bbeta,\tilde\balpha,\tilde\bgamma}-\mu$ are
identical and are known as the efficient influence function for
$\mu$. However, the remainder terms in (\ref{eqmis_cbps}) and
(\ref{eqmis_aipw}) may have different order. Consider the following
two scenarios.  First, if $\xi\delta\gg n^{-1/2}$, then we have
$\hat\mu_{\hat\bbeta}-\mu=O_p(\xi^2\delta+n^{-1/2})$ and
$\hat\mu^{AIPW}_{\tilde\bbeta,\tilde\balpha,\tilde\bgamma}-\mu=O_p(\xi\delta)$. Thus, the proposed
estimator $\hat\mu_{\hat\bbeta}$ converges in probability to the ATE at a faster
rate than $\hat\mu^{AIPW}_{\tilde\bbeta,\tilde\balpha,\tilde\bgamma}$. Second, if $\xi\delta= o(n^{-1/2})$, the two estimators have the
same limiting distribution, i.e.,
$\sqrt{n}(\hat{\mu}-\mu) \stackrel{d}{\longrightarrow} N(0,
V_{\textrm{opt}})$, where $\hat\mu$ can be either
$\hat{\mu}_{\hat\bbeta}$ or
$\hat\mu^{AIPW}_{\tilde\bbeta,\tilde\balpha,\tilde\bgamma}$. However,
the rates of convergence of the Gaussian approximation determined by
the remainder terms in (\ref{eqmis_cbps}) and (\ref{eqmis_aipw}) are
different. For instance, assume that $\xi=\delta=n^{-(1/4+\epsilon)}$
for some small positive $\epsilon<1/4$. We observe that the remainder
term in (\ref{eqmis_cbps}) is of order $O_p(n^{-(3/4+\epsilon)})$ and
is smaller in magnitude than the corresponding term in
(\ref{eqmis_aipw}), which is of order $O_p(n^{-(1/2+2\epsilon)})$. As
a result, the proposed estimator converges in distribution to
$N(0, V_{\textrm{opt}})$ at a faster rate than the AIPW estimator. The
above analysis justifies the theoretical advantage of the proposed
oCBPS estimator over the standard AIPW estimator.

Furthermore, the proposed estimator is related to the class of
bias-reduced doubly robust estimators \citep{vermeulen2015bias}, see also \cite{robins2007comment}. To see
this, we consider the derivative of
$\hat\mu^{AIPW}_{\bbeta,\balpha,\bgamma}$ with respect to the nuisance
parameters $\balpha,\bgamma$. In particular, under the linear
outcome models, it is easily shown that
$\partial \hat\mu^{AIPW}_{\bbeta,\balpha,\bgamma}/\partial
\balpha=\bar\bg_{1\bbeta}(\bT, \bX)$ and
$\partial \hat\mu^{AIPW}_{\bbeta,\balpha,\bgamma}/\partial
\bgamma=\bar\bg_{2\bbeta}(\bT, \bX)$, where
$\bar\bg_{1\bbeta}(\bT, \bX)$ and $\bar\bg_{2\bbeta}(\bT, \bX)$ are
our covariate balancing functions in (\ref{eq:cvb2}). This provides an
alternative justification for the proposed method: the oCBPS estimator
$\hat\bbeta$, which satisfies $\bar\bg_{2\bbeta}(\bT, \bX)=0$ and
$\bar\bg_{1\bbeta}(\bT, \bX)=0$, removes the local effect of the
estimated nuisance parameters, i.e.,
$\partial \hat\mu^{AIPW}_{\hat\bbeta,\balpha,\bgamma}/\partial
\balpha=0$ and
$\partial \hat\mu^{AIPW}_{\hat\bbeta,\balpha,\bgamma}/\partial
\bgamma=0$. This property would not hold if we replace $\hat\bbeta$ by
the maximum likelihood estimator or other convenient estimators of
$\bbeta$. \cite{vermeulen2015bias} defined the class of bias-reduced
doubly robust estimator as
$\hat\mu^{AIPW}_{\bar\bbeta,\bar\balpha,\bar\bgamma}$, where
$(\bar\bbeta,\bar\balpha,\bar\bgamma)$ are the estimators
corresponding to the estimating equations
$\partial \hat\mu^{AIPW}_{\bbeta,\balpha,\bgamma}/\partial \balpha=0,
\partial \hat\mu^{AIPW}_{\bbeta,\balpha,\bgamma}/\partial \bgamma=0,
\partial \hat\mu^{AIPW}_{\bbeta,\balpha,\bgamma}/\partial
\bbeta=0$. The first two sets of estimating equations are identical to
the covariate balancing estimating equations in (\ref{eq:cvb2}),
whereas the last set of estimating equations
$\partial \hat\mu^{AIPW}_{\bbeta,\balpha,\bgamma}/\partial \bbeta=0$
(leading to the estimators $\bar\balpha,\bar\bgamma$) is unnecessary
in our setting because
$\hat\mu_{\hat\bbeta}=\hat\mu^{AIPW}_{\hat\bbeta,\balpha,\bgamma}$
does not rely on how $\balpha$ and $\bgamma$ are estimated. As
expected, all the theoretical properties of the bias-reduced doubly
robust estimator in Section 3 of \cite{vermeulen2015bias} hold for our
estimator.

Recently, a variety of empirical likelihood based estimators are
proposed to match the moment of covariates in treatment and control
groups
\citep[e.g.,][]{tan2006distributional,tan:10,hain:12,grah:pint:egel:12,han2013estimation,chan:yam:zhan:16,zubizarreta2015stable,zhao2017entropy}. Usually,
these methods aim to estimate $\EE(Y_i(1))$ and $\EE(Y_i(0))$ (or
$\EE(Y_i(1)\mid T_i=0)$ and $\EE(Y_i(0)\mid T_i=1)$) separately and
combine then to estimate the ATE. Our approach directly estimates the
propensity score and ATE by jointly solving the potentially
over-identified estimating functions (\ref{eq:cvb2}). In addition, our
asymptotic results and the discussion rely on the GMM theory for
over-identified estimating functions which is different from these
methods. Another recent paper by \cite{zhao2016covariate} studied the
robustness of a general class of loss function based covariate
balancing methods. When the goal is to estimate the ATE, his score
function reduces to our first set of estimating functions
$\bar\bg_{1\bbeta}(\bT, \bX)$ in (\ref{eq:cvb2}). In this case, his
estimator is robust to the misspecification of the propensity score
model under the constant treatment effect model, i.e., $L(\bX)=\tau^*$
for some constant $\tau^*$. In comparison, our methodology and
theoretical results cover a broader case that allows for heterogeneous
treatment effects.

\section{Nonparametric oCBPS Methodology}\label{secnon}

In this section, we extend our theoretical results of the oCBPS methodology to
 nonparametric estimation.  As seen in Corollary
\ref{cormq}, the proposed estimator is efficient if both the
propensity score $\PP(T_i=1\mid \bX_i)$ and the conditional mean
functions $K(\cdot)$ and $L(\cdot)$ are correctly specified.  To avoid
model misspecification, we can choose a large number of basis
functions $\bh_1(\cdot)$ and $\bh_2(\cdot)$, such that the conditional
mean functions $K(\cdot)$ and $L(\cdot)$ satisfy the condition 2 in Corollary \ref{cormq}.

However, the parametric assumption for the propensity score model
$\PP(T_i=1\mid \bX_i)=\pi_{\bbeta^o}(\bX_i)$ imposed in
Corollary~\ref{cormq} may be too restrictive. Once the propensity
score model is misspecified, the proposed oCBPS-based IPTW estimator
$\hat\mu_{\hat\bbeta}$ is inefficient and could even become
inconsistent. To relax the strong parametric assumptions imposed in
the previous sections, we propose a flexible
nonparametric approach for modeling the propensity
score and the conditional mean functions. The main advantage of this
nonparametric approach is that, the resulting oCBPS-based IPTW
estimator is semiparametrically efficient under a much broader class
of propensity score models and the conditional mean models than those
of Corollary~\ref{cormq}.

Specifically, we assume $\PP(T_i=1\mid \bX_i)=J(\psi^*(\bX_i))$, where
$J(\cdot)$ is a known monotonic link function (e.g.,
$J(\cdot)=\exp(\cdot)/(1+\exp(\cdot))$), and $\psi^*(\cdot)$ is an
unknown smooth function. One practical way to estimate $\psi^*(\cdot)$ is
to approximate it by the linear combination of $\kappa$ basis functions,
where $\kappa$ is allowed to grow with $n$. This approach is known as the
sieve estimation
\citep{andrews1991asymptotic,newey1997convergence}. In detail, let
$\bB(\bx)=\{b_1(\bx),...,b_\kappa(\bx)\}$ denote a collection of $\kappa$ basis
functions, whose mathematical requirement is given in Assumption
\ref{assnon}. Intuitively, we would like to approximate $\psi^*(\bx)$ by
$\bbeta^{*\top}\bB(\bx)$, for some coefficient $\bbeta^*\in\RR^\kappa$.


To estimate $\bbeta^*$, similar to the parametric case, we define
$\bar\bg_{\bbeta}(\bT, \bX)=\sum_{i=1}^n\bg_{\bbeta}(\bT_i, \bX_i)/n$,
where
$\bg_{\bbeta}(\bT_i, \bX_i)=(\bg^{\top}_{1\bbeta}(\bT_i, \bX_i),
\bg^\top_{2\bbeta}(\bT_i, \bX_i))^\top$ with,
\begin{eqnarray*}
	\bg_{1\bbeta}(\bT_i, \bX_i) & = & \left(\frac{T_i}{J(\bbeta^\top\bB(\bX_i))}-\frac{1-T_i}{1-J(\bbeta^\top\bB(\bX_i))}\right)\bh_1(\bX_i),\\
	\bg_{2\bbeta}(\bT_i, \bX_i) & = & \left(\frac{T_i}{J(\bbeta^\top\bB(\bX_i))}-1\right)\bh_2(\bX_i).
\end{eqnarray*}
Recall that $\bh_1(\bX)\in\RR^{m_1}$ and $\bh_2(\bX)\in\RR^{m_2}$ are interpreted as the basis functions for $K(\bX)$ and $L(\bX)$.
Let $m_1+m_2=m$ and $\bh(\bX)=(\bh_1(\bX)^\top,
\bh_2(\bX)^\top)^\top$. Here, we assume $m=\kappa$, so that the number of
equations in $\bar\bg_{\bbeta}(\bT, \bX)$ is identical to the
dimension of the parameter $\bbeta$. Then define
$\tilde{\bbeta}=\arg\min_{\bbeta\in\Theta} \|\bar\bg_{\bbeta}(\bT,
\bX)\|_2^2$, where $\Theta$ is the  parameter space for $\bbeta$ and
$\|\bv\|_2$ represents the $L_2$ norm of the vector $\bv$. The
resulting IPTW estimator is, 
\begin{eqnarray*}
\tilde{\mu}_{\tilde\bbeta} & = & \frac{1}{n} \sum_{i=1}^{n}\left( \frac{T_iY_i}{J(\tilde\bbeta^\top\bB(\bX_i))}-\frac{(1-T_i)Y_i}{1-J(\tilde\bbeta^\top\bB(\bX_i))}\right).
\end{eqnarray*}
To establish the large sample properties of $\tilde{\mu}_{\tilde\bbeta}$, we require a few regularity conditions. Due to the space constraint, we defer the regularity conditions to the supplementary material. The following
theorem establishes the asymptotic normality and semiparametric
efficiency of the estimator $\tilde{\mu}_{\tilde\bbeta}$.

\begin{theorem}[\em Efficiency under nonparametric models] \label{thmnon}
Assume that Assumption~\ref{assnon} in the supplementary material holds, and there exist $r_b, r_h>1/2$, $\bbeta^*$ and $\balpha^*=(\balpha_1^{*\top},\balpha_2^{*\top})^\top\in\RR^{\kappa}$, such that the propensity score model satisfies
\begin{equation}\label{eqnonpropen}
\sup_{\bx\in\mathcal{X}}|\psi^*(\bx)-\bbeta^{*\top}\bB(\bx)|=O(\kappa^{-r_b}),
\end{equation}
and the outcome models $K(\cdot)$ and $L(\cdot)$ satisfy
\begin{equation}\label{eqnonoutcome}
\sup_{\bx\in\mathcal{X}}|K(\bx)-\balpha_1^{*\top}\bh_1(\bx)|=O(\kappa^{-r_h}),~~\sup_{\bx\in\mathcal{X}}|L(\bx)-\balpha_2^{*\top}\bh_2(\bx)|=O(\kappa^{-r_h}).
\end{equation}
Assume $\kappa=o(n^{1/3})$ and $n^{\frac{1}{2(r_b+r_h)}}=o(\kappa)$. Then
$$
n^{1/2}(\tilde{\mu}_{\tilde\bbeta}-\mu)\stackrel{d}{\longrightarrow} N(0,V_{\textrm{opt}}),
$$
where $V_{\textrm{opt}}$ is the asymptotic variance bound in  \eqref{eqinforbound1}. Thus, $\tilde{\mu}_{\tilde\bbeta}$ is semiparametrically efficient.
\end{theorem}

This theorem can be viewed as a nonparametric version of
Corollary~\ref{cormq}. It shows that one can construct a globally
efficient estimator of the treatment effect without imposing strong
parametric assumptions on the propensity score model and the outcome
model. Since the estimator is asymptotically equivalent to the sample
average of the efficient influence function, it is also adaptive in
the sense of \cite{bickel1998efficient}.


In the following, we comment on the technical assumptions of
Theorem~\ref{thmnon}. We assume $\psi^*(\bx)$ and $K(\bx)$ (also
$L(\bx)$) can be uniformly approximated by the basis functions
$\bB(\bx)$ and $\bh_1(\bx)$ (also $\bh_2(\bx)$) in (\ref{eqnonpropen}) and (\ref{eqnonoutcome}), respectively. It is
well known that the uniform rate of convergence is related to the
smoothness of the functions $\psi^*(\bx)$ and $K(\bx)$ (also $L(\bx)$)
and the dimension of $\bX$. For instance, if the function class
$\mathcal{M}$ for $\psi^*(\bx)$ and $\mathcal{H}$ for $K(\bx)$ (also
$L(\bx)$) correspond to the H\"older class with smoothness parameter
$s$ on the domain $\mathcal{X}=[0,1]^d$, under the assumption that $m_1\asymp m_2\asymp \kappa$,  (\ref{eqnonpropen}) and (\ref{eqnonoutcome})
hold for the spline basis and wavelet basis with $r_b=r_h=s/d$; see
\cite{newey1997convergence,chen2007large} for details. In the same setting,
\cite{hira:imbe:ridd:03} considered a nonparametric IPTW estimator,
which is globally efficient under the condition $s/d>7$.
\cite{imbens2005mean} established the asymptotic equivalence between a
regression based estimator and \cite{hira:imbe:ridd:03}'s estimator
under $s/d>9$. Recently, \cite{chan:yam:zhan:16} proposed a sieve
based calibration estimator under the condition $s/d>13$. Compared to
these existing results, our theorem needs a much weaker condition,
i.e.,
$s/d>3/4$. We refer to the supplementary material for further technical discussion of our nonparametric estimator.






\section{Simulation and Empirical Studies}\label{sec_empirical}

\subsection{Simulation Studies}
\label{sec:simulations}

In this section, we conduct a set of simulation studies to examine the
performance of the proposed methodology.  We consider the following
linear model for the potential outcomes,
\begin{eqnarray*}
  Y_i(1) &=& 200 + 27.4 X_{i1} + 13.7X_{i2} + 13.7X_{i3} + 13.7X_{i4} +\varepsilon_i,\\
  Y_i(0) &=& 200 + 13.7X_{i2} + 13.7X_{i3} + 13.7X_{i4} + \varepsilon_i.
\end{eqnarray*}
where $\varepsilon_i \sim N(0, 1)$, independent of $\bX_i$,
and consider the following true propensity score model
\begin{eqnarray}
\PP(T_i=1\mid \bX_i=\bx_i) = \frac{\exp(-\beta_1 x_{i1} + 0.5x_{i2} - 0.25x_{i3} - 0.1x_{i4})}{1+\exp(-\beta_1 x_{i1} + 0.5x_{i2} - 0.25x_{i3} - 0.1x_{i4})}, \label{eq:pscore}
\end{eqnarray}
where $\beta_1$ varies from 0 to 1.  When implementing the proposed
methodology, we set $\bh_1(\bx_{i}) = (1, x_{i2}, x_{i3}, x_{i4})$
and $\bh_2(\bx_{i}) = x_{i1}$ so that the number of equations is
equal to the number of parameters to be estimated.  Covariate $X_{i1}$ is
generated independently from $N(3,2)$ and $X_{i2}$, $X_{i3}$ and $X_{i4}$ are generated from $N(0,1)$.  Each set of results
is based on 500 Monte Carlo simulations.

\begin{table}[htp]
  \begin{center}
    \caption{The bias, standard deviation, root mean squared error (RMSE), and the coverage probability of the constructed 95\% C.I. of the IPTW estimator with known propensity score ({\sf True}), the IPTW estimator when the propensity score is fitted using the maximum likelihood ({\sf GLM}), the IPTW estimator when the propensity score is fitted using the generalized additive model ({\sf GAM}), the targeted maximum likelihood estimator ({\sf DR}), the standard CBPS estimator balancing the first moment ({\sf CBPS}), and the
proposed optimal CBPS estimator ({\sf oCBPS}) under the scenario that both the outcome model and the propensity score model are correctly specified. We vary the value of $\beta_1$ in the data generating model (\ref{eq:pscore}). \label{tb:setting_6}}
\footnotesize
    \begin{tabular}{lc.........}
      \hline
      \hline
      & &\multicolumn{4}{c}{$n=300$}& & \multicolumn{4}{c}{$n=1000$}\\
      \cline{3-6} \cline{8-11}
      & $\beta_1$ & 0 & 0.33 & 0.67 & 1  & & 0 & 0.33 & 0.67 & 1\\
      \hline
      \multirow{ 6}{*}{Bias}
      & {\sf True}&  -0.43 &  -0.01 & 1.15 & -5.19 &&  0.00 &  0.09 & -2.43 &  9.99\\
      & {\sf GLM} &   -0.18 &   -0.86 &  0.15 & -4.32 & &  -0.04 &  0.02 &  0.32 &  11.15\\
      & {\sf GAM} &  -0.74 &  -4.60 &  -15.55 & -35.38 & &  -0.19 &   -1.16 &  -2.85 &  -6.86\\
      & {\sf DR} &   0.08 &  -1.04 &  -3.41 & -8.32 & &   0.18 &   -0.56 &   -2.14 &  -4.50\\
      & {\sf CBPS} &  -0.05 &  -0.09 & 0.54 & -0.27 & &  0.04 &  0.04 &  0.20 & 0.45\\
      & {\sf oCBPS} &  -0.04 &  0.03 & 0.07 & 0.06 & &  0.04 &  0.06 &  0.16 & 0.08\\
      \hline
      \multirow{ 6}{*}{}
      & {\sf True} &  29.52 &  39.46 &  72.56 & 138.33 &&   15.73 &  22.36 &  38.18 & 88.33\\
      & {\sf GLM} &   4.45 &  12.31 &  63.35 & 144.25 &  &  2.21 &   5.49 &  22.93 & 114.45\\
      Std& {\sf GAM} &   4.31 &  14.91 &  43.08 & 100.16 & &  2.06 &  5.22 & 21.27 &  51.96\\
      Dev& {\sf DR} &   2.39 &  2.57 &  4.25 & 8.06 & &  1.20 &  1.29 &  1.76 &  3.32\\
      & {\sf CBPS} & 2.39 & 2.35 & 2.66 & 15.94 && 1.24 & 1.26 & 1.27 & 1.45\\
      & {\sf oCBPS} & 2.26 & 2.16 & 2.27 & 2.39 & & 1.20 & 1.20 & 1.18 & 1.22\\
      \hline
      \multirow{ 6}{*}{RMSE}
      & {\sf True}&  29.52 &  39.46 &  72.57 & 138.43 & &  15.73 &  22.36 &  38.26 & 88.89\\
      & {\sf GLM} &   4.46 &  12.34 &  63.35 & 144.32 & &   2.21 &   5.49 &  22.93 & 114.99\\
      & {\sf GAM} &   4.37 &   15.60 &  45.81 & 106.23 & &   2.07 &   5.35 &   21.46 &  52.41\\
      & {\sf DR} &   2.39 &   2.77 &  5.45 & 11.58 & &   1.21 &   1.41 &   2.77 &  5.59\\
      & {\sf CBPS} & 2.39 & 2.35 & 2.72 & 15.94 & & 1.24 & 1.26 & 1.29 & 1.52\\
      & {\sf oCBPS} & 2.26 & 2.16 & 2.27 & 2.39 && 1.20 & 1.20 & 1.19 & 1.23\\
      \hline
      \multirow{ 6}{*}{}
      & {\sf True}&  0.936 &  0.938 & 0.922 & 0.948 & & 0.962 & 0.942 &  0.926 & 0.948\\
      Coverage& {\sf GLM} &   0.946 &  0.946 & 0.946 & 0.946 & & 0.944 & 0.954 &  0.954 & 0.958\\
      Probability& {\sf GAM} &   0.704 &  0.310 & 0.090 & 0.028 & & 0.754 & 0.382 &  0.108 & 0.048\\
      (of the & {\sf DR} &   0.928 &  0.876 & 0.576 & 0.278 & & 0.960 & 0.906 &  0.562 & 0.268\\
      95\% C.I.)& {\sf CBPS} & 0.944 &  0.944 & 0.944 & 0.944 & & 0.960 & 0.958 &  0.958 & 0.968\\
      & {\sf oCBPS} & 0.950 &  0.964 & 0.962 & 0.982 & & 0.956 & 0.954 &  0.962 & 0.966\\
      \hline
    \end{tabular}
  \end{center}
\end{table}

We examine the performance of the IPTW estimator when the propensity
score model is fitted using maximum likelihood ({\sf GLM}), the
standard CBPS with balancing the first moment ({\sf CBPS}), and the
proposed optimal CBPS ({\sf oCBPS}) as well as the case where the true
propensity score ({\sf True}), i.e., $\bbeta = \bbeta^*$, is used for
the IPTW estimator. In addition, we include the IPTW estimator when
the propensity score model is estimated by logistic series
\citep{hira:imbe:ridd:03}. Since a fully nonparametric logistic series
approach is impractical to implement due to the curse of
dimensionality, instead we consider a generalized additive model  ({\sf GAM})
and apply the logistic series approach to each of the covariate
separately. Finally, we also include the targeted maximum likelihood
estimator, a doubly robust estimator ({\sf DR},
\cite{benkeser2016doubly}), using the R package {\sf drtmle}.  

In the first set of simulations, we use the correctly specified propensity score and outcome models. 
Table~\ref{tb:setting_6} shows the
standard deviation, bias, root mean square error (RMSE), and the coverage probability of the constructed 95\% confidence intervals of these estimators when the sample size is $n=300$ and $n=1000$. The confidence intervals are constructed using estimates of the asymptotic variances of the estimators. The exact formulas can be found in the supplementary material. We find that
 {\sf CBPS} and {\sf oCBPS} substantially outperform {\sf
  True}, {\sf GLM}, and {\sf GAM} in terms of efficiency, and in most cases outperform {\sf DR} as well.  
  In addition, {\sf oCBPS} is more efficient than {\sf CBPS} in all the cases as well.  The efficiency improvement is consistent with Corollary~\ref{cormq}. The coverage probabilities of {\sf True}, {\sf GLM},  {\sf CBPS} and {\sf oCBPS} are close to the nominal level. However, {\sf GAM} yields much lower coverage probability, partly because the estimates of the propensity score from logistic series are unstable. The pattern becomes more evident as $\beta_1$ increases, corresponding to the setting that the propensity score can be close to 0 or 1. Similarly, the coverage probability of {\sf DR} also deteriorates  as $\beta_1$ increases. 
    

\begin{table}[htp]
\begin{center}
  \caption{Correct Outcome Model with a Misspecified Propensity Score
    Model. \label{tb:setting_7}}
 \footnotesize
 \begin{tabular}{lc.........}
      \hline
      \hline
      & &\multicolumn{4}{c}{$n=300$}& & \multicolumn{4}{c}{$n=1000$}\\
      \cline{3-6} \cline{8-11}
      & $\beta_1$ & 0 & 0.33 & 0.67 & 1  & & 0 & 0.33 & 0.67 & 1\\
      \hline
      \multirow{ 6}{*}{Bias}
      & {\sf True}&  0.00 &  2.13 & 0.08 & 4.79 &&  -1.28 &  -0.36 & 1.83 &  3.62\\
      & {\sf GLM} &   0.41 &   -6.67 &  -18.84 & -32.15 & &  0.19 &  -6.33 &  -19.21 & -32.96\\
      & {\sf GAM} &  15.61 &  3.11 &  -7.16 & -20.76 & &  4.07 &   0.28 &  -4.98 &  -14.11\\
      & {\sf DR} &   -0.29 &  -0.68 &  -1.89 & -3.60 & &   -0.21 &   -0.39 &   -1.23 &  -2.75\\
      & {\sf CBPS} &  0.84 &  -0.05 & -2.06 & -2.44 & &  0.06 &  -0.79 &  -2.74 & -3.28\\
      & {\sf oCBPS} &  -0.20 &  -0.02 & -0.13 & 0.07 & &  -0.04 &  0.03 &  0.01 & -0.05\\
      \hline
      \multirow{ 6}{*}{}
      & {\sf True} &  45.43 &  36.03 &  39.77 & 77.26 &&   26.32 &  19.36 &  39.15 & 88.45\\
      & {\sf GLM} &   11.23 &  12.66 &  15.73 & 26.82 &  &  2.17 &   5.32 &  8.61 & 10.92\\
      Std& {\sf GAM} &   19.91 &  9.40 & 8.81 & 16.18 & &  4.29 &  2.87 & 4.14 &  8.52\\
      Dev& {\sf DR} &   3.35 &  2.57 &  2.52 & 3.16 & &  1.42 &  1.27 &  1.28 &  1.57\\
      & {\sf CBPS} & 3.21 & 2.74 & 3.18 & 3.61 && 1.25 & 1.41 & 1.74 & 2.04\\
      & {\sf oCBPS} & 2.26 & 2.30 & 2.28 & 2.34 & & 1.24 & 1.26 & 1.24 & 1.29\\
      \hline
      \multirow{ 6}{*}{RMSE}
      & {\sf True}&  45.43 &  36.10 &  39.77 & 77.40 & &  26.36 &  19.36 &  39.20 & 88.52\\
      & {\sf GLM} &   11.24 &  14.31 &  24.55 & 41.86 & &   2.18 &   8.27 &  21.05 & 34.72\\
      & {\sf GAM} &   25.30 &   9.90 &  11.35 & 26.32 & &   5.91 &   2.89 &   6.47 &  16.48\\
      & {\sf DR} &   3.37 &   2.65 &  3.15 & 4.79 & &   1.44 &   1.33 &   1.78 &  3.16\\
      & {\sf CBPS} & 3.32 & 2.74 & 3.79 & 4.36 & & 1.26 & 1.62 & 3.24 & 3.86\\
      & {\sf oCBPS} & 2.27 & 2.30 & 2.29 & 2.34 && 1.24 & 1.26 & 1.24 & 1.29\\
      \hline
      \multirow{ 6}{*}{}
      & {\sf True}&  0.952 &  0.936 & 0.964 & 0.972 & & 0.946 & 0.950 &  0.960 & 0.988\\
      Coverage& {\sf GLM} &   0.964 &  0.898 & 0.740 & 0.834 & & 0.948 & 0.714 &  0.300 & 0.346\\
      Probability& {\sf GAM} &   0.236 &  0.434 & 0.286 & 0.066 & & 0.356 & 0.648 &  0.178 & 0.042\\
      (of the & {\sf DR} &   0.882 &  0.904 & 0.822 & 0.596 & & 0.908 & 0.938 &  0.788 & 0.392\\
      95\% C.I.)& {\sf CBPS} & 0.956 &  0.978 & 0.924 & 0.914 & & 0.944 & 0.928 &  0.742 & 0.654\\
      & {\sf oCBPS} & 0.946 &  0.944 & 0.952 & 0.944 & & 0.950 & 0.950 &  0.954 & 0.954\\
      \hline
    \end{tabular}
\end{center}
\end{table}

We further evaluate our method by considering different cases of
misspecification for the outcome and propensity score models.  We
begin with the case where the outcome model is linear like before but
the propensity score is misspecified.  While we use the model given in
equation \eqref{eq:pscore} when estimating the propensity score, the
actual treatment is generated according to the following different
model,
\begin{eqnarray*}
\PP(T_i=1\mid \bX=\bx_i) & = &  \frac{\exp(-\beta_1x_{i1}^\ast + 0.5x_{i2}^\ast - 0.25x_{i3}^\ast - 0.1x_{i4}^\ast)}{1+\exp(-\beta_1x_{i1}^\ast + 0.5x_{i2}^\ast - 0.25x_{i3}^\ast - 0.1x_{i4}^\ast)},
\end{eqnarray*}
with $x_{i1}^\ast = \exp(x_{i1}/3)$,
$x_{i2}^\ast = x_{i2}/\{1+\exp(x_{i1})\} + 10$,
$x_{i3}^\ast = x_{i1}x_{i3}/25+0.6$, and
$x_{i4}^\ast = x_{i1}+x_{i4}+20$ where $\beta_1$ again varies from 0
to 1.  In other words, the model misspecification is introduced using
nonlinear transformations.  Table~\ref{tb:setting_7} shows the results
for this case.  As expected from the double robustness property shown
in Theorem~\ref{thmdr}, we find that the bias for the {\sf oCBPS}
becomes significantly smaller than all the other estimators.  
The {\sf oCBPS} also dominates the other estimators in terms of efficiency and maintains the desired coverage probability.

\begin{table}[htp]
\begin{center}
  \caption{Correctly Specified Outcome with a Locally Misspecified Propensity Score
    Model. \label{tb:setting_10}}
    \footnotesize
   \begin{tabular}{lc.........}
      \hline
      \hline
      & &\multicolumn{4}{c}{$n=300$}& & \multicolumn{4}{c}{$n=1000$}\\
      \cline{3-6} \cline{8-11}
      & $\beta_1$ & 0 & 0.33 & 0.67 & 1  & & 0 & 0.33 & 0.67 & 1\\
      \hline
      \multirow{ 6}{*}{Bias}
      & {\sf True}&  -1.96 &  0.69 & 0.80 & 4.87 &&  0.04 &  0.87 & -0.42 &  3.07\\
      & {\sf GLM} &   -16.73 &   8.43 &  5.85 & 19.96 && 8.55 &  0.84 &  4.65 &  21.07 \\
      & {\sf GAM} &  -8.19 &  7.68 &  -4.35 & -10.79 & &  4.62 &   -0.25 &  -0.63 &  2.95\\
      & {\sf DR} &   0.43 &  0.34 &  -0.83 & -3.67 & &   0.38 &   0.08 &   -1.39 &  -3.50\\
      & {\sf CBPS} &  -0.76 &  -2.15 & 0.56 & 1.34 & &  -1.92 &  -0.34 &  0.22 & 0.37\\
      & {\sf oCBPS} &  -0.41 &  0.05 & 0.10 & 0.06 & &  -0.05 &  0.02 &  -0.01 & -0.02 \\
      \hline
      \multirow{ 6}{*}{}
      & {\sf True} &  41.03 &  33.16 &  41.86 & 82.09 &&   20.65 &  18.39 &  28.44 & 59.63\\
      & {\sf GLM} &   67.79 &  9.55 &  23.67 & 72.99 &&  9.43 &   3.23 &  13.86 & 81.20\\
      Std& {\sf GAM} &   46.08 &  8.92 & 21.56 & 52.34 & &  11.06 &2.91 & 11.78 &  52.31\\
      Dev& {\sf DR} &   3.10 &  2.51 &  2.87 & 5.74 & &  1.37 &  1.29 &  1.59 &  2.60\\
      & {\sf CBPS} & 3.26 & 2.56 & 2.44 & 2.77 && 1.58 & 1.28 & 1.33 & 1.43\\
      & {\sf oCBPS} & 2.47 & 2.24 & 2.25 & 2.26 & & 1.29 & 1.22 & 1.26 & 1.29\\
      \hline
      \multirow{ 6}{*}{RMSE}
      & {\sf True}&  41.07 &  33.17 &  41.87 & 82.24 & &  20.65 &  18.41 &  28.44 & 59.70\\
      & {\sf GLM} &   69.82 &  12.74 &  24.39 & 75.67 && 12.73 &   3.34 &   14.62 &  83.89 \\
      & {\sf GAM} &   46.80 &   11.77 &  21.99 & 53.44 & &   11.98 &  2.92 &   11.80 &  52.39\\
      & {\sf DR} &   3.13 &   2.53 &  2.99 & 6.81 & &   1.42 &   1.29 &   2.11 &  4.36\\
      & {\sf CBPS} & 3.35 & 3.34 & 2.51 & 3.07 & & 2.49 & 1.32 & 1.34 & 1.48\\
      & {\sf oCBPS} & 2.50 & 2.24 & 2.26 & 2.27 && 1.29 & 1.22 & 1.26 & 1.29\\
      \hline
      \multirow{ 6}{*}{}
      & {\sf True}&  0.962 &  0.948 & 0.962 & 0.938 & & 0.934 & 0.946 &  0.954 & 0.942\\
      Coverage& {\sf GLM} &   0.804 &  0.788 & 0.888 & 0.916 & & 0.652 & 0.936 &  0.918 & 0.910\\
      Probability& {\sf GAM} &   0.132 &  0.294 & 0.238 & 0.076 & & 0.154 & 0.612 &  0.144 & 0.052\\
      (of the& {\sf DR} &   0.856 &  0.922 & 0.866 & 0.556 & & 0.916 & 0.936 &  0.736 & 0.332\\
      95\% C.I.)& {\sf CBPS} & 0.912 &  0.914 & 0.926 & 0.958 & & 0.752 & 0.954 &  0.954 & 0.952\\
      & {\sf oCBPS} & 0.916 &  0.946 & 0.936 & 0.954 & & 0.950 & 0.948 &  0.958 & 0.954\\
      \hline
    \end{tabular}
\end{center}
\end{table}

We also consider the case when the propensity score is locally misspecified with the equation (\ref{eq:misp}). In the case, we use \eqref{eq:pscore} as the working model $\pi_{\bbeta}(\bX_i)$, set $\xi=n^{-1/2}$ as in Theorem \ref{th:asymnormal_mis} and choose the function $u(\bX_i;\bbeta)=X_{i1}^2$ as the direction of misspecification.  We compute the true propensity score from the model (\ref{eq:misp}) and use it to generate the treatment variables. We note that sometimes the true propensity score may exceed 1. In this case we simply replace its value with $0.95$.  The results are given in Table~\ref{tb:setting_10}. {\sf oCBPS} dominates all the other estimators in terms of bias, standard deviation and root mean square error, but {\sf CBPS} and {\sf DR} are also noticeably better than {\sf True}, {\sf GLM}, and {\sf GAM}. 

\begin{table}[htp]
\begin{center}
  \caption{Misspecified Outcome Model with Correct Propensity Score
    Model.}\label{tb:setting_8}
   \footnotesize
   \begin{tabular}{lc.........}
      \hline
      \hline
      & &\multicolumn{4}{c}{$n=300$}& & \multicolumn{4}{c}{$n=1000$}\\
      \cline{3-6} \cline{8-11}
      & $\beta_1$ & 0 & 0.13 & 0.27 & 0.4  & & 0 & 0.13 & 0.27 & 0.4\\
      \hline
      \multirow{ 6}{*}{Bias}
      & {\sf True}&  -4.37 &  -0.03 & -4.24 & 1.51 &&  0.80 &  -1.00 & 2.31 &  2.67\\
      & {\sf GLM} &   0.38 &   -0.64 &  -2.67 & -1.33 & &  0.11 &  -0.44 &  0.05 & 0.75\\
      & {\sf GAM} &  -2.03 &  -5.49 &  -10.43 & -13.66 & &  -0.65 &   -1.72 &  -1.95 &  -3.04\\
      & {\sf DR} &   -2.77 &  -5.06 &  -9.92 & -14.36 & &   -2.98 &   -4.98 &   -7.43 &  -10.11\\
      & {\sf CBPS} &  0.07 &  -0.69 & -2.59 & -3.94 & &  0.05 &  -0.55 &  -0.71 & -1.63\\
      & {\sf oCBPS} &  -0.56 &  -0.97 & -3.05 & -4.37 & &  -0.03 &  -0.68 &  -0.84 & -1.70\\
      \hline
      \multirow{ 6}{*}{}
      & {\sf True} &  49.87 &  58.75 &  74.32 & 100.35 &&   27.61 &  33.62 &  44.75 & 53.58\\
      & {\sf GLM} &   18.12 &  24.87 &  34.83 & 56.17 &  &  9.68 &   12.37 &  18.45 & 31.16\\
      Std& {\sf GAM} &   17.59 &  23.19 & 34.72 & 49.87 & &  9.07 & 11.36 & 16.85 &  26.50\\
      Dev& {\sf DR} &   14.02 &  14.65 &  15.58 & 16.65 & &  7.95 &  8.26 &  8.21 &  8.40\\
      & {\sf CBPS} & 15.51 & 17.60 & 18.83 & 20.66 && 8.74 & 9.47 & 10.64 & 12.05\\
      & {\sf oCBPS} & 14.74 & 16.15 & 17.13 & 18.55 & & 8.44 & 9.03 & 9.68 & 10.87\\
      \hline
      \multirow{ 6}{*}{RMSE}
      & {\sf True}&  50.06 &  58.75 &  74.45 & 100.36 & &  27.62 &  33.64 &  44.81 & 53.60\\
      & {\sf GLM} &   18.13 &  24.88 &  34.93 & 56.18 & &   9.68 &   12.37 &  18.45 & 31.17\\
      & {\sf GAM} &   17.71 &   23.83 &  36.25 & 51.71 & &   9.09 &  11.49 &   16.96 &  26.67\\
      & {\sf DR} &   14.29 &   15.50 &  18.47 & 21.99 & &   8.49 &   9.65 &   11.07 &  13.15\\
      & {\sf CBPS} & 15.51 & 17.62 & 19.01 & 21.03 & & 8.74 & 9.49 & 10.66 & 12.16\\
      & {\sf oCBPS} & 14.75 & 16.18 & 17.40 & 19.06 && 8.44 & 9.06 & 9.72 & 11.00\\
      \hline
      \multirow{ 6}{*}{}
      & {\sf True}&  0.948 &  0.954 & 0.946 & 0.920 & & 0.938 & 0.950 &  0.910 & 0.922\\
      Coverage& {\sf GLM} &   0.896 &  0.852 & 0.870 & 0.868 & & 0.908 & 0.862 &  0.816 & 0.802\\
      Probability& {\sf GAM} &   0.912 &  0.832 & 0.676 & 0.476 & & 0.932 & 0.846 &  0.690 & 0.516\\
      (of the & {\sf DR} &   0.930 &  0.910 & 0.838 & 0.716 & & 0.924 & 0.874 &  0.794 & 0.688\\
      95\% C.I.) & {\sf CBPS} & 0.920 &  0.870 & 0.790 & 0.676 & & 0.914 & 0.862 &  0.776 & 0.668\\
      & {\sf oCBPS} & 0.950 &  0.930 & 0.908 & 0.904 & & 0.954 & 0.920 &  0.902 & 0.862\\
      \hline
    \end{tabular}
    \end{center}
\end{table}

We next examine the cases where the outcome model is misspecified.  We
do this by generating potential outcomes from the following quadratic model
\begin{eqnarray*}
\EE(Y_i(1)\mid\bX_i=\bx_i) &=& 200 + 27.4 x_{i1}^2 + 13.7x_{i2}^2 + 13.7x_{i3}^2 + 13.7x_{i4}^2,\\
\EE(Y_i(0)\mid\bX_i=\bx_i) &=& 200 + 13.7x_{i2}^2 + 13.7x_{i3}^2 + 13.7x_{i4}^2,
\end{eqnarray*}
whereas the propensity score model is the same as the one in
\eqref{eq:pscore} with $\beta_1$ varying from 0 to 0.4.
Table~\ref{tb:setting_8} shows the results when the outcome model is
misspecified but the propensity score model is correct. We find that
the magnitude of bias is similar across all estimators with the exception of 
{\sf GAM} and {\sf DR}, which seem to have a significantly larger bias. 
The {\sf DR} dominates in terms of standard deviation, but {\sf oCBPS} closely follows. 
In terms of the root mean square error, {\sf oCBPS} is on par with  {\sf DR}. 

\begin{table}[h!]
\begin{center}
  \caption{Misspecified Outcome with Misspecified Propensity Score
    Models. \label{tb:setting_9}}
    \footnotesize
   \begin{tabular}{lc.........}
      \hline
      \hline
      & &\multicolumn{4}{c}{$n=300$}& & \multicolumn{4}{c}{$n=1000$}\\
      \cline{3-6} \cline{8-11}
      & $\beta_1$ & 0 & 0.13 & 0.27 & 0.4  & & 0 & 0.13 & 0.27 & 0.4\\
      \hline
      \multirow{ 6}{*}{Bias}
      & {\sf True}&  0.54 &  -1.74 & 1.71 & -3.56 &&  -2.66 &  -2.52 & -2.06 &  -0.36\\
      & {\sf GLM} &   2.94 &   -1.70 &  -8.47 & -20.25 && -0.18 &  -2.07 &  -8.89 &  -18.79 \\
      & {\sf GAM} &  20.74 &  12.05 &  3.42 & -8.06 & &  4.95 &   2.35 &  -1.03 &  -5.01\\
      & {\sf DR} &   9.16 &  6.66 &  4.52 & 0.46 & &   6.55 &   4.91 &   2.80 &  0.36\\
      & {\sf CBPS} &  9.57 &  4.10 & 0.37 & -7.62 & &  0.46 &  -0.81 &  -4.94 & -11.18\\
      & {\sf oCBPS} &  2.51 &  -0.24 & -1.62 & -4.82 & &  0.04 &  -0.61 &  -2.29 & -4.54\\
      \hline
      \multirow{ 6}{*}{}
      & {\sf True} &  59.12 &  55.64 &  54.16 & 58.35 &&   34.79 &  31.31 &  28.41 & 31.62\\
      & {\sf GLM} &   25.00 &  19.44 &  22.49 & 26.01 &  &  9.67 &   9.79 &  11.17 & 12.44\\
      Std& {\sf GAM} &   30.85 &  23.01 & 19.46 & 21.72 & &  10.23 & 9.53 & 9.19 &  9.25\\
      Dev& {\sf DR} &   15.18 &  15.14 &  13.71 & 13.60 & &  7.86 &  7.85 &  7.69 &  7.70\\
      & {\sf CBPS} & 26.74 & 18.65 & 19.74 & 18.92 && 9.16 & 9.11 & 9.36 & 9.63\\
      & {\sf oCBPS} & 16.28 & 15.38 & 15.08 & 14.42 & & 8.93 & 8.60 & 8.32 & 8.27\\
      \hline
      \multirow{ 6}{*}{RMSE}
      & {\sf True}&  59.12 &  55.66 &  54.19 & 58.45 & &  34.89 &  31.42 &  28.48 & 31.62\\
      & {\sf GLM} &   25.18 &  19.51 &  24.03 & 32.96 && 9.67 &   10.00 &   14.28 &  22.53 \\
      & {\sf GAM} &   37.17 &   25.97 &  19.76 & 23.17 & &   11.37 &  9.81 &   9.25 &  10.52\\
      & {\sf DR} &   17.73 &   16.54 &  14.43 & 13.60 & &   10.24 &   9.26 &   8.19 &  7.71\\
      & {\sf CBPS} & 28.40 & 19.10 & 19.75 & 20.40 & & 9.17 & 9.15 & 10.59 & 14.76\\
      & {\sf oCBPS} & 16.47 & 15.38 & 15.17 & 15.20 && 8.93 & 8.62 & 8.63 & 9.43\\
      \hline
      \multirow{ 6}{*}{}
      & {\sf True}&  0.952 &  0.940 & 0.936 & 0.952 & & 0.936 & 0.940 &  0.952 & 0.916\\
      Coverage& {\sf GLM} &   0.854 &  0.902 & 0.866 & 0.788 & & 0.890 & 0.878 &  0.772 & 0.540\\
      Probability& {\sf GAM} &   0.714 &  0.810 & 0.860 & 0.832 & & 0.868 & 0.902 &  0.916 & 0.834\\
      (of the& {\sf DR} &   0.878 &  0.920 & 0.934 & 0.946 & & 0.876 & 0.906 &  0.936 & 0.946\\
      95\% C.I.)& {\sf CBPS} & 0.866 &  0.892 & 0.890 & 0.866 & & 0.894 & 0.888 &  0.852 & 0.670\\
      & {\sf oCBPS} & 0.940 &  0.964 & 0.926 & 0.934 & & 0.944 & 0.942 &  0.926 & 0.894\\
      \hline
    \end{tabular}
\end{center}
\end{table}

Finally, when both the outcome and propensity score models
are misspecified, we observe that {\sf DR} and {\sf oCBPS} dominate all other estimators with respect to all three criteria. In particular, {\sf oCBPS} performs much better than {\sf CBPS} in all scenarios. The results are organized in Table~\ref{tb:setting_9}.

In summary, the proposed {\sf oCBPS} method outperforms the {\sf CBPS}
method with respect to root mean square error (RMSE) under all five scenarios we
examined. In addition, the {\sf oCBPS} method often yields better or at least comparable results relative to all the other estimators.

\subsection{An Empirical Application}

We next apply the oCBPS methodology to a well-known study where the
experimental benchmark estimate is available.  Specifically,
\cite{lalonde1986evaluating} conducted a study, in which after the
randomized evaluation study was implemented, the experimental control
group is replaced with a set of untreated individuals taken from the
Panel Study of Income Dynamics.  This created an artificial
observational study with 297 treated observations and 2,490
control observations.  Ever since the original study, this data set
has been used for evaluating whether a new statistical methodology can
recover the experimental benchmark estimate \citep[see
e.g.,][]{dehejia1999causal,smith2005does}.  In the original CBPS
article, \cite{imai:ratk:14} use this data set to show that the
propensity score matching estimator based on the CBPS method
outperforms the matching estimator based on the standard logistic
regression. In the following, we evaluate whether the proposed oCBPS
method can further improve the CBPS methodology.

We begin by replicating the original results of \cite{imai:ratk:14}
and then compare those results with those of the proposed oCBPS
methodology.  To do this, we focus on the estimation of the average
treatment effect for the treated (ATT). The response of interest is earnings in 1978 and the treatment variable is whether or not the individual participates the job training program. 
The original randomized
experiment yields the ATT estimate \$886, which is used as a benchmark
for the later comparison.  \cite{imai:ratk:14} consider the propensity
score estimation based on the standard logistic regression (GLM), the
just-identified CBPS with moment balance condition only (CBPS1) and
the over-identified CBPS with score equation and moment balance
condition (CBPS2).  Based on each set of these estimated propensity
scores, we estimate the ATT using the 1-to-1 nearest neighbor matching
with replacement. The estimates of standard errors are based on the
results in \cite{abadie2006large}.  We then add the estimated
propensity score based on the proposed oCBPS methodology.  Since the
quantity of interest is the ATT, we use a slightly modified oCBPS
estimator described in Appendix~\ref{app_att}.

We follow the propensity score model specifications examined in
\cite{imai:ratk:14}.  The covariates we adjust include age, education,
race (white, black or Hispanic), marriage status, high school degree,
earnings in 1974 and earnings in 1975 as pretreatment variables. We
consider three different specification of balance conditions: the
first moment of covariates (Linear), the first and second moment of
covariates (Quadratic), and the Quadratic specification with some
interactions selected by \cite{smith2005does} (Smith \& Todd).  We
compare the performance of each methodology across these three
specifications.

\begin{table}[t]
\begin{center}
  \caption{The bias and standard errors (shown in parentheses) of estimates of the average treatment effect for the treated in the LaLonde's Study. We use the benchmark \$886 as the true value.  \label{tab_data}}
\begin{tabular}{lcccc}
\hline
\hline
&  {\sf GLM} & {\sf CBPS1} & {\sf CBPS2} & {\sf oCBPS} \\
\hline
Linear &-1190.92 & -462.7 & -702.33 & -306.01 \\
&(1437.02)& (1295.19)& (1240.79)& (1662.02)\\
Quadratic &  -1808.16 &  -646.54 & 207.13 & -370.03 \\
&(1382.38) &(1284.13)& (1567.33)& (1773.03)\\
Smith \& Todd &  -1620.49 &  -1154.07 &  -462.24 &  -383.12\\
&(1424.57)& (1711.66)& (1404.15)& (1748.87)\\
\hline\hline
\end{tabular}
\end{center}
\end{table}

The results are shown in Table~\ref{tab_data}.  We find that although
the standard error is relatively large as in any evaluation study
based on the LaLonde data, the proposed oCBPS method yields much
smaller bias than GLM and CBPS1 under all three specifications. The
oCBPS also improves the CBPS2 under the linear and Smith \& Todd's
specifications of covariates. We note that the standard error of the
oCBPS method appears to be larger than the competing methods. This may
be due to the fact that the uncertainty of the estimated propensity
score is ignored when we calculate the standard error of the matching
estimators (i.e., GLM, CBPS1 and CBPS2). In summary, consistent with
the theoretical results, the proposed oCBPS method yields more
accurate estimates of ATT than the original CBPS estimator or the
standard logistic regression.  Finally, it is important to note that
these results are only suggestive since we do not know whether the
assumptions of propensity score methods hold in this study.

\section{Conclusion}

This paper presents a theoretical investigation of the covariate
balancing propensity score methodology that others have found work
well in practice \citep[e.g.,][]{wyss:etal:14,frol:hube:wies:15}.  We
derive the optimal choice of the covariate balancing function so that
the resulting IPTW estimator is first order unbiased under local
misspecification of the propensity score model.  Furthermore, it turns
out that the CBPS-based IPTW estimator with the same covariate
balancing function attains the semiparametric efficiency bound.

Given these theoretical insights, we propose an optimal CBPS
methodology by carefully choosing the covariate balancing estimating
functions.  We prove that the proposed {oCBPS}-based IPTW estimator is
doubly robust and locally
efficient.  
More importantly, we show that the rate of convergence of the proposed
estimator is faster than the standard AIPW estimator under locally
misspecified models.  To relax the parametric assumptions and improve
the double robustness property, we further extend the oCBPS method to
the nonparametric setting. We show that the proposed
estimator can achieve the semiparametric efficiency bound without
imposing parametric assumptions on the propensity score and outcome
models. The theoretical results require weaker technical conditions
than existing methods and the estimator has smaller asymptotic
bias. Our simulation and empirical studies confirm the theoretical
results, demonstrating the advantages of the proposed oCBPS
methodology.

	In this work, we mainly focus on the theoretical development of the IPTW estimator with the propensity score estimated by the optimal CBPS approach. It is a very interesting research problem to establish the theoretical results for the matching estimators combined with the optimal CBPS approach or some variants. While the asymptotic theory (i.e., consistency and asymptotic normality) for the estimated propensity score  via the optimal CBPS approach can be  derived from the current results (by the Delta method), the full development is beyond the scope of this work. We leave it for a future study.\
	


\section*{Supplementary Material}
The supplementary material contains the Appendix of this paper which collects the proofs and further technical details. 

\setlength{\bibsep}{0.85pt}{
\bibliographystyle{ims}
\bibliography{cbps,my,imai,spglm}
}

\newpage

\appendix

\section*{Supplementary Material}

\section{Locally Semiparametric Efficient Estimator}

For clarification, we reproduce the following definition of
locally semiparametric efficient estimator given in
\cite{robins1994estimation},
\begin{definition}
Given a semiparametric model, say $A$, and an additional restriction $R$ on the joint distribution of the data not imposed by the model, we say that an estimator $\hat\alpha$ is locally semiparametric efficient in model $A$ at $R$ if $\hat\alpha$ is a semiparametric estimator in model $A$ whose asymptotic variance attains the semiparametric variance bound for model $A$ when $R$ is true. 
\end{definition}

In our setting, the semiparametric model $A$ corresponds to the joint
distribution of the observed data $(T_i,Y_i,\bX_i)$ subject to the
strong ignorability of the treatment assignment
$ \{Y_i(1), Y_i(0)\} \perp T_i \mid \bX_i$; see
\cite{hahn1998role}. The semiparametric variance bound for model $A$
is $V_{opt}$. The restriction $R$ is the intersection of $R_1$ and
$R_2$ (denoted by $R_1\cap R_2$), where $R_1$ is the model that
satisfies the first condition in Theorem 3.1 (i.e., the propensity
score is correctly specified) and $R_2$ is the model that satisfies
the second condition in Theorem 3.1 (i..e,
$K(\bX_i)=\balpha_1^\top\Mb_1\bh_1(\bX_i)$ and
$L(\bX_i)=\balpha_2^\top\Mb_2\bh_2(\bX_i)$). In Corollary 3.2, we show
that the asymptotic variance of our estimator of ATE
$\hat\mu_{\hat\bbeta}$ is $V_{opt}$ when $R_1\cap R_2$ is true. From
the above definition of locally semiparametric efficient estimator, we
can claim that $\hat\mu_{\hat\bbeta}$ is locally semiparametric
efficient at $R_1\cap R_2$.

\section{Preliminaries}
\label{appendix::A1}

To simplify the notation, we use $\pi_i^* = \pi_{\bbeta^*}(\bX_i)$ and $\pi_i^o = \pi_{\bbeta^o}(\bX_i)$. For any vector $\bC\in\RR^{K}$, we denote $|\bC|=(|C_1|,...,|C_K|)^\top$ and write $\bC\leq \bB$ for $C_k\leq B_k$ for any $1\leq k\leq K$.

\begin{assumption} (Regularity Conditions for CBPS in Section \ref{sec:model.misspecification})
\label{reg_mis}
\begin{enumerate}
	\item There exists a positive definite matrix $\Wb^*$ such that
	$\hat{\Wb}\stackrel{p}{\longrightarrow} \Wb^*$.
	\item The minimizer
	$\bbeta^o=\argmin_{\bbeta} \EE(\bar\bg_{\bbeta}(\bT,\bX))^\top
	\Wb^*\EE(\bar\bg_{\bbeta}(\bT,\bX))$ is unique.
	\item  $\bbeta^o\in\textrm{int}(\Theta)$, where $\Theta$ is a
	compact set.
	\item $\pi_{\bbeta}(\bX)$ is continuous in $\bbeta$.
	\item There exists a constant $0<c_0<1/2$ such that with
	probability tending to one, $c_0\leq \pi_{\bbeta}(\bX)\leq 1-c_0$,
	for any $\bbeta\in \textrm{int}(\Theta)$.
	\item $\EE|f_j(\bX)|<\infty$ for $1\leq j\leq m$ and $\EE|Y(1)|^2<\infty$, $\EE|Y(0)|^2<\infty$.
	\item  $\Gb^*:=\EE(\partial \bg({\bbeta^o})/\partial\bbeta)$
	exists and there is a $q$-dimensional function $C(\bX)$ and a
	small constant $r>0$ such that
	$\sup_{\bbeta\in\mathbb{B}_r(\bbeta^o)}|\partial\pi_{\bbeta}(\bX)/\partial\bbeta|\leq
	C(\bX)$
	and $\EE(|f_{j}(\bX)|C(\bX))<\infty$ for $1\leq j\leq m$, where
	$\mathbb{B}_r(\bbeta^o)$ is a ball in $\RR^q$ with radius $r$ and
	center $\bbeta^o$. In addition, $\EE(|Y|C(\bX))<\infty$.
	\item $\Gb^{*\top}\Wb^*\Gb^*$ and
	$\EE(\bg_{\bbeta^o}(T_i,\bX_i)\bg_{\bbeta^o}(T_i,\bX_i)^\top)$
	are nonsingular.
	\item In the locally misspecified model (\ref{eq:misp}), assume $|u(\bX; \bbeta^*)|\leq C$ almost surely for some constant $C>0$.
\end{enumerate}
\end{assumption}

\begin{lemma}[Lemma 2.4 in \cite{newey1994large}]\label{lemulan}
	Assume that the data $Z_i$ are i.i.d., $\Theta$ is compact, $a(Z,\theta)$ is continuous for $\theta\in\Theta$, and there is $D(Z)$ with $|a(Z,\theta)|\leq D(Z)$ for all $\theta\in\Theta$ and $\EE(D(Z))<\infty$, then $\EE(a(Z,\theta))$ is continuous and $\sup_{\theta\in\Theta}|n^{-1}\sum_{i=1}^na(Z_i,\theta)-\EE(a(Z,\theta))|\stackrel{p}{\longrightarrow} 0$.
\end{lemma}

\begin{lemma}\label{lemconsistency}
	Under Assumption \ref{reg_mis} (or Assumptions \ref{ass1}), we have $\hat\bbeta\stackrel{p}{\longrightarrow}\bbeta^o$. 
\end{lemma}

\begin{proof}[Proof of Lemma \ref{lemconsistency}]
	The proof of $\hat\bbeta\stackrel{p}{\longrightarrow}\bbeta^o$ follows from Theorem 2.6 in \cite{newey1994large}. Note that their conditions (i)--(iii) follow directly from Assumption \ref{ass1} (1)--(4). We only need to verify their condition (iv), i.e., $\EE(\sup_{\bbeta\in\Theta}|g_{\bbeta j}(T_i, \bX_i)|)<\infty$ where
	$$
	g_{\bbeta j}(T_i,\bX_i)=\Big(\frac{T_i}{\pi_{\bbeta}(\bX_i)}-\frac{1-T_i}{1-\pi_{\bbeta}(\bX_i)}\Big)f_{j}(\bX_{i}),
	$$
	By Assumption \ref{reg_mis} (5), we have $|g_{\bbeta j}(T_i, \bX_i)|\leq 2|f_{j}(\bX_i)|/c_0$ and thus $\EE(\sup_{\bbeta\in\Theta}|g_{\bbeta j}(T_i, \bX_i)|)<\infty$ by Assumption \ref{reg_mis} (6). In addition, for the proof of Theorem \ref{thmdr}, we similarly verify the following conditions to prove this lemma for the oCBPS estimator, i.e., $\EE(\sup_{\bbeta\in\Theta}|g_{1\bbeta j}(T_i, X_i)|)<\infty$ and $\EE(\sup_{\bbeta\in\Theta}|g_{2\bbeta j}(T_i, \bX_i)|)<\infty$, where
	$$
	g_{1\bbeta j}(T_i,\bX_i)=\Big(\frac{T_i}{\pi_{\bbeta}(\bX_i)}-\frac{1-T_i}{1-\pi_{\bbeta}(\bX_i)}\Big)h_{1j}(\bX_i),~~\textrm{and}~~g_{2\bbeta j}(T_i, \bX_i)=\Big(\frac{T_i}{\pi_{\bbeta}(\bX_i)}-1\Big)h_{2j}(\bX_i).
	$$
We have $|g_{1\bbeta j}(T_i, \bX_i)|\leq 2|h_{1j}(\bX_i)|/c_0$ and thus $\EE(\sup_{\bbeta\in\Theta}|g_{1\bbeta j}(T_i, \bX_i)|)<\infty$. Similarly, we can prove $\EE(\sup_{\bbeta\in\Theta}|g_{2\bbeta j}(T_i,\bX_i)|)<\infty$.	
	This completes the proof.	
\end{proof}

\begin{lemma}\label{lemAN}
Under Assumption \ref{reg_mis} (or Assumptions \ref{ass1} and \ref{ass2}), we have
\begin{equation}\label{eqhatbetaexpansion}
n^{1/2}(\hat{\bbeta}-\bbeta^o)=-(\bH_{\fb}^{*\top}\Wb^*\bH_{\fb}^*)^{-1}n^{1/2}\bH_{\fb}^{*\top}\Wb^*\bar\bg_{\bbeta^o}(\bT, \bX)+o_p(1),
\end{equation}
\begin{equation}\label{eqhatbetaAN}
n^{1/2}(\hat{\bbeta}-\bbeta^o)\stackrel{d}{\longrightarrow} N(0,(\bH_{\fb}^{*\top}\Wb^*\bH_{\fb}^*)^{-1}\bH_{\fb}^{*\top}\Wb^*\bOmega\Wb^*\bH_{\fb}^*(\bH_{\fb}^{*\top}\Wb^*\bH_{\fb}^*)^{-1}),
\end{equation}
where $\Omega = \Var(\bg_{\bbeta^*}(T_i, \bX_i))$. If the propensity score model is correctly specified with $\PP(T_i=1\mid \bX_i)=\pi_{\bbeta^o}(\bX_i)$ and $\Wb^*=\bOmega^{-1}$ holds, then $
n^{1/2}(\hat{\bbeta}-\bbeta^o)\stackrel{d}{\longrightarrow} N(0,(\bH_{\fb}^{*\top}\bOmega^{-1}\bH_{\fb}^*)^{-1})$.
\begin{proof}
The proof of (\ref{eqhatbetaexpansion}) and (\ref{eqhatbetaAN}) follows from Theorem 3.4 in \cite{newey1994large}. Note that their conditions (i), (ii), (iii) and (v) are directly implied by our Assumption \ref{reg_mis} (3), (4), (2) and  Assumption \ref{reg_mis} (1), respectively. In addition, their condition (iv), that is,  $\EE(\sup_{\bbeta\in \cN}|\partial \bg_{\bbeta^o}(T_i,\bX_i)/\partial\beta_j|)<\infty$ for some small neighborhood $\cN$ around $\bbeta^o$, is also implied by our Assumption \ref{reg_mis}.	To see this, by Assumption \ref{reg_mis} some simple calculations show that
	$$
	\sup_{\bbeta\in \cN}\Big|\frac{\partial\bg_{\bbeta}(T_i,\bX_i)}{\partial\beta_j}\Big|\leq\Big(\frac{T_i|\fb(\bX_i)|}{c_0^2}+\frac{(1-T_i)|\fb(\bX_i)|}{c_0^2}\Big) \sup_{\bbeta\in \cN}\Big|\frac{\partial\pi_{\bbeta}(\bX_i)}{\partial\beta_j}\Big|\leq C_j(\bX)|\fb(\bX_i)|/c_0^2,
	$$
	for $\cN\in \mathbb{B}_r(\bbeta^o)$. Hence, $\EE(\sup_{\bbeta\in \cN}|\partial \bg_{\bbeta^o}({T_i,\bX_i})/\partial\beta_j|)<\infty$, by Assumption \ref{reg_mis} (7). Thus, condition (iv) in Theorem 3.4 in \cite{newey1994large} holds. In order to apply this lemma to the proofs in Section 3, we need to further verify this condition for
	 $\bg_{\bbeta}(\cdot)=(\bg^{\top}_{1\bbeta}(\cdot),\bg^{\top}_{2\bbeta}(\cdot))^{\top}$, where
	$$
	\bg_{1\bbeta}(T_i,\bX_i)=\Big(\frac{T_i}{\pi_{\bbeta}(\bX_i)}-\frac{1-T_i}{1-\pi_{\bbeta}(\bX_i)}\Big)\bh_{1}(\bX_i),~~\textrm{and}~~\bg_{2\bbeta}(T_i,\bX_i)=\Big(\frac{T_i}{\pi_{\bbeta}(\bX_i)}-1\Big)\bh_{2}(\bX_i).
	$$
	To this end, by Assumption \ref{ass1} some simple calculations show that when
	$$
	\sup_{\bbeta\in \cN}\Big|\frac{\partial\bg_{1\bbeta}(T_i,\bX_i)}{\partial\beta_j}\Big|\leq\Big(\frac{T_i|\bh_{1}(\bX_i)|}{c_0^2}+\frac{(1-T_i)|\bh_{1}(\bX_i)|}{c_0^2}\Big) \sup_{\bbeta\in \cN}\Big|\frac{\partial\pi_{\bbeta}(\bX_i)}{\partial\beta_j}\Big|\leq C_j(\bX)|\bh_{1}(\bX_i)|/c_0^2,
	$$
	for $\cN\in \mathbb{B}_r(\bbeta^o)$. Hence, $\EE(\sup_{\bbeta\in \cN}|\partial \bg_{1\bbeta^o}({T_i,\bX_i})/\partial\beta_j|)<\infty$, by Assumption \ref{ass1} (7). Following the similar arguments, we can prove that $\EE(\sup_{\bbeta\in \cN}|\partial \bg_{2\bbeta^o}({T_i, \bX_i})/\partial\beta_j|)<\infty$ holds.
	This completes the proof of (\ref{eqhatbetaAN}). As shown in Lemma \ref{lemconsistency}, if	$\PP(T_i=1\mid \bX_i)=\pi_{\bbeta^o}(\bX_i)$ holds, the asymptotic normality of $n^{1/2}(\hat{\bbeta}-\bbeta^o)$ follows from (\ref{eqhatbetaAN}). The proof is complete.
\end{proof}
\end{lemma}

\section{Proof of Results in Section \ref{sec:model.misspecification}}

\subsection{Proof of Theorem~\ref{th:asymnormal_mis}}
\label{app:biasmis}
\begin{proof}
First, we derive the bias of $\hat \bbeta$. By the arguments in the proof of Lemma \ref{lemAN}, we can show that $\hat \bbeta =\bbeta^o+O_p(n^{-1/2})$, where $\bbeta^o$ satisfies $\bbeta^o=\argmin_{\bbeta} \EE(\bar\bg_{\bbeta}(\bT,\bX))^\top
\Wb^*\EE(\bar\bg_{\bbeta}(\bT,\bX))$. Let $u_i^* = u(\bX_i; \bbeta^*)$.
By the propensity score model and the fact that $|u(\bX_i; \bbeta^*)|$ is a bounded random variable and $\EE|f_j(\bX_i)|<\infty$, we can show that
$$
\EE(\bar \bg_{\bbeta^o})=\EE\Big\{\frac{\pi_i^*(1+\xi u_i^*)\fb(\bX_i)}{\pi_i^o}-\frac{(1-\pi_i^*-\xi \pi_i^*u_i^*)\fb(\bX_i)}{1-\pi_i^o}\Big\}+O(\xi^2).
$$
In addition, following the similar calculation, we have $\EE(\bar \bg_{\bbeta^*})=O(\xi)$. Therefore, $$\lim_{n\rightarrow\infty}\EE(\bar\bg_{\bbeta^*}(\bT,\bX))^\top
\Wb^*\EE(\bar\bg_{\bbeta^*}(\bT,\bX))=0.$$
Clearly, this quadratic form $\EE(\bar\bg_{\bbeta}(\bT,\bX))^\top \Wb^*\EE(\bar\bg_{\bbeta}(\bT,\bX))$ must be nonnegative for any $\bbeta$. By the uniqueness of $\bbeta^o$, we have $\bbeta^o-\bbeta^*=o(1)$. Therefore, we can expand $\pi_i^o$ around $\pi_i^*$, which yields
$$
\EE(\bar \bg_{\bbeta^o})=\EE\Big\{\xi \Big(\frac{u_i^*}{1-\pi_i^*}\Big)\fb(\bX_i)+\bH_{\fb}^* (\bbeta^o-\bbeta^*)\Big\}+O(\xi^2+\|\bbeta^o-\bbeta^*\|_2^2).
$$
This implies that the bias of $\bbeta^o$ is
\begin{equation}\label{eqbetao}
\bbeta^o-\bbeta^*=-\xi (\bH_{\fb}^{*\top}\Wb^*\bH_{\fb}^*)^{-1}\bH_{\fb}^{*\top}\Wb^*\EE\Big\{\Big(\frac{u_i^*}{1-\pi_i^*}\Big)\fb(\bX_i)\Big\}+O(\xi^2).
\end{equation}
Our next step is to derive the bias of $\hat \mu_{\hat \bbeta}$. Similar to the proof of Theorem \ref{thmAN}, we have
	$$
	\hat{\mu}_{\hat\bbeta}-\mu=\frac{1}{n}\sum_{i=1}^nD_i+\Hb_y^{*\top}(\hat{\bbeta}-\bbeta^o)+o_p(n^{-1/2}),
	$$
	where
	$$
	D_i=\frac{T_iY_i(1)}{\pi_i^o}-\frac{(1-T_i)Y_i(0)}{1-\pi_i^o}-\mu,
	$$
and
$$
n^{1/2}(\hat{\bbeta}-\bbeta^o)=-(\bH_{\fb}^{*\top}\Wb^*\bH_{\fb}^*)^{-1}n^{1/2}\bH_{\fb}^{*\top}\Wb^*\bar\bg_{\bbeta^o}(\bT, \bX)+o_p(1).
$$
In addition, following the similar steps, we can show that $\EE(D_i)=Bn^{-1/2}+o(n^{-1/2})$. Thus,
$$
\hat{\mu}_{\hat\bbeta}-\mu=\frac{1}{n}\sum_{i=1}^n\{D_i-\EE(D_i)\}+\Hb_y^{*\top}(\hat{\bbeta}-\bbeta^o)+Bn^{-1/2}+o_p(n^{-1/2}).
$$
Then the asymptotic normality of $\sqrt{n}(\hat \mu_{\hat\bbeta} - \mu)$ follows from the above asymptotic expansion and the central limit theorem. 
This completes the proof.
\end{proof}

\subsection{Proof of Corollary~\ref{th:localbias}} \label{app:localbias}
\begin{proof}
When $\bH_{\fb}^*$ is invertible, it is easy to show the bias term can be written as
\begin{eqnarray*}
B \ = \ \left[\EE\left\{\frac{u(\bX_i; \bbeta^*)(K(\bX_i)+(1-\pi_{\bbeta^*}(\bX_i)) L(\bX_i))}{1-\pi_{\bbeta^*}(\bX_i)}\right\}+\bH_y^*\bH_{\fb}^{*-1}\EE\left(\frac{u(\bX_i; \bbeta^*)\fb(\bX_i)}{1-\pi_{\bbeta^*}(\bX_i)}\right)\right],
\end{eqnarray*}
when the propensity score model is locally misspecified. If we choose the balancing function $\fb(\bX)$ such that $\balpha^{\top} \fb(\bX) = K(\bX_i)+(1-\pi_i^*) L(\bX_i)$ for some $\balpha \in \RR^q$, we have
\begin{eqnarray*}
\bH_y^*  & = &  -\EE\left(\frac{K(\bX_i)+(1-\pi_i^*)L(\bX_i)}{\pi_i^* (1-\pi_i^*)}\cdot \frac{\partial \pi_i^*}{\partial\bbeta}\right) =- \balpha^{\top}  \EE \Big(\frac{\fb(\bX_i)}{\pi_i^* (1-\pi_i^*)}\Big(\frac{\partial \pi_i^*}{\partial\bbeta}\Big)^{\top}\Big), \\
		\bH_{\fb}^* & = & -\mathbb{E}\left(\frac{\partial g_{\bbeta^*}(T_i, \bX_i)}{\partial \bbeta}\right) = -\EE \left(\frac{\fb(\bX_i)}{\pi_i^*(1-\pi_i^*)}\Big(\frac{\partial \pi_i^*}{\partial \bbeta}\Big)^{\top}\right).
\end{eqnarray*}
So the bias becomes
\begin{eqnarray*}
B &=& \left[\balpha^{\top}\EE\left\{\frac{u(\bX_i; \bbeta^*)\fb(\bX_i)}{1-\pi_{\bbeta^*}(\bX_i)}\right\} + \balpha^{\top} \bH_{\fb}^* (\bH_{\fb}^*)^{-1} \EE\left(\frac{u(\bX_i; \bbeta^*)\fb(\bX_i)}{1-\pi_{\bbeta^*}(\bX_i)}\right) \right] = 0.
\end{eqnarray*}
This proves that $\hat \mu_{\hat \bbeta}$ is first order unbiased.
\end{proof}

\subsection{Proof of Corollary~\ref{th:efficiency}} \label{app:efficiency}
\begin{proof}
Recall that even if the propensity score mode is known or pre-specified, the minimum asymptotic variance over the class of regular estimators is given by $V_{\textrm{opt}}$. In the following, we will verify that with the optimal choice of $\fb(\bX)$ our estimator has asymptotic variance $V_{\textrm{opt}}$.

The asymptotic variance bound $V_{\textrm{opt}}$ can be written as,
$V_{\textrm{opt}}=\Sigma_{\mu}-\balpha^{\top}\bOmega \balpha$, where $$\bOmega=\EE(\bg_{\bbeta^o}(T_i,\bX_i)\bg_{\bbeta^o}(T_i,\bX_i)^\top) = \EE\left(\frac{\fb(\bX_i)\fb(\bX_i)^{\top}}{\pi_i^*(1-\pi_i^*)}\right).$$
We can write the asymptotic variance of our estimator as
\begin{eqnarray*}
V = \Sigma_{\mu} + 2 \bH_y^{*\top} \bSigma_{\mu\bbeta} + \bH_y^{*\top} \bSigma_{\bbeta} \bH_y^*,
\end{eqnarray*}
where
\begin{eqnarray*}
&&\bH_y^* = \EE\left(\frac{\partial \mu_{\bbeta^*}(T_i, Y_i, \bX_i)}{\partial \bbeta}\right) = -\EE \left(\frac{K(\bX_i)+(1-\pi_i^*)L(\bX_i)}{\pi_i^*(1-\pi_i^*)} \frac{\partial \pi_i^*}{\partial \bbeta}\right), \\
&&\bSigma_{\mu\bbeta} = - (\bH_{\fb}^*)^{-1}\Cov(\mu_{\bbeta^*}(T_i, Y_i, \bX_i), \bg_{\bbeta^*}(T_i, \bX_i)),\\
&&\bH_{\fb}^* = \EE\left(\frac{\partial \bg_{\bbeta^*}(T_i,\bX_i)}{\partial \bbeta}\right) = - \EE \left(\frac{\fb(\bX_i)}{\pi_i^*(1-\pi_i^*)}\left(\frac{\partial \pi_i^{*}}{\partial \bbeta}\right)^{\top}\right), \\
&&\Cov(\mu_{\bbeta^*}(T_i,Y_i,\bX_i), \bg_{\bbeta^*}(T_i, \bX_i)) = \EE \left(\frac{K(\bX)+ (1-\pi_i^*)L(\bX_i)}{\pi_i^*(1-\pi_i^*)}\fb(\bX_i)\right),\\
&&\bSigma_{\bbeta} = (\bH_{\fb}^*)^{-1} \Var(\bg_{\bbeta^*}(T_i, \bX_i)) (\bH_{\fb}^{*\top})^{-1},\\
&&\Var(\bg_{\bbeta^*}(T_i,\bX_i)) = \EE \left(\frac{\fb(\bX_i)\fb(\bX_i)^{\top}}{\pi_i^*(1-\pi_i^*)}\right).
\end{eqnarray*}
If $K(\bX_i)+(1-\pi_i^*)L(\bX_i)$ lies in the linear space spanned by $\fb(\bX_i)$, that is, $K(\bX_i)+(1-\pi_i^*)L(\bX_i) = \balpha^{\top} \fb(\bX_i)$, we have
\begin{eqnarray*}
\bH_y^* = - \EE \left(\frac{\balpha^{\top} \fb(\bX_i)}{\pi_i^*(1-\pi_i^*)} \frac{\partial \pi_i^{*}}{\partial \bbeta} \right) = (\balpha^{\top} \bH_{\fb}^*)^{\top}.
\end{eqnarray*}
So
\begin{eqnarray*}
\bH_y^{*\top} \bSigma_{\mu\bbeta} = - \balpha^{\top} \bH_{\fb}^* (\bH_{\fb}^{*})^{-1} \EE \left(\frac{\balpha^{\top}\fb(\bX_i)\fb(\bX_i)}{\pi_i^*(1-\pi_i^*)}\right) = -\balpha^{\top}\EE\left(\frac{\fb(\bX_i)\fb(\bX_i)^{\top}}{\pi_i^*(1-\pi_i^*)}\right)\balpha,
\end{eqnarray*}
and
\begin{eqnarray*}
\bH_y^{*\top} \bSigma_{\bbeta} \bH_y^* = \balpha^{\top} \bH_{\fb}^* (\bH_{\fb}^*)^{-1} \EE \left(\frac{\fb(\bX_i)\fb(\bX_i)^{\top}}{\pi_i^*(1-\pi_i^*)}\right) (\bH_{\fb}^{*\top})^{-1}(\balpha^{\top} \bH_{\fb}^*)^{\top} = \balpha^{\top}\EE\left(\frac{\fb(\bX_i)\fb(\bX_i)^{\top}}{\pi_i^*(1-\pi_i^*)}\right)\balpha.
\end{eqnarray*}
It is seen that $\bH_y^{*\top} \bSigma_{\mu\bbeta} = - \bH_y^{*\top} \bSigma_{\bbeta} \bH_y^*$. Then we have
\begin{eqnarray*}
V = \Sigma_{\mu} - \balpha^{\top} \bOmega \balpha,
\end{eqnarray*}
which corresponds to the minimum asymptotic variance $V_{\textrm{opt}}$.
\end{proof}

\section{Proof of Results in Section \ref{sec:proposed.methodology}}\label{appsec3}

\subsection{Proof of Theorem~\ref{thmdr}}
\begin{proof}[Proof of Theorem \ref{thmdr}]
	We first consider the case (1). That is the propensity score model is correctly specified. By Lemma \ref{lemconsistency}, we have $\hat\bbeta\stackrel{p}{\longrightarrow}\bbeta^o$. Let
	$$
	r_{\bbeta}(T,Y,\bX)=\frac{TY}{\pi_{\bbeta}(\bX)}-\frac{(1-T)Y}{1-\pi_{\bbeta}(\bX)}.
	$$
	It is seen that $|r_{\bbeta}(T,Y,\bX)|\leq 2|Y|/c_0$ and by Assumption \ref{ass1} (6), $\EE|Y|<\infty$. Then Lemma \ref{lemulan} yields $\sup_{\bbeta\in\Theta}|n^{-1}\sum_{i=1}^nr_{\bbeta}(T_i,Y_i,\bX_i)-\EE(r_{\bbeta}(T_i,Y_i,\bX_i))|=o_p(1)$. In addition, by $\hat\bbeta\stackrel{p}{\longrightarrow}\bbeta^o$ and the dominated convergence theorem, we obtain that
	\begin{eqnarray*}
		\hat{\mu}_{\hat \bbeta}&=&\EE\Big(\frac{T_iY_i}{\pi_i^o}-\frac{(1-T_i)Y_i}{1-\pi_i^o}\Big)+o_p(1),
	\end{eqnarray*}
	where $\pi_i^o=\pi_{\bbeta^o}(\bX_i)$.
	Since $Y_i=Y_i(1)T_i+Y_i(0)(1-T_i)$ and $Y_i(1), Y_i(0)$ are independent of $T_i$ given $\bX_i$, we can further simplify the above expression,
	\begin{eqnarray*}
		\hat{\mu}_{\hat \bbeta}&=&\EE\Big(\frac{T_iY_i}{\pi_i^o}-\frac{(1-T_i)Y_i}{1-\pi_i^o}\Big)+o_p(1)=\EE\Big(\frac{T_iY_i(1)}{\pi_i^o}-\frac{(1-T_i)Y_i(0)}{1-\pi_i^o}\Big)+o_p(1)\\
		&=&\EE\Big(\frac{\EE(T_i\mid \bX_i)\EE(Y_i(1)\mid \bX_i)}{\pi_i^o}-\frac{(1-\EE(T_i\mid \bX_i))\EE(Y_i(1)\mid \bX_i)}{1-\pi_i^o}\Big)+o_p(1).
	\end{eqnarray*}
	In addition, if the propensity score model is correctly specified, it further implies
	\begin{eqnarray*}
		\hat{\mu}_{\hat\bbeta}&=&\EE(\EE(Y_i(1)\mid \bX_i)-\EE(Y_i(0)\mid \bX_i))+o_p(1)=\EE(Y_i(1)-Y_i(0))+o_p(1)=\mu+o_p(1).
	\end{eqnarray*}
	This completes the proof of consistence of $\hat{\mu}$ when the propensity score model is correctly specified.
	
	In the following, we consider the case (2). That is $K(\cdot)\in\textrm{span}\{\Mb_1\bh_1(\cdot)\}$ and $L(\cdot)\in\textrm{span}\{\Mb_2\bh_2(\cdot)\}$. By Lemma \ref{lemconsistency}, we have $\hat\bbeta\stackrel{p}{\longrightarrow}\bbeta^o$. The first order condition for $\bbeta^o$ yields $\partial Q(\bbeta^o)/\partial\bbeta=0$, where  $Q(\bbeta)=\EE(\bg^{\top}_{\bbeta})\Wb^*\EE(\bg_{\bbeta})$. By Assumption \ref{ass1} (7) and the dominated convergence theorem, we can interchange the differential with integral, and thus $\Gb^{*\top}\Wb^*\EE(\bg_{\bbeta^o})=0$. Under the assumption that $\PP(T_i=1\mid \bX_i)=\pi(\bX_i)\neq \pi_i^o$, we have
	$$
	\EE(\bg_{1\bbeta^o})=\EE\Big\{\Big(\frac{\pi(\bX_i)}{\pi_i^o}-\frac{1-\pi(\bX_i)}{1-\pi_i^o}\Big)\bh_1(\bX_i)\Big\},
	$$
	$$
	\EE(\bg_{2\bbeta^o})=\EE\Big\{\Big(\frac{\pi(\bX_i)}{\pi_i^o}-1\Big)\bh_2(\bX_i)\Big\}.
	$$
	Rewrite $\Gb^{*\top}\Wb^*=(\Mb_1,\Mb_2)$, where $\Mb_1\in\RR^{q\times m_1}$ and  $\Mb_1\in\RR^{q\times m_2}$. Then, $\bbeta^o$ satisfies
	\begin{equation}\label{eqmh}
	\EE\Big\{\Big(\frac{\pi(\bX_i)}{\pi_i^o}-\frac{1-\pi(\bX_i)}{1-\pi_i^o}\Big)\Mb_1\bh_1(\bX_i)+\Big(\frac{\pi(\bX_i)}{\pi_i^o}-1\Big)\Mb_2\bh_2(\bX_i) \Big\}=0.
	\end{equation}
	Following the similar arguments to that in case (1), we can prove that
	\begin{eqnarray*}
		\hat{\mu}_{\hat\bbeta}&=&\EE\Big(\frac{T_iY_i}{\pi_i^o}-\frac{(1-T_i)Y_i}{1-\pi_i^o}\Big)+o_p(1)\\
		&=&\EE\Big(\frac{\EE(T_i\mid \bX_i)\EE(Y_i(1)\mid \bX_i)}{\pi_i^o}-\frac{(1-\EE(T_i\mid \bX_i))\EE(Y_i(1)\mid \bX_i)}{1-\pi_i^o}\Big)+o_p(1).
	\end{eqnarray*}
	By $\EE(T_i\mid \bX_i)=\pi(\bX_i)$ and the outcome model, it further implies
	\begin{eqnarray*}
		\hat{\mu}_{\hat\bbeta}-\mu&=&\EE\Big\{\frac{\pi(\bX_i)(K(\bX_i)+L(\bX_i))}{\pi_i^o}-\frac{(1-\pi(\bX_i))K(\bX_i)}{1-\pi_i^o}\Big\}-\mu+o_p(1)\\
		&=&\EE\Big\{\Big(\frac{\pi(\bX_i)}{\pi_i^o}-\frac{1-\pi(\bX_i)}{1-\pi_i^o}\Big)K(\bX_i)\Big\}+\EE\Big\{\frac{\pi(\bX_i)L(\bX_i)}{\pi_i^o}\Big\}-\mu+o_p(1)\\
		&=&\EE\Big\{\Big(\frac{\pi(\bX_i)}{\pi_i^o}-\frac{1-\pi(\bX_i)}{1-\pi_i^o}\Big)K(\bX_i)\Big\}+\EE\Big\{\Big(\frac{\pi(\bX_i)}{\pi_i^o}-1\Big)L(\bX_i)\Big\}+o_p(1),
	\end{eqnarray*}
	where in the last step we use $\mu=\EE(L(\bX_i))$. By equation (\ref{eqmh}), we obtain $\hat{\mu}=\mu+o_p(1)$, provided $K(\bX_i)=\balpha_1^{\top}\Mb_1\bh_1(\bX_i)$ and $L(\bX_i)=\balpha_2^{\top}\Mb_2\bh_2(\bX_i)$, where $\balpha_1$ and $\balpha_2$ are $q$-dimensional vectors of constants. This completes the whole proof.

\end{proof}

\subsection{Proof of Theorem \ref{thmAN}}

\begin{proof}[Proof of Theorem \ref{thmAN}]
	We first consider the case (1). That is the propensity score model is correctly specified. By the mean value theorem, we have $\hat{\mu}=\bar{\mu}+\hat\Hb(\tilde\bbeta)^{\top}(\hat{\bbeta}-\bbeta^o)$,
	where
	\begin{eqnarray*}
		\bar{\mu}&=&\frac{1}{n}\sum_{i=1}^n\Big(\frac{T_iY_i}{\pi_i^o}-\frac{(1-T_i)Y_i}{1-\pi_i^o}\Big),~~~\hat\Hb(\tilde\bbeta)=-\frac{1}{n}\sum_{i=1}^n\Big(\frac{T_iY_i}{\tilde\pi_i^2}+\frac{(1-T_i)Y_i}{(1-\tilde\pi_i)^2}\Big)\frac{\partial \tilde{\pi}_i}{\partial\bbeta},
	\end{eqnarray*}
	where $\pi_i^o=\pi_{\bbeta^o}(\bX_i)$, $\tilde\pi_i=\pi_{\tilde\bbeta}(\bX_i)$ and $\tilde{\bbeta}$ is an intermediate value between $\hat{\bbeta}$ and $\bbeta^o$. By Assumption \ref{ass2} (2), we can show that the summand in $\hat\Hb(\tilde\bbeta)$ has a bounded envelop function.  By Lemma \ref{lemulan}, we have $\sup_{\bbeta\in\mathbb{B}_r(\bbeta^o)}|\hat\Hb(\bbeta)-\EE(\hat\Hb(\bbeta))|=o_p(1)$. Since $\hat{\bbeta}$ is consistent, by the dominated convergence theorem we can obtain $\hat\Hb(\tilde\bbeta)=\Hb^*+o_p(1)$, where \begin{eqnarray*}
		\Hb^*&=&-\EE\Big\{\Big(\frac{T_iY_i}{\pi_i^{o2}}+\frac{(1-T_i)Y_i}{(1-\pi_i^o)^{2}}\Big)\frac{\partial \pi_i^o}{\partial\bbeta}\Big\}= -\EE\Big\{\Big(\frac{Y_i(1)}{\pi_i^o}+\frac{Y_i(0)}{1-\pi_i^o}\Big)\frac{\partial \pi_i^o}{\partial\bbeta}\Big\}\\
		&=&-\EE\Big\{\frac{K(\bX_i)+L(\bX_i)(1-\pi_i^o)}{\pi_i^o(1-\pi_i^o)}\frac{\partial \pi_i^o}{\partial\bbeta}\Big\}.
	\end{eqnarray*}
	Finally, we invoke the central limit theorem and equation (\ref{eqhatbetaexpansion}) to obtain that
	$$
	n^{1/2}(\hat{\mu}-\mu)\stackrel{d}{\longrightarrow} N(0, \bar\Hb^{*\top}\bSigma\bar\Hb^*),
	$$
	where $\bar\Hb^*=(1,\Hb^{*\top})^{\top}$, $\bSigma_{\bbeta}=(\Gb^{*\top}\Wb^*\Gb^*)^{-1}\Gb^{*\top}\Wb^*\bOmega\Wb^*\Gb^*(\Gb^{*\top}\Wb^*\Gb^*)^{-1}$ and
	\[ \bSigma= \left( \begin{array}{cc}
	\Sigma_{\mu} & \bSigma^{\top}_{\mu\bbeta}  \\
	\bSigma_{\mu\bbeta} & \bSigma_{\bbeta} \end{array} \right).\]
	Denote $b_i(T_i,\bX_i,Y_i(1),Y_i(0))=T_iY_i(1)/\pi_i^o-(1-T_i)Y_i(0)/(1-\pi_i^o)-\mu$. Here, some simple calculations yield,
	$$
	\Sigma_{\mu}=\EE[b^2_i(T_i,\bX_i,Y_i(1),Y_i(0))]=\EE\Big(\frac{Y_i^2(1)}{\pi_i^o}+\frac{Y_i^2(0)}{1-\pi_i^o}\Big)-\mu^{2}.
	$$
	In addition, the off diagonal matrix can be written as  $\bSigma_{\mu\bbeta}=(\bSigma^{\top}_{1\mu\bbeta},\bSigma^{\top}_{2\mu\bbeta})^{\top}$, where
	$$
	\bSigma_{\mu\bbeta} = -(\Gb^{*\top}\Wb^*\Gb^*)^{-1}\Gb^{*\top}\Wb^*\Tb,
	$$
	where $\Tb=(\EE[\bg^\top_{1\bbeta^o}(T_i, \bX_i)b_i(T_i,\bX_i,Y_i(1),Y_i(0))],\EE[\bg^\top_{2\bbeta^o}(T_i, \bX_i)b_i(T_i,\bX_i,Y_i(1),Y_i(0))])^{\top}$
with
	$$
	\bg_{1\bbeta}(T_i, \bX_i)=\Big(\frac{T_i}{\pi_{\bbeta}(\bX_i)}-\frac{1-T_i}{1-\pi_{\bbeta}(\bX_i)}\Big)\bh_{1}(\bX_i),~~\textrm{and}~~\bg_{2\bbeta}(T_i, \bX_i)=\Big(\frac{T_i}{\pi_{\bbeta}(\bX_i)}-1\Big)\bh_{2}(\bX_i).
	$$
	After some algebra, we can show that
	$$\Tb=\left\{\EE\Big(\frac{K(\bX_i)+(1-\pi_i^o)L(\bX_i)}{(1-\pi_i^o)\pi_i^o}\bh^\top_{1}(\bX_i)\Big), \EE\Big(\frac{K(\bX_i)+(1-\pi_i^o)L(\bX_i)}{\pi_i^o}\bh^\top_{2}(\bX_i)\Big)\right\}^{\top}.
	$$
	This completes the proof of equation (\ref{eqdeltaAN1}). Next, we consider the case (2). Recall that  $\PP(T_i=1\mid \bX_i)=\pi(\bX_i)\neq \pi_{\beta^o}(\bX_i)$. Following the similar arguments, we can show that
	$$
	\hat{\mu}_{\hat\bbeta}-\mu=\frac{1}{n}\sum_{i=1}^nD_i+\Hb^{*\top}(\hat{\bbeta}-\bbeta^o)+o_p(n^{-1/2}),
	$$
	where
	$$
	D_i=\frac{T_iY_i(1)}{\pi_i^o}-\frac{(1-T_i)Y_i(0)}{1-\pi_i^o}-\mu,
	$$
	and
	$$ \Hb^*=-\EE\Big\{\Big(\frac{\pi(\bX_i)(K(\bX_i)+L(\bX_i))}{\pi_i^{o2}}+\frac{(1-\pi(\bX_i))K(\bX_i)}{(1-\pi_i^o)^{2}}\Big)\frac{\partial {\pi}^o_i}{\partial\bbeta}\Big\}.
	$$
	By equation (\ref{eqhatbetaexpansion}) in Lemma \ref{lemAN}, we have that
	$$
	n^{1/2}(\hat{\mu}_{\hat\bbeta}-\mu)\stackrel{d}{\longrightarrow} N(0, \tilde\Hb^{*\top}\tilde\bSigma\tilde\Hb^*),
	$$
	where $\tilde\Hb^*=(1,\Hb^{*\top})^{\top}$, $\bSigma_{\bbeta}=(\Gb^{*\top}\Wb^*\Gb^*)^{-1}\Gb^{*\top}\Wb^*\bOmega\Wb^*\Gb^*(\Gb^{*\top}\Wb^*\Gb^*)^{-1}$ and
	\[ \tilde\bSigma= \left( \begin{array}{cc}
	\Sigma_{\mu} & \tilde\bSigma^{\top}_{\mu\bbeta}  \\
	\tilde\bSigma_{\mu\bbeta} & \tilde\bSigma_{\bbeta} \end{array} \right).\]
	Denote $c_i(T_i,\bX_i,Y_i(1),Y_i(0))=T_iY_i(1)/\pi_i^o-(1-T_i)Y_i(0)/(1-\pi_i^o)-\mu$. As shown in the proof of Theorem \ref{thmdr}, $\EE[b_i(T_i,\bX_i,Y_i(1),Y_i(0))]=0$. Thus,
	\begin{eqnarray*}
		\Sigma_{\mu}&=&\EE[c^2_i(T_i,\bX_i,Y_i(1),Y_i(0))]=\EE\Big(\frac{T_iY_i^2(1)}{\pi_i^{o2}}+\frac{(1-T_i)Y_i^2(0)}{(1-\pi_i^o)^2}\Big)-\mu^{2}\\
		&=&\EE\Big(\frac{\pi(\bX_i)Y_i^2(1)}{\pi_i^{o2}}+\frac{(1-\pi(\bX_i))Y_i^2(0)}{(1-\pi_i^o)^2}\Big)-\mu^{2}.
	\end{eqnarray*}
	Similarly, the off diagonal matrix can be written as  $\tilde\bSigma_{\mu\bbeta}=(\tilde\bSigma^{\top}_{1\mu\bbeta},\tilde\bSigma^{\top}_{2\mu\bbeta})^{\top}$, where
	$$
	\tilde\bSigma_{\mu\bbeta}= -(\Gb^{*\top}\Wb^*\Gb^*)^{-1}\Gb^{*\top}\Wb^* \bS,
	$$
	where $\bS=(\EE[\bg^\top_{1\bbeta^o}(T_i, \bX_i)c_i(T_i,\bX_i,Y_i(1),Y_i(0))]
,\EE[\bg^\top_{2\bbeta^o}(T_i, \bX_i)c_i(T_i,\bX_i,Y_i(1),Y_i(0))])^{\top}$
with
	\begin{equation}\label{eqg}
	\bg_{1\bbeta}(T_i, \bX_i)=\Big(\frac{T_i}{\pi_{\bbeta}(\bX_i)}-\frac{1-T_i}{1-\pi_{\bbeta}(\bX_i)}\Big)\bh_{1}(\bX_i),~~\textrm{and}~~\bg_{2\bbeta}(T_i, \bX_i)=\Big(\frac{T_i}{\pi_{\bbeta}(\bX_i)}-1\Big)\bh_{2}(\bX_i).
	\end{equation}
	After some tedious algebra, we can show that $\bS=(\bS^\top_1,\bS^\top_2)^{\top}$,	
	where
	$$
	\bS_1=\EE\Big\{\Big(\frac{\pi(\bX_i)(K(\bX_i)+L(\bX_i)-\pi^o_i\mu)}{\pi_i^{o2}}+\frac{(1-\pi(\bX_i))(K(\bX_i)+(1-\pi^o_i)\mu)}{(1-\pi_i^{o})^2}\Big)\bh_{1}(\bX_i)\Big\},
	$$
	$$
	\bS_2=\EE\Big\{\Big(\frac{\pi(\bX_i)[(K(\bX_i)+L(\bX_i))(1-\pi^o_i)-\pi^o_i\mu]}{\pi_i^{o2}}+\frac{(1-\pi(\bX_i))K(\bX_i)+(1-\pi^o_i)\mu}{1-\pi_i^{o}}\Big)\bh_{2}(\bX_i)\Big\}.
	$$
	This completes the proof of equation (\ref{eqdeltaAN2}).
	
	Finally, we start to prove part 3.
	By (\ref{eqdeltaAN1}), the asymptotic variance of $\hat{\mu}$ denoted by $V$, can be written as
	\begin{equation}\label{eqV}
	V=\Sigma_{\mu}+2\Hb^{*\top}\bSigma_{\mu\bbeta}+\Hb^{*\top}\bSigma_{\bbeta}\Hb^{*}.
	\end{equation}
	Note that by Lemma \ref{lemAN}, we have $\bSigma_{\bbeta}=(\Gb^{*\top}\bOmega^{-1}\Gb^*)^{-1}$. Under this correctly specified   propensity score model, some algebra yields
	\[ \bOmega= \EE[\bg_{\bbeta^o}(T_i, \bX_i)\bg^{\top}_{\bbeta^o}(T_i, \bX_i)]=\left( \begin{array}{cc}
	\EE(\frac{\bh_1\bh_1^{\top}}{\pi_i^o(1-\pi_i^o)}) & \EE(\frac{\bh_1\bh_2^{\top}}{\pi_i^o})  \\
	\EE(\frac{\bh_2\bh^{\top}_1}{\pi_i^o})  & \EE(\frac{\bh_2\bh_2^{\top}(1-\pi_i^o)}{\pi_i^o})  \end{array} \right),\]
	where $\bg_{\bbeta}(T_i, \bX_i)=(\bg^{\top}_{1\bbeta}(T_i, \bX_i),\bg^{\top}_{2\bbeta}(T_i, \bX_i))^{\top}$ and $\bg_{1\bbeta}(T_i, \bX_i)$ and $\bg_{2\bbeta}(T_i, \bX_i)$ are defined in (\ref{eqg}). In addition, $\Gb^*=(\Gb^{*\top}_1,\Gb^{*\top}_{2})^{\top}$, where
	\begin{equation}\label{eqgb12}
	\Gb_1^*=-\EE\Big(\frac{\bh_1(\bX_i)}{\pi_i^o(1-\pi_i^o)}\Big(\frac{\partial\pi_i^o}{\partial\bbeta}\Big)^{\top}\Big),~~\Gb_2^*=-\EE\Big(\frac{\bh_2(\bX_i)}{\pi_i^o}\Big(\frac{\partial\pi_i^o}{\partial\bbeta}\Big)^{\top}\Big).
	\end{equation}
	Since the functions $\bK(\cdot)$ and $\bL(\cdot)$ lie in the linear space spanned by the functions $\Mb_1\bh_1(\cdot)$ and $\Mb_2\bh_2(\cdot)$ respectively, where $\Mb_1\in\RR^{q\times m_1}$ and $\Mb_1\in\RR^{q\times m_2}$ are the partitions of $\Gb^{*\top}\Wb^*=(\Mb_1,\Mb_2)$. We have  $K(\bX_i)=\balpha_1^{\top}\Mb_1\bh_1(\bX_i)$ and $L(\bX_i)=\balpha_2^{\top}\Mb_2\bh_2(\bX_i)$, where $\balpha_1$ and $\balpha_2$ are $q$-dimensional vectors of constants. Thus
	\begin{eqnarray*}
		\Hb^*&=&-\EE\Big\{\frac{K(\bX_i)+L(\bX_i)(1-\pi_i^o)}{\pi_i^{o}(1-\pi_i^o)}\frac{\partial {\pi}^o_i}{\partial\bbeta}\Big\}\\
		&=&-\EE\Big\{\frac{\balpha_1^{\top}\Mb_1\bh_1(\bX_i)+\balpha_2^{\top}\Mb_2\bh_2(\bX_i)(1-\pi_i^o)}{\pi_i^{o}(1-\pi_i^o)}\frac{\partial {\pi}^o_i}{\partial\bbeta}\Big\}.
	\end{eqnarray*}
	Comparing to the expression of $\Gb^*$ in (\ref{eqgb12}), we can rewrite $\Hb^*$ as
	\[\Hb^*=\Gb^{*\top} \left( \begin{array}{c}
	\Mb_1^{\top}\balpha_1  \\
	\Mb_2^{\top}\balpha_2  \end{array} \right).\]
	Following the similar derivations, it is seen that
	\[\bSigma_{\mu\bbeta}=-(\Gb^{*\top}\bOmega^{-1}\Gb^{*})^{-1}\Gb^{*\top}\bOmega^{-1} \left( \begin{array}{c}
	\EE\{\frac{\balpha_1^{\top}\Mb_1\bh_1(\bX_i)+\balpha_2^{\top}\Mb_2\bh_2(\bX_i)(1-\pi_i^o)}{\pi_i^{o}(1-\pi_i^o)}\bh_1(\bX_i)\}  \\
	\EE\{\frac{\balpha_1^{\top}\Mb_1\bh_1(\bX_i)+\balpha_2^{\top}\Mb_2\bh_2(\bX_i)(1-\pi_i^o)}{\pi_i^{o}}\bh_2(\bX_i)\}  \end{array} \right),\]
	which is equivalent to
	\[\bSigma_{\mu\bbeta}=-(\Gb^{*\top}\bOmega^{-1}\Gb^{*})^{-1}\Gb^{*\top} \left( \begin{array}{c}
	\Mb_1^{\top}\balpha_1  \\
	\Mb_2^{\top}\balpha_2  \end{array} \right).\]
	Hence,
	\[\Hb^{*\top}\bSigma_{\mu\bbeta}=-(\balpha_1^{\top}\Mb_1,\balpha_2^{\top}\Mb_2)\Gb^*(\Gb^{*\top}\bOmega^{-1}\Gb^{*})^{-1}\Gb^{*\top} \left( \begin{array}{c}
	\Mb_1^{\top}\balpha_1  \\
	\Mb_2^{\top}\balpha_2  \end{array} \right)=-\Hb^{*\top}\bSigma_{\bbeta}\Hb^{*}.\]
	Together with (\ref{eqV}), we have
	\[V=\Sigma_{\mu}-(\balpha_1^{\top}\Mb_1,\balpha_2^{\top}\Mb_2)\Gb^*(\Gb^{*\top}\bOmega^{-1}\Gb^{*})^{-1}\Gb^{*\top} \left( \begin{array}{c}
	\Mb_1^{\top}\balpha_1  \\
	\Mb_2^{\top}\balpha_2  \end{array} \right).\]
	This completes of the proof.

\end{proof}
\subsection{Proof of Corollary \ref{coraddmoment}}

\begin{proof}[Proof of Corollary \ref{coraddmoment}]
	By Theorem \ref{thmAN}, it suffices to show that
	\begin{equation}\label{eqcoraddmoment1}
	(\bar\balpha_1^{\top}\bar\Mb_1,\bar\balpha_2^{\top}\bar\Mb_2)\bar\Gb^*\bar\Cb\bar\Gb^{*\top} \left( \begin{array}{c}
	\bar\Mb_1^{\top}\bar\balpha_1  \\
	\bar\Mb_2^{\top}\bar\balpha_2  \end{array} \right)\leq (\balpha_1^{\top}\Mb_1,\balpha_2^{\top}\Mb_2)\Gb^*\Cb\Gb^{*\top} \left( \begin{array}{c}
	\Mb_1^{\top}\balpha_1  \\
	\Mb_2^{\top}\balpha_2  \end{array} \right),
	\end{equation}
	where $\Cb=(\Gb^{*\top}\bOmega^{-1}\Gb^{*})^{-1}$ and $\bar\balpha_1$ and $\bar\Mb_1$ among others are the corresponding quantities with $\bar \bh_1(\bX)$ and $\bar \bh_2(\bX)$. Assume that $\bar \bh_1(\bX)\in\RR^{m_1+a_1}$ and $\bar \bh_2(\bX)\in\RR^{m_2+a_2}$.
	Since $K(\bX_i)=\balpha_1^{\top}\Mb_1\bh_1(\bX_i)$ and $L(\bX_i)=\balpha_2^{\top}\Mb_2\bh_2(\bX_i)$, we find that $(\bar\balpha_1^{\top}\bar\Mb_1,\bar\balpha_2^{\top}\bar\Mb_2)=(\balpha_1^{\top}\Mb_1,0,\balpha_2^{\top}\Mb_2,0)$, which is a vector in $\RR^{m+a}$ with $a=a_1+a_2$.  Because some components of $(\bar\balpha_1^{\top}\bar\Mb_1,\bar\balpha_2^{\top}\bar\Mb_2)$ are $0$, by the matrix algebra, (\ref{eqcoraddmoment1}) holds if $\Cb-\bar\Cb$ is positive semidefinite. Without loss of generality, we rearrange orders and write the $(m+a)\times q$ matrix $\bar{\Gb}^*$ and the $(m+a)\times (m+a)$ matrix $\bar{\bOmega}^*$ as
	\begin{equation*}
	\bar{\Gb}^*=\Big( \begin{array}{c}
	\Gb^*  \\
	\Ab  \end{array} \Big), ~~\textrm{and}~~\bar{\bOmega}=\Big( \begin{array}{cc}
	\bOmega & \bOmega_1 \\
	\bOmega_1 & \bOmega_2  \end{array} \Big).
	\end{equation*}
	For simplicity, we use the following notation: two matrices satisfy $\Ob_1\geq \Ob_2$ if $\Ob_1- \Ob_2$ is positive semidefinite. To show $\Cb\geq \bar\Cb$, we have the following derivation
	\begin{eqnarray*}
		\bar\Gb^{*\top}\bar\bOmega^{-1}\bar\Gb^{*}&=&(\Gb^{*\top},\Ab^\top)\Big( \begin{array}{cc}
			\bOmega & \bOmega_1 \\
			\bOmega_1 & \bOmega_2  \end{array} \Big)^{-1}\Big( \begin{array}{c}
			\Gb^*  \\
			\Ab  \end{array} \Big)\\
		&\geq& (\Gb^{*\top},\Ab^\top)\Big( \begin{array}{cc}
			\bOmega^{-1} & 0 \\
			0 & 0  \end{array} \Big)\Big( \begin{array}{c}
			\Gb^*  \\
			\Ab  \end{array} \Big)=\Gb^{*\top}\bOmega^{-1}\Gb^{*}.
	\end{eqnarray*}
	This completes the proof of (\ref{eqcoraddmoment1}), and therefore the corollary holds.
	
\end{proof}

\subsection{Proof of Corollary \ref{cormq}}
\begin{proof}[Proof of Corollary \ref{cormq}]
	The proof of the double robustness property mainly follows from Theorem \ref{thmdr}. In this case, we only need to verify that $\textrm{span}\{\bh_1(\cdot)\}=\textrm{span}\{\Mb_1\bh_1(\cdot)\}$ and $\textrm{span}\{\bh_2(\cdot)\}=\textrm{span}\{\Mb_2\bh_2(\cdot)\}$, where $\Mb_1\in\RR^{q\times m_1}$ and $\Mb_1\in\RR^{q\times m_2}$ are the partitions of $\Gb^{*\top}\Wb^*=(\Mb_1,\Mb_2)$. Apparently, we have  $\textrm{span}\{\Mb_1\bh_1(\cdot)\}\subseteq\textrm{span}\{\bh_1(\cdot)\}$, since the former can always be written as a linear combination of $\bh_1(\cdot)$. To show  $\textrm{span}\{\bh_1(\cdot)\}\subseteq\textrm{span}\{\Mb_1\bh_1(\cdot)\}$, note that the $m_1\times m_1$ principal submatrix $\Mb_{11}$ of $\Mb_1$ is invertible. Thus, $\textrm{span}\{\bh_1(\cdot)\}=\textrm{span}\{\Mb_{11}\bh_1(\cdot)\}\subseteq\textrm{span}\{\Mb_1\bh_1(\cdot)\}$. This is because the $m_1$ dimensional functions $\Mb_{11}\bh_1(\cdot)$ are identical to the first $m_1$ coordinates of $\Mb_{1}\bh_1(\cdot)$. This completes the proof of double robustness property. The efficiency property follows from  Theorem \ref{thmAN}. We do not replicate the details.
\end{proof}
%

\section{Regularity Conditions in Section \ref{secnon}}
\begin{assumption}\label{assnon} The following regularity
  conditions are assumed.
  \begin{enumerate}
  \item
  The minimizer
    $\bbeta^o=\argmin_{\bbeta\in\Theta} \|\EE(\bar\bg_{\bbeta}(\bT,\bX))\|^2_2$ is unique.
  \item  $\bbeta^o\in\textrm{int}(\Theta)$, where $\Theta$ is a
    compact set.
  \item There exist constants $0<c_0<1/2$, $c_1>0$ and $c_2>0$ such that $c_0\leq J(v)\leq 1-c_0$ and $0< c_1\leq \partial J(v)/\partial v\leq c_2$,  for any $v=\bbeta^\top\bB(\bx)$ with $\bbeta\in \textrm{int}(\Theta)$. There exists a small neighborhood of $v^*=\bbeta^{*\top}\bB(\bx)$, say $\mathcal{B}$ such that for any $v\in\mathcal{B}$ it holds that $|\partial^2 J(v)/\partial v^2|\leq c_3$ for some constant $c_3>0$.
  \item   $\EE|Y(1)|^2<\infty$ and $\EE|Y(0)|^2<\infty$.
  \item Let $\Gb^*:=\EE[\bB(\bX_i)\bh(\bX_i)^\top\bDelta_i(\psi^*(\bX_i))]$, where $\bDelta_i(\psi(\bX_i))=\textrm{diag}(\xi_i(\psi(\bX_i)){\bf 1}_{m_1}, \phi_i(\psi(\bX_i)){\bf 1}_{m_2})$ is a $\kappa\times \kappa$ diagonal matrix with
\begin{align*}
\xi_i(\psi(\bX_i))&=-\Big(\frac{T_i}{J^2(\psi(\bX_i))}+\frac{1-T_i}{(1-J(\psi(\bX_i)))^2}\Big)\frac{\partial J(\psi(\bX_i))}{\partial \psi}, \\
\phi_i(\psi(\bX_i))&=-\frac{T_i}{J^2(\psi(\bX_i))}\frac{\partial J(\psi(\bX_i))}{\partial \psi}.
\end{align*}
Here, ${\bf 1}_{m_1}$ is a vector of $1$'s with length $m_1$.  Assume that there exists a constant $C_1>0$, such that $\lambda_{\min}(\Gb^{*\top}\Gb^*)\geq C_1$, where $\lambda_{\min}(\cdot)$ denotes the minimum eigenvalue of a matrix.
 \item  For some constant $C$, it holds $\|\EE[\bh(\bX_i)\bh(\bX_i)^\top]\|_2\leq C$ and $\|\EE[\bB(\bX_i)\bB(\bX_i)^\top]\|_2\leq C$, where $\|\Ab\|_2$ denotes the spectral norm of the matrix $\Ab$. In addition, $\sup_{\bx\in\mathcal{X}}\|\bh(\bx)\|_2\leq C\kappa^{1/2}$, and $\sup_{\bx\in\mathcal{X}}\|\bB(\bx)\|_2\leq C\kappa^{1/2}$.
 \item Let $m^*(\cdot)\in\mathcal{M}$ and $K(\cdot), L(\cdot)\in\mathcal{H}$, where $\mathcal{M}$ and $\mathcal{H}$ are two sets of smooth functions. Assume that $\log N_{[~]}(\epsilon, \mathcal{M}, L_2(P))\leq C(1/\epsilon)^{1/k_1}$ and $\log N_{[~]}(\epsilon, \mathcal{H}, L_2(P))\leq C(1/\epsilon)^{1/k_2}$, where $C$ is a positive constant and $k_1, k_2>1/2$. Here, $N_{[~]}(\epsilon, \mathcal{M}, L_2(P))$ denotes the minimum number of $\epsilon$-brackets needed to cover $\mathcal{M}$; see  Definition 2.1.6 of \cite{van1996weak}.
  \end{enumerate}
\end{assumption}
Note that the first five conditions are similar to
Assumptions~\ref{ass1}~and~\ref{ass2}. In particular, Condition~5 is
the natural extension of Condition~1 of Assumption~\ref{ass2}, when
the dimension of the matrix $\Gb^*$ grows with the sample size
$n$. Condition~6 is a mild technical condition on the basis functions
$\bh(\bx)$ and $\bB(\bx)$, which is implied by Assumption 2 of
\cite{newey1997convergence}. In particular, this condition is
satisfied by many bases such as the regression spline, trigonometric
polynomial, wavelet bases; see
\cite{newey1997convergence,horowitz2004nonparametric,chen2007large,belloni2015some}. Finally,
Condition~7 is a technical condition on the complexity of the function
classes $\mathcal{M}$ and $\mathcal{H}$. Specifically, it requires
that the bracketing number $N_{[~]}(\epsilon, \cdot, L_2(P))$ of
$\mathcal{M}$ and $\mathcal{H}$ cannot increase too fast as $\epsilon$
approaches to 0. This condition holds for many commonly used function
classes. For instance, if $\mathcal{M}$ corresponds to the H\"older
class with smoothness parameter $s$ defined on a bounded convex subset
of $\RR^d$, then
$\log N_{[~]}(\epsilon, \mathcal{M}, L_2(P))\leq C(1/\epsilon)^{d/s}$
by Corollary 2.6.2 of \cite{van1996weak}. Hence, this condition simply
requires $s/d>1/2$. Given Assumption~\ref{assnon}, the following
theorem establishes the asymptotic normality and semiparametric
efficiency of the estimator $\tilde{\mu}_{\tilde\bbeta}$.

\section{Proof of Results in Section \ref{secnon}}

For notational simplicity, we denote $\pi^*(\bx)=J(m^*(\bx))$, $J^*(\bx)=J(\bbeta^{*\top}\bB(\bx))$,  and $\tilde J(\bx)=J(\tilde\bbeta^{\top}\bB(\bx))$. Define $Q_n(\bbeta)=\|\bar\bg_{\bbeta}(\bT, \bX)\|_2^2$ and $Q(\bbeta)=\|\EE\bg_{\bbeta}(\bT_i, \bX_i)\|_2^2$. In the following proof, we use $C, C'$ and $C''$ to denote generic positive constants, whose values may change from line to line. In this section, denote $K=\kappa$ and $\psi(\bX)=m(\bX)$.

\begin{lemma}[Bernstein's inequality  for $U$-statistics \citep{arcones1995bernstein}]\label{lemustat}
Given   i.i.d. random variables $Z_1,\ldots Z_n$ taking values in a measurable space $(\mathbb{S},\mathcal{B})$ and a symmetric and measurable kernel function  $h\colon\mathbb{S}^m \rightarrow R$, we define the $U$-statistics with kernel $h$ as
$
U \coloneqq {n \choose m}^{-1}\sum_{i_1<\ldots<i_m} h(Z_{i_1},\ldots, Z_{i_m}).
$
Suppose that $\EE h(Z_{i_1},\ldots, Z_{i_m}) = 0$, $\EE \bigl\{\EE[h(Z_{i_1},\ldots, Z_{i_m})\mid Z_{i_1}]\bigr\}^2 = \sigma^2$ and $\|h\|_{\infty} \leq b$.  There exists a constant $K(m)>0$ depending on $m$ such that
$$
\PP( |U| > t) \leq 4 \exp\big\{-nt^2/ [2m^2\sigma^2 + K(m) bt]\big\},~\forall t>0.
$$
\end{lemma}

\begin{lemma}\label{lemnon0}
Under the conditions in Theorem \ref{thmnon}, it holds that
$$
\sup_{\bbeta\in\Theta}\Big|Q_n(\bbeta)-Q(\bbeta)\Big|=O_p\Big(\sqrt{\frac{K^2\log K}{n}}\Big).
$$
\end{lemma}

\begin{proof}[Proof of Lemma \ref{lemnon0}]
Let $\bxi(\bbeta)=(\xi_1(\bbeta),...,\xi_n(\bbeta))^\top$ and $\bphi(\bbeta)=(\phi_1(\bbeta),...,\phi_n(\bbeta))^\top$, where
$$\xi_i(\bbeta)=\frac{T_i}{J(\bbeta^\top\bB(\bX_i))}-\frac{1-T_i}{1-J(\bbeta^\top\bB(\bX_i))},~~\phi_i(\bbeta)=\frac{T_i}{J(\bbeta^\top\bB(\bX_i))}-1.
$$
Then we have
\begin{align*}
Q_n(\bbeta)&=n^{-2}\sum_{i=1}^n\sum_{j=1}^n\big[\xi_i(\bbeta)\xi_j(\bbeta)\bh_1(\bX_i)^\top\bh_1(\bX_j)+\phi_i(\bbeta)\phi_j(\bbeta)\bh_2(\bX_i)^\top\bh_2(\bX_j)\big]\\
&=n^{-2}\sum_{1\leq i\neq j\leq n}\big[\xi_i(\bbeta)\xi_j(\bbeta)\bh_1(\bX_i)^\top\bh_1(\bX_j)+\phi_i(\bbeta)\phi_j(\bbeta)\bh_2(\bX_i)^\top\bh_2(\bX_j)\big]+A_n(\bbeta),
\end{align*}
where $A_n(\bbeta)=n^{-2}\sum_{i=1}^n\big[\xi_i(\bbeta)^2\|\bh_1(\bX_i)\|^2_2+\phi_i(\bbeta)^2\|\bh_2(\bX_i)\|_2^2\big]$. Since there exists a constant $c_0>0$ such that $c_0\leq |J(\bbeta^\top\bB(\bx))|\leq 1-c_0$ for any $\bbeta\in\Theta$ and $T_i\in\{0,1\}$, it implies that
$\sup_{\bbeta\in\Theta}\max_{1\leq i\leq n}|\xi_i(\bbeta)|\leq C$ and $\sup_{\bbeta\in\Theta}\max_{1\leq i\leq n}|\phi_i(\bbeta)|\leq C$ for some constant $C>0$. Then we can show that
\begin{align*}
\EE\Big(\sup_{\bbeta\in\Theta}|A_n(\bbeta)|\Big)\leq \frac{C}{n}\EE(\|\bh(\bX_i)\|^2_2)=O(K/n).
\end{align*}
By the Markov inequality, we have $\sup_{\bbeta\in\Theta}|A_n(\bbeta)|=O_p(K/n)=o_p(1)$. Following the similar arguments, it can be easily shown that $\sup_{\bbeta\in\Theta}|Q(\bbeta)|/n=O(K/n)$. Thus, it holds that
\begin{align}
\sup_{\bbeta\in\Theta}|Q_n(\bbeta)-Q(\bbeta)|&=\sup_{\bbeta\in\Theta}\Big|\frac{2}{n(n-1)}\sum_{1\leq i<j\leq n}u_{ij}(\bbeta)\Big|+O_p(K/n),\label{eqlemnon00}
\end{align}
where $u_{ij}(\bbeta)=u_{1ij}(\bbeta)+u_{2ij}(\bbeta)$ is a kernel function of a U-statistic with
\begin{align*}
u_{1ij}(\bbeta)&=\xi_i(\bbeta)\xi_j(\bbeta)\bh_1(\bX_i)^\top\bh_1(\bX_j)-\EE[\xi_i(\bbeta)\xi_j(\bbeta)\bh_1(\bX_i)^\top\bh_1(\bX_j)],\\
u_{2ij}(\bbeta)&=\phi_i(\bbeta)\phi_j(\bbeta)\bh_2(\bX_i)^\top\bh_2(\bX_j)-\EE[\phi_i(\bbeta)\phi_j(\bbeta)\bh_2(\bX_i)^\top\bh_2(\bX_j)].
\end{align*}
Since $\Theta$ is a compact set in $\RR^K$, by the covering number theory, there exists a constant $C$ such that $M=(C/r)^K$ balls with the radius $r$ can cover $\Theta$. Namely, $\Theta\subseteq \cup_{1\leq m\leq M} \Theta_m$, where $\Theta_m=\{\bbeta\in\RR^K: \|\bbeta-\bbeta_m\|_2\leq r\}$ for some $\bbeta_1,...,\bbeta_M$. Thus, for any given $\epsilon>0$,
\begin{align}
\PP\Big(\sup_{\bbeta\in\Theta}&\Big|\frac{2}{n(n-1)}\sum_{1\leq i<j\leq n}u_{1ij}(\bbeta)\Big|>\epsilon\Big)\leq \sum_{m=1}^M\PP\Big(\sup_{\bbeta\in\Theta_m}\Big|\frac{2}{n(n-1)}\sum_{1\leq i<j\leq n}u_{1ij}(\bbeta)\Big|>\epsilon\Big)\nonumber\\
&\leq \sum_{m=1}^M\Big[\PP\Big(\Big|\frac{2}{n(n-1)}\sum_{1\leq i<j\leq n}u_{1ij}(\bbeta_m)\Big|>\epsilon/2\Big)\nonumber\\
&\quad\quad\quad+\PP\Big(\sup_{\bbeta\in\Theta_m}\frac{2}{n(n-1)}\sum_{1\leq i<j\leq n}\Big|u_{1ij}(\bbeta)-u_{1ij}(\bbeta_m)\Big|>\epsilon/2\Big)\Big].\label{eqlemnon01}
\end{align}
By the Cauchy-Schwarz inequality, $|\bh_1(\bX_i)^\top\bh_1(\bX_j)|\leq \|\bh_1(\bX_i)\|_2\|\bh_1(\bX_j)\|_2\leq CK$, and thus $|u_{1ij}(\bbeta_m)|\leq CK$.  In addition, for any $\bbeta$,
\begin{align*}
&\EE\big\{\xi_i(\bbeta)\bh_1(\bX_i)^\top\EE[\xi_j(\bbeta)\bh_1(\bX_j)]-\EE[\xi_i(\bbeta)\xi_j(\bbeta)\bh_1(\bX_i)^\top\bh_1(\bX_j)]\big\}^2\\
&\leq \EE\big\{\xi_i(\bbeta)\bh_1(\bX_i)^\top\EE[\xi_j(\bbeta)\bh_1(\bX_j)]\big\}^2\leq \|\EE\xi^2_i(\bbeta)\bh_1(\bX_i)\bh_1(\bX_i)^\top\|_2\cdot\|\EE\xi_j(\bbeta)\bh_1(\bX_j)\|_2^2\leq CK,
\end{align*}
for some constant $C>0$. Here, in the last step we use that fact that
$$
\|\EE\xi_j(\bbeta)\bh_1(\bX_j)\|_2^2\leq \EE\|\xi_j(\bbeta)\bh_1(\bX_j)\|_2^2\leq C\cdot \EE\|\bh_1(\bX_j)\|_2^2\leq CK,$$
and $\|\EE\xi^2_i(\bbeta)\bh_1(\bX_i)\bh_1(\bX_i)^\top\|_2$ is bounded because $\|\EE\bh_1(\bX_j)\bh_1(\bX_j)^\top\|_2$ is bounded by assumption. Thus, we can apply the Bernstein's inequality in Lemma \ref{lemustat} to the U-statistic with kernel function $u_{1ij}(\bbeta_m)$,
\begin{align}
\PP\Big(\Big|\frac{2}{n(n-1)}\sum_{1\leq i<j\leq n}u_{1ij}(\bbeta_m)\Big|>\epsilon/2\Big)\leq 2\exp\big(-Cn\epsilon^2/[K+K\epsilon]\big),\label{eqlemnon02}
\end{align}
for some constant $C>0$. Since $|\partial J(v)/\partial v|$ is upper bounded by a constant for any $v=\bbeta^\top\bB(\bx)$, it is easily seen that for any $\bbeta\in\Theta_m$, $|\xi_i(\bbeta)-\xi_i(\bbeta_m)|\leq C|(\bbeta-\bbeta_m)^\top\bB(\bX_i)|\leq CrK^{1/2}$, where the last step follows from the Cauchy-Schwarz inequalty. This further implies $|\xi_i(\bbeta)\xi_j(\bbeta)-\xi_i(\bbeta_m)\xi_j(\bbeta_m)|\leq  CrK^{1/2}$ for some constant $C>0$ by performing a standard perturbation analysis. Thus,
\begin{align*}
|u_{1ij}(\bbeta)-u_{1ij}(\bbeta_m)|&\leq CrK^{1/2}|\bh_1(\bX_i)^\top\bh_1(\bX_j)|\leq CrK^{3/2},
\end{align*}
and note that with $r=K^{-2}$, then $CrK^{1/2}\EE|\bh_1(\bX_i)^\top\bh_1(\bX_j)|\leq \epsilon/4$ for $n$ large enough. Thus
\begin{align}
&\PP\Big(\sup_{\bbeta\in\Theta_m}\frac{2}{n(n-1)}\sum_{1\leq i<j\leq n}\Big|u_{1ij}(\bbeta)-u_{1ij}(\bbeta_m)\Big|>\epsilon/2\Big)\nonumber\\
&\leq \PP\Big(\frac{2CrK^{1/2}}{n(n-1)}\sum_{1\leq i<j\leq n}|\bh_1(\bX_i)^\top\bh_1(\bX_j)|>\epsilon/2\Big)\nonumber\\
&\leq \PP\Big(\frac{2CrK^{1/2}}{n(n-1)}\sum_{1\leq i<j\leq n}\big[|\bh_1(\bX_i)^\top\bh_1(\bX_j)|-\EE|\bh_1(\bX_i)^\top\bh_1(\bX_j)|\big]>\epsilon/4\Big)\nonumber\\
&\leq 2\exp(-CnK\epsilon^2),\label{eqlemnon03}
\end{align}
where the last step follows from the Hoeffding inequality for U-statistic. Thus, combining (\ref{eqlemnon01}), (\ref{eqlemnon02}) and (\ref{eqlemnon03}), we have for some constants $C_1, C_2, C_3>0$, as $n$ goes to infinity,
\begin{align*}
\PP\Big(\sup_{\bbeta\in\Theta}&\Big|\frac{2}{n(n-1)}\sum_{1\leq i<j\leq n}u_{1ij}(\bbeta)\Big|>\epsilon\Big)\\
&\leq \exp(C_1K\log K-C_2n\epsilon^2/[K+K\epsilon])+\exp(C_1K\log K-C_3n\epsilon^2K)\rightarrow 0,
\end{align*}
where we take $\epsilon=C\sqrt{K^2\log K/n}$ for some constant $C$ sufficiently large. This implies
$$
\sup_{\bbeta\in\Theta}\Big|\frac{2}{n(n-1)}\sum_{1\leq i<j\leq n}u_{1ij}(\bbeta)\Big|=O_p\Big(\sqrt{\frac{K^2\log K}{n}}\Big).
$$
Following the same arguments, we can show that with the same choice of $\epsilon$,
$$
\sup_{\bbeta\in\Theta}\Big|\frac{2}{n(n-1)}\sum_{1\leq i<j\leq n}u_{2ij}(\bbeta)\Big|=O_p\Big(\sqrt{\frac{K^2\log K}{n}}\Big).
$$
Plugging these results into (\ref{eqlemnon00}), we complete the proof.
\end{proof}

\begin{lemma}[Bernstein's inequality for random matrices \citep{tropp2015introduction}]\label{lemtropp}
Let $\{\Zb_k\}$ be a sequence of independent random matrices with dimensions $d_1\times d_2$. Assume that $\EE\Zb_k={\bf 0}$ and $\|\Zb_k\|_2\leq R_n$ almost sure. Define
$$
\sigma^2_n=\max\Big\{\Big\|\sum_{k=1}^n\EE(\Zb_k\Zb_k^\top)\Big\|_2, \Big\|\sum_{k=1}^n\EE(\Zb_k^\top\Zb_k)\Big\|_2\Big\}.
$$
Then, for all $t\geq 0$,
$$
\PP\Big(\Big\|\sum_{k=1}^n\Zb_k\Big\|_2\geq t\Big)\leq (d_1+d_2)\exp\Big(-\frac{t^2/2}{\sigma_n^2+R_nt/3}\Big).
$$
\end{lemma}

\begin{lemma}\label{lemnon2}
Let $\Hb=(\bh(\bX_1),...,\bh(\bX_n))^\top$ and $\Bb=(\bB(\bX_1),...,\bB(\bX_n))^\top$ be two $n\times K$ matrices. Under the conditions in Theorem \ref{thmnon}, then
\begin{equation}\label{eqlemnon21}
\|\Hb^\top\Hb/n-\EE[\bh(\bX_i)\bh(\bX_i)^\top]\|_2=O_p(\sqrt{K\log K/n})
\end{equation}
and
\begin{equation}\label{eqlemnon22}
\|\Bb^\top\Bb/n-\EE[\bB(\bX_i)\bB(\bX_i)^\top]\|_2=O_p(\sqrt{K\log K/n}).
\end{equation}
\end{lemma}

\begin{proof}[Proof of Lemma \ref{lemnon2}]
We prove this result by applying Lemma \ref{lemtropp}. In particular, to prove (\ref{eqlemnon21}), we take $\Zb_i=n^{-1}[\bh(\bX_i)\bh(\bX_i)^\top-\EE(\bh(\bX_i)\bh(\bX_i)^\top)]$. It is easily seen that
$$
\|\Zb_i\|_2\leq n^{-1}[\textrm{tr}(\bh(\bX_i)\bh(\bX_i)^\top)+\|\EE(\bh(\bX_i)\bh(\bX_i)^\top)\|_2]\leq (CK+C)/n,
$$
where $C$ is some positive constant. Moreover,
\begin{align*}
\Big\|\sum_{i=1}^n\EE(\Zb_i\Zb_i^\top)\Big\|_2&\leq n^{-1}\Big(\|\EE\bh(\bX_i)\bh(\bX_i)^\top\bh(\bX_i)\bh(\bX_i)^\top\|_2+\|\EE(\bh(\bX_i)\bh(\bX_i)^\top)\|_2^2\Big)\\
&\leq n^{-1}(CK\cdot\|\EE(\bh(\bX_i)\bh(\bX_i)^\top)\|_2+C^{2})\leq n^{-1}(C^2K+C^{2}).
\end{align*}
Note that $\sqrt{K\log K/n}=o(1)$. Now, if we take $t=C\sqrt{K\log K/n}$ in Lemma \ref{lemtropp} for some constant $C$ sufficiently large, then we have $\PP(\|\sum_{k=1}^n\Zb_k\|_2\geq t)\leq 2K\exp(-C'\log K)$ for some $C'>1$. Then, the right hand side converges to $0$, as $K\rightarrow\infty$. This completes the proof of  (\ref{eqlemnon21}). The proof of (\ref{eqlemnon22}) follows from the same arguments and is omitted for simplicity.
\end{proof}

\begin{lemma}\label{lemnon1}
Under the conditions in Theorem \ref{thmnon}, the following results hold.
\begin{itemize}
\item[1] Let $\bar\bU=\frac{1}{n}\sum_{i=1}^n\bU_i$, $\bU_i=(\bU^\top_{i1},\bU^\top_{i2})^\top$, with
$$
\bU_{i1}=	\Big(\frac{T_i}{\pi^*_i}-\frac{1-T_i}{1-\pi^*_i}\Big)\bh_1(\bX_i),~~\bU_{i2}= \Big(\frac{T_i}{\pi^*_i}-1\Big)\bh_2(\bX_i).
$$
Then $\|\bar\bU\|_2=O_p(K^{1/2}/n^{1/2})$.
\item[2] Let $\mathbb{B}(r)=\{\bbeta\in\RR^{K}: \|\bbeta-\bbeta^*\|_2\leq r\}$, and $r=O(K^{1/2}/n^{1/2}+K^{-r_b})$. Then
$$
\sup_{\bbeta\in\mathbb{B}(r)}\Big\|\frac{\partial\bar\bg_{\bbeta}(\bT,\bX)}{\partial\bbeta}-\Gb^*\Big\|_2=O_p\Big(K^{1/2}r+\sqrt{\frac{K\log K}{n}}\Big).
$$
\item[3] Let $J_i=J(\bbeta^\top\bB(\bX_i))$, $\dot{J}_i=\partial J(v)/\partial v|_{v=\bbeta^\top\bB(\bX_i)}$, and
$$
\Tb^*=\EE\Big\{\Big[\frac{\EE(Y_i(1)\mid \bX_i)}{\pi^{*}_i}-\frac{\EE(Y_i(0)\mid \bX_i)}{1-\pi^*_i}\Big]\dot{J}^*_i\bB(\bX_i)\Big\}.
$$
Then
$$
\sup_{\bbeta\in\mathbb{B}(r)}\Big\|\frac{1}{n}\sum_{i=1}^n \Big[\frac{T_iY_i(1)}{J^{2}_i}+\frac{(1-T_i)Y_i(0)}{(1-J_i)^2}\Big]\dot{J}_i\bB(\bX_i)+\Gb^{*\top}\balpha^*\Big\|_2=O_p\Big(K^{1/2}r+K^{-r_h}\Big).
$$
\end{itemize}
\end{lemma}

\begin{proof}[Proof of Lemma \ref{lemnon1}]
We start from the proof of the first result. Note that $\EE (\bU_i)=0$. Then $\EE\|\bar\bU\|_2^2=\EE(\bU_i^\top\bU_i)/n$ and then there exists some constant $C>0$,
\begin{align*}
\EE\|\bar\bU\|_2^2&=\EE\Big[n^{-1}\sum_{k=1}^K\Big(\frac{T_i}{\pi^*_i}-\frac{1-T_i}{1-\pi^*_i}\Big)^2 h_k(\bX_i)^2I(k\leq m_1)+\Big(\frac{T_i}{\pi^*_i}-1\Big)^2 h_k(\bX_i)^2I(k>m_1)\Big]\\
&\leq C \sum_{k=1}^K\EE\{h_{k}(\bX_i)^2\}/n=O(K/n).
\end{align*}
By the Markov inequality, this implies $\|\bar\bU\|_2=O_p(K^{1/2}/n^{1/2})$, which completes the proof of the first result. In the following, we prove the second result. Denote
\begin{align*}
\xi_i(m(\bX_i))&=-\Big(\frac{T_i}{J^2(m(\bX_i))}+\frac{1-T_i}{(1-J(m(\bX_i)))^2}\Big)\dot{J}(m(\bX_i))\\
\phi_i(m(\bX_i))&=-\frac{T_i}{J^2(m(\bX_i))}\dot{J}(m(\bX_i)),
\end{align*}
and $\bDelta_i(m(\bX_i))=\textrm{diag}(\xi_i(m(\bX_i)){\bf 1}_{m_1}, \phi_i(m(\bX_i)){\bf 1}_{m_2})$ is a $K\times K$ diagonal matrix, where ${\bf 1}_{m_1}$ is a vector of $1$ with length $m_1$. Then, note that
\begin{align*}
\frac{\partial\bar\bg_{\bbeta}(\bT,\bX)}{\partial\bbeta}-\Gb^*&=\frac{1}{n}\sum_{i=1}^n\bB(\bX_i)\bh(\bX_i)^\top\bDelta_i(\bbeta^\top\bB(\bX_i))-\EE[\bB(\bX_i)\bh(\bX_i)^\top\bDelta_i(m^*(\bX_i))],
\end{align*}
which can be decomposed into the two terms $I_{\bbeta}+II$, where
\begin{align*}
I_{\bbeta}&=\frac{1}{n}\sum_{i=1}^n\bB(\bX_i)\bh(\bX_i)^\top[\bDelta_i(\bbeta^\top\bB(\bX_i))-\bDelta_i(m^*(\bX_i))],~~II=\sum_{i=1}^n\Zb_i,\\
\Zb_i&=n^{-1}\Big\{\bB(\bX_i)\bh(\bX_i)^\top\bDelta_i(m^*(\bX_i))-\EE[\bB(\bX_i)\bh(\bX_i)^\top\bDelta_i(m^*(\bX_i))]\Big\}.
\end{align*}
We first consider the term II. It can be easily verified that $\|\bDelta_i(m^*(\bX_i))\|_2\leq C$ for some constant $C>0$. In addition,
$\|\bB(\bX_i)\bh(\bX_i)^\top\|_2\leq \|\bB(\bX_i)\|_2\cdot\|\bh(\bX_i)\|_2\leq CK$. Thus, $\|\Zb_i\|_2\leq CK/n$. Following the similar argument in the proof of Lemma \ref{lemnon2},
\begin{align*}
\Big\|\sum_{i=1}^n\EE(\Zb_i\Zb_i^\top)\Big\|_2&\leq n^{-1}\|\EE\bB(\bX_i)\bh(\bX_i)^\top\bDelta_i(m^*(\bX_i))\bDelta_i(m^*(\bX_i))\bh(\bX_i)\bB(\bX_i)^\top\|_2\\
&~~~~~~+n^{-1}\|\EE\bB(\bX_i)\bh(\bX_i)^\top\bDelta_i(m^*(\bX_i))\|_2^2.
\end{align*}
We now consider the last two terms separately. Note that
\begin{align}
&\|\EE\bB(\bX_i)\bh(\bX_i)^\top\bDelta_i(m^*(\bX_i))\|^2_2=\sup_{\|\ub\|_2=1,\|\vb\|_2=1} |\EE\ub^\top\bB(\bX_i)\bh(\bX_i)^\top\bDelta_i(m^*(\bX_i))\vb|^2\nonumber\\
&\leq \sup_{\|\ub\|_2=1} |\EE\ub^\top\bB(\bX_i)\bB(\bX_i)^\top\ub|\cdot \sup_{\|\vb\|_2=1} |\EE\vb^\top\bDelta_i(m^*(\bX_i))\bh(\bX_i)\bh(\bX_i)^\top\bDelta_i(m^*(\bX_i))\vb|\nonumber\\
&\leq \|\EE(\bB(\bX_i)\bB(\bX_i)^\top)\|_2\cdot C\|\EE(\bh(\bX_i)\bh(\bX_i)^\top)\|_2\leq C',\label{eqlemnon11}
\end{align}
where $C, C'$ are some positive constants. Following the similar arguments to (\ref{eqlemnon11}),
\begin{align*}
&\|\EE\bB(\bX_i)\bh(\bX_i)^\top\bDelta_i(m^*(\bX_i))\bDelta_i(m^*(\bX_i))\bh(\bX_i)\bB(\bX_i)^\top\|_2\\
&\leq CK\cdot \sup_{\|\ub\|_2=1} |\EE\ub^\top\bB(\bX_i)\bB(\bX_i)^\top\ub|\leq CK\cdot\|\EE\bB(\bX_i)\bB(\bX_i)^\top\|_2\leq C'K,
\end{align*}
for some constants $C,C'>0$. This implies $\|\sum_{i=1}^n\EE(\Zb_i\Zb_i^\top)\|_2\leq CK/n$.  Thus,  Lemma \ref{lemtropp} implies $\|II\|_2=O_p(\sqrt{K\log K/n})$. Next, we consider the term $I_{\bbeta}$. Following the similar arguments to (\ref{eqlemnon11}), we can show that
\begin{align*}
\sup_{\bbeta\in\mathbb{B}(r)}\|I_{\bbeta}\|_2&=\sup_{\bbeta\in\mathbb{B}(r)}\sup_{\|\ub\|_2=1,\|\vb\|_2=1} \Big|\frac{1}{n}\sum_{i=1}^n\ub^\top\bB(\bX_i)\bh(\bX_i)^\top[\bDelta_i(\bbeta^\top\bB(\bX_i))-\bDelta_i(m^*(\bX_i))]\vb\Big|\\
&\leq \Big\|\frac{1}{n}\sum_{i=1}^n\bB(\bX_i)\bB(\bX_i)^\top\Big\|^{1/2}_2\cdot \Big\|\frac{1}{n}\sum_{i=1}^n\bh(\bX_i)\bh(\bX_i)^\top\Big\|^{1/2}_2\\
&\quad\quad \cdot \sup_{\bbeta\in\mathbb{B}(r)}\max_{1\leq i\leq n}\|\bDelta_i(\bbeta^\top\bB(\bX_i))-\bDelta_i(m^*(\bX_i))\|_2\\
&\leq C\sup_{\bbeta\in\mathbb{B}(r)}\sup_{\bx\in\mathcal{X}}|(\bbeta^*-\bbeta)^\top\bB(\bx)|+ C\sup_{\bx\in\mathcal{X}}|m^*(\bx)-\bbeta^{*\top}\bB(\bx)|\\
&\leq C'(K^{1/2}r+K^{-r_b})\leq C''K^{1/2}r,
\end{align*}
for some $C, C', C''>0$, where the second inequality follows from Lemma \ref{lemnon2} and the Lipschitz property of  $\xi_i(\cdot)$ and $\phi_i(\cdot)$, and the third inequality is due to the Cauchy-Schwarz inequality and approximation assumption of the sieve estimator. This completes the proof of the second result. For the third result, let
\begin{align*}
\eta_i(m(\bX_i))&=\Big(\frac{T_iY_i(1)}{J^2(m(\bX_i))}+\frac{(1-T_i)Y_i(0)}{(1-J(m(\bX_i)))^2}\Big)\dot{J}(m(\bX_i)).
\end{align*}
Thus, the following decomposition holds,
\begin{align*}
\frac{1}{n}\sum_{i=1}^n\eta_i(\bbeta^\top\bB(\bX_i))\bB(\bX_i)+\Gb^{*\top}\balpha^*=T_{1\bbeta}+T_2+T_3,
\end{align*}
where
\begin{align*}
T_{1\bbeta}&=\frac{1}{n}\sum_{i=1}^n[\eta_i(\bbeta^\top\bB(\bX_i))-\eta_i(m^*(\bB(\bX_i)))]\bB(\bX_i)\\
T_2&=\frac{1}{n}\sum_{i=1}^n\Big[\eta_i(m^*(\bB(\bX_i)))\bB(\bX_i)-\EE\eta_i(m^*(\bB(\bX_i)))\bB(\bX_i)\Big]\\
T_3&=\EE\eta_i(m^*(\bB(\bX_i)))\bB(\bX_i)+\Gb^{*\top}\balpha^*.
\end{align*}
Similar to the proof for $\sup_{\bbeta\in\mathbb{B}(r)}\|I_{\bbeta}\|_2$ previously, we can easily show that
$\sup_{\bbeta\in\mathbb{B}(r)}\|T_{1\bbeta}\|_2=O_p(K^{1/2}r)$. Again, the key step is to use the results from Lemma \ref{lemnon2}.
For the second term $T_2$, we can use the similar arguments in the proof of the first result to show that  $\EE\|T_2\|^2_2\leq CK\cdot \EE[\eta_i(m^*(\bB(\bX_i))^2]/n=O(K/n)$. The Markov inequality implies $\|T_2\|_2=O_p(K^{1/2}/n^{1/2})$. For the third term $T_3$, after some algebra, we can show that
$$
\|T_3\|_2\leq C\Big(\sup_{\bx\in\mathcal{X}}|K(\bx)-\balpha_1^{*\top}\bh_1(\bx)|+\sup_{\bx\in\mathcal{X}}|L(\bx)-\balpha_2^{*\top}\bh_2(\bx)|\Big)=O_p(K^{-r_h}).
$$
Combining the $L_2$ error bound for $T_{1\bbeta}$, $T_2$ and $T_3$, we obtain the last result. This completes the whole proof.
\end{proof}

\begin{lemma}\label{lemnonconsistency}
Under the conditions in Theorem \ref{thmnon}, it holds that
$$
\|\tilde\bbeta-\bbeta^*\|_2=o_p(1).
$$
\end{lemma}

\begin{proof}[Proof of Lemma \ref{lemnonconsistency}]
Recall that $\bbeta^o$ is the minimizer of $Q(\bbeta)$. We now decompose $Q(\tilde\bbeta)-Q(\bbeta^o)$ as
\begin{equation}\label{eqlemnonrate2}
Q(\tilde\bbeta)-Q(\bbeta^o)=\underbrace{[Q(\tilde\bbeta)-Q_n(\tilde\bbeta)]}_{I}+\underbrace{[Q_n(\tilde\bbeta)-Q_n(\bbeta^o)]}_{II}+\underbrace{[Q_n(\bbeta^o)-Q(\bbeta^o)]}_{III}.
\end{equation}
In the following, we study the terms I, II and III one by one. For the term I, Lemma \ref{lemnon0} implies
$|Q(\tilde\bbeta)-Q_n(\tilde\bbeta)|\leq \sup_{\bbeta\in\Theta}\Big|Q_n(\bbeta)-Q(\bbeta)\Big|=o_p(1)$.
This shows that $|I|=o_p(1)$ and the same argument yields $|III|=o_p(1)$. For the term II, by the definition of $\tilde\bbeta$, it is easy to see that $II\leq 0$. Thus, combining with (\ref{eqlemnonrate2}), we have for any constant $\eta>0$ to be chosen later,  $Q(\tilde\bbeta)-Q(\bbeta^o)<\eta$ with probability tending to one. For any $\epsilon>0$, define $E_\epsilon=\Theta\cap\{\|\bbeta-\bbeta^o\|_2\geq \epsilon\}$. By the uniqueness of $\bbeta^o$, for any $\bbeta\in E_\epsilon$, we have $Q(\bbeta)>Q(\bbeta^o)$. Since $E_\epsilon$ is a compact set, we have $\inf_{\bbeta\in E_\epsilon} Q(\bbeta)>Q(\bbeta^o)$. This implies that for any $\epsilon>0$, there exists $\eta'>0$ such that $Q(\bbeta)>Q(\bbeta^o)+\eta'$ for any $\bbeta\in E_\epsilon$. If $\tilde\bbeta\in E_\epsilon$, then $Q(\bbeta^o)+\eta>Q(\tilde\bbeta)>Q(\bbeta^o)+\eta'$ with probability tending to one. Apparently, this does not holds if we take $\eta<\eta'$. Thus, we have proved that  $\tilde\bbeta\notin E_\epsilon$, that is $\|\tilde\bbeta-\bbeta^o\|_2\leq \epsilon$ for any $\epsilon>0$. Thus, we have $\|\tilde\bbeta-\bbeta^o\|_2=o_p(1)$. %

Next, we shall show that $\|\bbeta^o-\bbeta^*\|_2=o_p(1)$. It is easily seen that these together lead to the desired consistency result
$$\|\tilde\bbeta-\bbeta^*\|_2\leq \|\bbeta^o-\bbeta^*\|_2+\|\tilde\bbeta-\bbeta^o\|_2=o_p(1).
$$
To show $\|\bbeta^o-\bbeta^*\|_2=o_p(1)$, we use the similar strategy. That is we want to show that for any constant $\eta>0$,  $Q(\bbeta^*)-Q(\bbeta^o)<\eta$. In the following, we prove that $Q(\bbeta^*)=O(K^{1-2r_b})$. Note that
$$
Q(\bbeta^*)\leq C^2K^{-2r_b}\sum_{j=1}^{K}\EE|\bh_j(\bX)|^2=O(K^{1-2r_b}),
$$
where the first inequality follows from the Cauchy-Schwarz inequality and the last step uses the assumption that $\sup_{\bx\in\mathcal{X}}\|\bh(\bx)\|_2=O(K^{1/2})$.
In addition, it holds that $Q(\bbeta^o)\leq Q(\bbeta^*)=O(K^{1-2r_b})$. As $K\rightarrow\infty$, it yields $Q(\bbeta^*)-Q(\bbeta^o)<\eta$, for any constant $\eta>0$. The same arguments yield $\|\bbeta^o-\bbeta^*\|_2=o_p(1)$. This completes the proof of the consistency result.
\end{proof}

\begin{lemma}\label{lemnonrate}
Under the conditions in Theorem \ref{thmnon}, there exists a global minimizer $\tilde\bbeta$ (if $Q_n(\bbeta)$ has multiple minimizers), such that
\begin{equation}\label{eqlemnonrate0}
\|\tilde\bbeta-\bbeta^*\|_2=O_p(K^{1/2}/n^{1/2}+K^{-r_b}).
\end{equation}
\end{lemma}

\begin{proof}[Proof of Lemma \ref{lemnonrate}]
We first prove that there exists a local minimizer $\tilde\bDelta$ of $Q_n(\bbeta^*+\bDelta)$, such that $\tilde\Delta\in\mathcal{C}$, where $\mathcal{C}=\{\bDelta\in\RR^K: \|\bDelta\|_2\leq r\}$,  and  $r=C(K^{1/2}/n^{1/2}+K^{-r_b})$ for some constant $C$ large enough. To this end, it suffices to show that
\begin{equation}\label{eqlemnonrate1}
\PP\Big\{\inf_{\bDelta\in\partial\mathcal{C}} Q_n(\bbeta^*+\bDelta)-Q_n(\bbeta^*)>0\Big\}\rightarrow 1, ~~\textrm{as}~n\rightarrow\infty,
\end{equation}
where $\partial\mathcal{C}=\{\bDelta\in\RR^K: \|\bDelta\|_2=r\}$. Applying the mean value theorem to each component of $\bar\bg_{\bbeta^*+\bDelta}(\bT, \bX)$,
$$
\bar\bg_{\bbeta^*+\bDelta}(\bT, \bX)=\bar\bg_{\bbeta^*}(\bT, \bX)+\widetilde\Gb\bDelta,
$$
where $\widetilde\Gb=\frac{\partial \bar\bg_{\bar\bbeta}(\bT, \bX)}{\partial\bbeta}$ and for notational simplicity we assume there exists a common $\bar\bbeta=v\bbeta^*+(1-v)\tilde\bbeta$ for some $0\leq v\leq 1$ lies between $\bbeta^*$ and $\bbeta^*+\bDelta$ (Rigorously speaking, we need different $\bar\bbeta$ for different component of $\bar\bg_{\bbeta^*+\bDelta}(\bT, \bX)$). Thus, for any $\bDelta\in\partial\mathcal{C}$,
\begin{align}
Q_n(\bbeta^*+\bDelta)-Q_n(\bbeta^*)&=2\bar\bg_{\bbeta^*}(\bT, \bX)\widetilde\Gb\bDelta+\bDelta^\top(\widetilde\Gb^\top\widetilde\Gb)\bDelta\nonumber\\
&\geq-2\|\bar\bg_{\bbeta^*}(\bT, \bX)\|_2\cdot\|\widetilde\Gb\|_2\cdot\|\bDelta\|_2+\|\bDelta\|_2^2\cdot\lambda_{\min}(\widetilde\Gb^\top\widetilde\Gb)\nonumber\\
&\geq-C(K^{1/2}/n^{1/2}+K^{-r_b})\cdot r+C\cdot r^2,\label{eqlemnonrate3}
\end{align}
for some constant $C>0$. In the last step, we first use the results that $\|\bar\bg_{\bbeta^*}(\bT, \bX)\|_2=O_p(K^{1/2}/n^{1/2}+K^{-r_b})$, which is derived by combining Lemma \ref{lemnon1} with the arguments similar to (\ref{lemnonexpansion3}) in the proof of Lemma \ref{lemnonexpansion}. In addition, $\|\widetilde\Gb\|_2\leq \|\widetilde\Gb-\Gb^*\|_2+\|\Gb^*\|_2\leq C$, since $\|\Gb^*\|_2$ is bounded by a constant and $\|\widetilde\Gb-\Gb^*\|_2=o_p(1)$ by Lemma  \ref{lemnon1}. By the Weyl inequality and Lemma  \ref{lemnon1},
\begin{align*}
\lambda_{\min}(\widetilde\Gb^\top\widetilde\Gb)&\geq\lambda_{\min}(\Gb^{*\top}\Gb^*)-\|\widetilde\Gb^\top\widetilde\Gb-\Gb^{*\top}\Gb^*\|_2\\
&\geq C-\|\widetilde\Gb-\Gb^*\|_2\cdot\|\widetilde\Gb\|_2-\|\widetilde\Gb-\Gb^*\|_2\cdot \|\Gb^*\|_2\geq C/2,
\end{align*}
for $n$ sufficiently large. By (\ref{eqlemnonrate3}), if $r=C(K^{1/2}/n^{1/2}+K^{-r_b})$ for some constant $C$ large enough, the right hand side is positive for $n$ large enough. This establishes (\ref{eqlemnonrate1}). Next, we show that $\tilde\bbeta=\bbeta^*+\tilde\bDelta$ is a global minimizer of $Q_n(\bbeta)$. This is true because the first order condition implies
$$
\Big(\frac{\partial \bar\bg_{\tilde\bbeta}(\bT, \bX)}{\partial\bbeta}\Big)\bar\bg_{\tilde\bbeta}(\bT, \bX)=0, ~~\Longrightarrow~~\bar\bg_{\tilde\bbeta}(\bT, \bX)=0,
$$
provided $\partial \bar\bg_{\tilde\bbeta}(\bT, \bX)/\partial\bbeta$ is invertible. Following the similar arguments by applying the Weyl inequality, $\partial \bar\bg_{\tilde\bbeta}(\bT, \bX)/\partial\bbeta$ is invertible with probability tending to one. Since $\bar\bg_{\tilde\bbeta}(\bT, \bX)=0$, it implies $Q_n(\tilde\bbeta)=0$. Noting that $Q_n(\bbeta)\geq 0$ for any $\bbeta$, we obtain that $\tilde\bbeta$ is indeed a global minimizer of $Q_n(\bbeta)$.
\end{proof}

\begin{lemma}\label{lemnonexpansion}
Under the conditions in Theorem \ref{thmnon}, $\tilde\bbeta$ satisfies the following asymptotic expansion
\begin{equation}\label{lemnonexpansion1}
\tilde\bbeta-\bbeta^*=-\Gb^{-1}\bar\bU+\bDelta_{n},
\end{equation}
where $\bar\bU=\frac{1}{n}\sum_{i=1}^n\bU_i$, $\bU_i=(\bU^\top_{i1},\bU^\top_{i2})^\top$, with
$$
\bU_{i1}=	\left(\frac{T_i}{\pi^*_i}-\frac{1-T_i}{1-\pi^*_i}\right)\bh_1(\bX_i),~~\bU_{i2}= \left(\frac{T_i}{\pi^*_i}-1\right)\bh_2(\bX_i),
$$
and
$$
\|\bDelta_n\|_2=O_p\Big(K^{1/2}\cdot\Big(\frac{K^{1/2}}{n^{1/2}}+\frac{1}{K^{r_b}}\Big)^2+\sqrt{\frac{K\log K}{n}}\cdot \Big(\frac{K^{1/2}}{n^{1/2}}+\frac{1}{K^{r_b}}\Big)\Big).
$$
\end{lemma}

\begin{proof}[Proof of Lemma \ref{lemnonexpansion}]
Similar to the proof of Lemma \ref{lemnonrate}, we apply the mean value theorem to each component of $\bar\bg_{\tilde\bbeta}(\bT, \bX)$,
$$
\bar\bg_{\bbeta^*}(\bT, \bX)+\Big(\frac{\partial \bar\bg_{\bar\bbeta}(\bT, \bX)}{\partial\bbeta}\Big)(\tilde\bbeta-\bbeta^*)=0,
$$
where for notational simplicity we assume there exists a common $\bar\bbeta=v\bbeta^*+(1-v)\tilde\bbeta$ for some $0\leq v\leq 1$ lies between $\bbeta^*$ and $\tilde\bbeta$. After rearrangement, we derive
\begin{align}
\tilde\bbeta-\bbeta^*&=-\Gb^{*-1}\bar\bg_{\bbeta^*}(\bT, \bX)+\Big[\Gb^{*-1}-\Big(\frac{\partial \bar\bg_{\bar\bbeta}(\bT, \bX)}{\partial\bbeta}\Big)^{-1}\Big]\bar\bg_{\bbeta^*}(\bT, \bX)\nonumber\\
&=-\Gb^{*-1}\bar\bU+\bDelta_{n1}+\bDelta_{n2}+\bDelta_{n3},\label{lemnonexpansion2}
\end{align}
where
$$
\bDelta_{n1}=\Gb^{*-1}[\bar \bU-\bar\bg_{\bbeta^*}(\bT, \bX)], ~~~\bDelta_{n2}=\Big[\Gb^{*-1}-\Big(\frac{\partial \bar\bg_{\bar\bbeta}(\bT, \bX)}{\partial\bbeta}\Big)^{-1}\Big]\bar \bU
$$
and
$$
\bDelta_{n3}=\Big[\Gb^{*-1}-\Big(\frac{\partial \bar\bg_{\bar\bbeta}(\bT, \bX)}{\partial\bbeta}\Big)^{-1}\Big]\cdot[\bar\bg_{\bbeta^*}(\bT, \bX)-\bar \bU].
$$
We first consider $\bDelta_{n1}$ in (\ref{lemnonexpansion2}). Let $\bxi=(\xi_1,...,\xi_n)^\top$, where
$$\xi_i=T_i\Big(\frac{1}{\pi^*_i}-\frac{1}{J^*_i}\Big)-(1-T_i)\Big(\frac{1}{1-\pi^*_i}-\frac{1}{1-J^*_i}\Big), ~~\textrm{for}~ 1\leq i\leq m_1,
$$
and
$$\xi_i=T_i\Big(\frac{1}{\pi^*_i}-\frac{1}{J^*_i}\Big), ~~\textrm{for}~m_1+1\leq i\leq K.
$$
Let $\Hb=(\bh(X_1),...,\bh(X_n))^\top$ be a $n\times K$ matrix. Then, for some constants $C,C'>0$,
\begin{align}
\|\bDelta_{n1}\|^2_2&= n^{-2}\bxi^\top\Hb\Gb^{*-1}\Gb^{*-1}\Hb^\top\bxi\leq n^{-2}\|\bxi\|^2_2\cdot \|\Hb\Gb^{*-1}\Gb^{*-1}\Hb^\top\|_2\nonumber\\
&\leq Cn^{-1}\|\bxi\|^2_2\cdot \|\Hb^\top\Hb/n\|_2\leq C'n^{-1}\|\bxi\|^2_2,\label{lemnonexpansion3}
\end{align}
where the third step follows from the fact that $\|\Gb^{*-1}\|_2$ is bounded and the last step follows from Lemma \ref{lemnon2} and the maximum eigenvalue of $\EE[\bh(\bX_i)\bh(\bX_i)^\top]$ is bounded. Since $|\partial J(v)/\partial v|$ is upper bounded by a constant for any $v\leq \sup_{\bx\in\mathcal{X}}|m^*(\bx)|$, then there exist some constants $C,C'>0$, suc that for any $m_1+1\leq i\leq K$,
$$
|\xi_i|\leq C|\pi_i^*-J_i^*|\leq C'\sup_{\bx\in\mathcal{X}}|m^*(\bx)-\bbeta^{*\top}\bB(\bx)|\leq C'K^{-r_b}.
$$
Similarly, $|\xi_i|\leq 2C'K^{-r_b}$ for any $1\leq i\leq m_1$. Thus, it yields $n^{-1}\|\bxi\|^2_2=O_p(K^{-2r_b})$. Combining with (\ref{lemnonexpansion3}), we conclude that $\|\bDelta_{n1}\|_2=O_p(K^{-r_b})$.

Next, we consider $\bDelta_{n2}$. Since $\|\Gb^{*-1}\|_2$ is bounded, we have
\begin{align*}
\|\bDelta_{n2}\|_2&\leq \|\Gb^{*-1}\|_2 \cdot\Big\|\Big(\frac{\partial \bar\bg_{\bar\bbeta}(\bT, \bX)}{\partial\bbeta}\Big)^{-1}\Big\|_2\cdot \Big\|\Gb^*-\frac{\partial \bar\bg_{\bar\bbeta}(\bT, \bX)}{\partial\bbeta}\Big\|_2\cdot\|\bar \bU\|_2\\
&\leq C\Big(\|\tilde\bbeta-\bbeta^*\|_2K^{1/2}+\sqrt{\frac{K\log K}{n}}\Big)\cdot \sqrt{\frac{K}{n}},
\end{align*}
where the last step follows from Lemma \ref{lemnon1}.

Finally, we consider $\bDelta_{n3}$. By the same arguments in the control of terms $\bDelta_{n1}$ and $\bDelta_{n2}$, we can prove that
\begin{align*}
\|\bDelta_{n3}\|_2&\leq \Big\|\Gb^{*-1}-\Big(\frac{\partial \bar\bg_{\bar\bbeta}(\bT, \bX)}{\partial\bbeta}\Big)^{-1}\Big\|_2\cdot\|\bar\bg_{\bbeta^*}(\bT, \bX)-\bar \bU\|_2\\
&\leq C\Big(\|\tilde\bbeta-\bbeta^*\|_2K^{1/2}+\sqrt{\frac{K\log K}{n}}\Big)\cdot K^{-r_b}.
\end{align*}
Combining the rates of $\|\bDelta_{n1}\|_2$, $\|\bDelta_{n2}\|_2$ and $\|\bDelta_{n3}\|_2$ with (\ref{lemnonexpansion2}), by Lemma \ref{lemnon1}, we obtain
\begin{align*}
\|\tilde\bbeta-\bbeta^*\|_2&\leq \|\Gb^{*-1}\bar\bg_{\bbeta^*}(\bT, \bX)\|_2+\|\bDelta_{n1}\|_2+\|\bDelta_{n2}\|_2+\|\bDelta_{n3}\|_2\\
&\leq C\Big(\frac{K^{1/2}}{n^{1/2}}+\frac{1}{K^{r_b}}\Big)+C'\Big(\|\tilde\bbeta-\bbeta^*\|_2K^{1/2}+\sqrt{\frac{K\log K}{n}}\Big)\cdot \Big(\frac{K^{1/2}}{n^{1/2}}+\frac{1}{K^{r_b}}\Big),
\end{align*}
for some constants $C,C'>0$. 
Therefore, (\ref{lemnonexpansion1}) holds with $\bDelta_n=\bDelta_{n1}+\bDelta_{n2}+\bDelta_{n3}$, where
$$
\|\bDelta_n\|_2=O_p\Big(K^{1/2}\cdot\Big(\frac{K^{1/2}}{n^{1/2}}+\frac{1}{K^{r_b}}\Big)^2+\sqrt{\frac{K\log K}{n}}\cdot \Big(\frac{K^{1/2}}{n^{1/2}}+\frac{1}{K^{r_b}}\Big)\Big).
$$
This completes the proof.
\end{proof}

\begin{proof}[Proof of Theorem \ref{thmnon}]
We now consider the following decomposition of $\tilde\mu_{\tilde\bbeta}-\mu$,
\begin{align*}
\tilde\mu_{\tilde\bbeta}-\mu&=\frac{1}{n}\sum_{i=1}^n\Big[\frac{T_i(Y_i(1)-K(\bX_i)-L(\bX_i))}{\tilde J_i}-\frac{(1-T_i)(Y_i(0)-K(\bX_i))}{1-\tilde J_i}\Big]\\
&~~~~~~~~+\frac{1}{n}\sum_{i=1}^n\Big(\frac{T_i}{\tilde J_i}-\frac{1-T_i}{1-\tilde J_i}\Big)K(\bX_i)+\frac{1}{n}\sum_{i=1}^n\Big(\frac{T_i}{\tilde J_i}-1\Big)L(\bX_i)+\frac{1}{n}\sum_{i=1}^n L(\bX_i)-\mu\\
&=\frac{1}{n}\sum_{i=1}^n\Big[\frac{T_i(Y_i(1)-K(\bX_i)-L(\bX_i))}{\tilde J_i}-\frac{(1-T_i)(Y_i(0)-K(\bX_i))}{1-\tilde J_i}\Big]\\
&~~~~~~~~+\frac{1}{n}\sum_{i=1}^n\Big(\frac{T_i}{\tilde J_i}-\frac{1-T_i}{1-\tilde J_i}\Big)\Delta_K(\bX_i)+\frac{1}{n}\sum_{i=1}^n\Big(\frac{T_i}{\tilde J_i}-1\Big)\Delta_L(\bX_i)+\frac{1}{n}\sum_{i=1}^n L(\bX_i)-\mu,
\end{align*}
where $\tilde J_i=J(\tilde\bbeta^\top\bB(X_i))$, $\Delta_K(\bX_i)=K(\bX_i)-\balpha_1^{*\top}\bh_1(\bX_i)$ and $\Delta_L(\bX_i)=L(\bX_i)-\balpha_2^{*\top}\bh_2(\bX_i)$. Here, the second equality holds by the definition of $\tilde\bbeta$. Thus, we have
\begin{align*}
\tilde\mu_{\tilde\bbeta}-\mu&=\frac{1}{n}\sum_{i=1}^n S_i+R_0+R_1+R_2+R_3
\end{align*}
where
$$
S_i=\frac{T_i}{\pi^*_i}\big[Y_i(1)-K(\bX_i)-L(\bX_i)\big]-\frac{1-T_i}{1-\pi^*_i}\big[Y_i(0)-K(\bX_i)\big]+L(\bX_i)-\mu,
$$
$$
R_0=\frac{1}{n}\sum_{i=1}^n \frac{T_i(Y_i(1)-K(\bX_i)-L(\bX_i))}{\tilde J_i\pi^*_i}(\pi^*_i-\tilde J_i),
$$
$$
R_1=\frac{1}{n}\sum_{i=1}^n \frac{(1-T_i)(Y_i(0)-K(\bX_i))}{(1-\tilde J_i)(1-\pi^*_i)}(\pi^*_i-\tilde J_i),
$$
$$
R_2=\frac{1}{n}\sum_{i=1}^n\Big(\frac{T_i}{\tilde J_i}-\frac{1-T_i}{1-\tilde J_i}\Big)\Delta_K(\bX_i),
~~R_3=\frac{1}{n}\sum_{i=1}^n\Big(\frac{T_i}{\tilde J_i}-1\Big)\Delta_L(\bX_i).
$$
In the following, we will show that $R_j=o_p(n^{-1/2})$ for $0\leq j\leq 3$. Thus, the asymptotic normality of $n^{1/2}(\tilde\mu_{\tilde\bbeta}-\mu)$ follows from the previous decomposition. In addition, $S_i$ agrees with the efficient score function for estimating $\mu$ \citep{hahn1998role}. Thus, the proposed estimator $\tilde\mu_{\tilde\bbeta}$ is also semiparametrically efficient.

Now, we first focus on $R_0$. Consider the following empirical process $\GG_n(f_0)=n^{1/2}(\PP_n-\PP)f_0(T, Y(1), \bX)$, where $\PP_n$ stands for the empirical measure and $\PP$ stands for the expectation, and
$$
f_0(T, Y(1), \bX)=\frac{T(Y(1)-K(\bX)-L(\bX))}{J(m(\bX))\pi^*(\bX)}[\pi^*(\bX)-J(m(\bX))].
$$
By Lemma \ref{lemnonrate}, we can easily show that
\begin{align*}
\sup_{\bx\in\mathcal{X}}&|J(\tilde\bbeta^\top\bB(\bx))-\pi^*(\bx)|\lesssim \sup_{\bx\in\mathcal{X}}|\tilde\bbeta^\top\bB(\bx)-\bbeta^{*\top}\bB(\bx)|\\
&+\sup_{\bx\in\mathcal{X}}|m^*(\bx)-\bbeta^{*\top}\bB(\bx)|=O_p(K/n^{1/2}+K^{1/2-r_b})=o_p(1).
\end{align*}
For notational simplicity, we denote $\|f\|_\infty=\sup_{\bx\in\mathcal{X}}|f(\bx)|$. Define the set of functions $\mathcal{F}=\{f_0: \|m-m^*\|_\infty\leq \delta\}$, where $\delta=C(K/n^{1/2}+K^{1/2-r_b})$ for some constant $C>0$.
By the strong ignorability of the treatment assignment, we have that $\PP f_0(T, Y(1), \bX)=0$. By the Markov inequality and the maximal inequality in Corollary 19.35 of \cite{van2000asymptotic},
$$
n^{1/2}R_0\leq \sup_{f_0\in\mathcal{F}}\GG_n(f_0)\lesssim \EE \sup_{f_0\in\mathcal{F}}\GG_n(f_0)\lesssim J_{[~]}(\|F_0\|_{P,2}, \mathcal{F}, L_2(P)),
$$
where $J_{[~]}(\|F_0\|_{P,2}, \mathcal{F}, L_2(P))$ is the bracketing integral, and $F_0$ is the envelop function. Since $J$ is bounded away from 0, we have $|f_0(T, Y(1), \bX)|\lesssim \delta |Y(1)-K(\bX)-L(\bX)|:=F_0$. Then $\|F_0\|_{P,2}\leq \delta \{\EE|Y(1)|^2\}^{1/2}\lesssim \delta$. Next, we consider $N_{[~]}(\epsilon, \mathcal{F}, L_2(P))$. Define $\mathcal{F}_0=\{f_0: \|m-m^*\|_\infty\leq C\}$ for some constant $C>0$. Thus, it is easily seen that $\log N_{[~]}(\epsilon, \mathcal{F}, L_2(P))\lesssim \log N_{[~]}(\epsilon, \mathcal{F}_0\delta, L_2(P))=\log N_{[~]}(\epsilon/\delta, \mathcal{F}_0, L_2(P))\lesssim \log N_{[~]}(\epsilon/\delta, \mathcal{M}, L_2(P))\lesssim (\delta/\epsilon)^{1/k_1}$, where we use the fact that $J$ is bounded away from 0 and $J$ is Lipschitz. The last step follows from the assumption on the bracketing number of $\mathcal{M}$. Then
\begin{align*}
J_{[~]}(\|F_0\|_{P,2}, \mathcal{F}, L_2(P))\lesssim \int_0^\delta \sqrt{\log N_{[~]}(\epsilon, \mathcal{F}, L_2(P))}d\epsilon \lesssim \int_0^\delta (\delta/\epsilon)^{1/(2k_1)}d \epsilon,
\end{align*}
which goes to 0, as $\delta\rightarrow 0$, because $2k_1>1$ by assumption and thus the integral converges. Thus, this shows that $n^{1/2}R_0=o_p(1)$. By the similar argument, we can show that $n^{1/2}R_1=o_p(1)$.

Next, we consider $R_2$. Define the following empirical process $\GG_n(f_2)=n^{1/2}(\PP_n-\PP)f_2(T, \bX)$, where
$$
f_2(T, \bX)=\frac{T-J(m(\bX))}{J(m(\bX))(1-J(m(\bX)))}\Delta_K(\bX).
$$
By the assumption on the approximation property of the basis functions, we have $\|\Delta_K\|_\infty\lesssim K^{-r_h}$. In addition,
\begin{align*}
\|J(\tilde\bbeta^\top\bB(\bX))-\pi^*(\bX)\|_{P,2}&\leq \|J(\tilde\bbeta^\top\bB(\bX))-J(\bbeta^{*\top}\bB(\bX))\|_{P,2}+\|J(\bbeta^{*\top}\bB(\bX))-\pi^*(\bX)\|_{P,2}\\
&\lesssim \|\tilde\bbeta^\top\bB(\bX)-\bbeta^{*\top}\bB(\bX)\|_{P,2}+\sup_{\bx\in\mathcal{X}}|m^*(\bx)-\bbeta^{*\top}\bB(\bx)|\\
&=O_p(K^{1/2}/n^{1/2}+K^{-r_b}),
\end{align*}
where the last step follows from Lemma \ref{lemnonrate}.

Define the set of functions $\mathcal{F}=\{f_2: \|m-m^*\|_{P,2}\leq \delta_1, \|\Delta\|_\infty\leq \delta_2\}$, where $\delta_1=C(K^{1/2}/n^{1/2}+K^{-r_b})$ and $\delta_2=CK^{-r_h}$ for some constant $C>0$. Thus,
$$
n^{1/2}R_2\leq \sup_{f_2\in\mathcal{F}}\GG_n(f_2)+n^{1/2}\sup_{f_2\in\mathcal{F}}\PP f_2.
$$
We first consider the second term $n^{1/2}\sup_{f_2\in\mathcal{F}}\PP f_2$. Let $\mathcal{G}_1=\{m\in\mathcal{M}:  \|m-m^*\|_{P,2}\leq \delta_1\}$ and $\mathcal{G}_2=\{\Delta\in\mathcal{H}-\balpha_1^{*\top}\bh_1:  \|\Delta\|_{\infty}\leq \delta_2\}$.
By the definition of the propensity score and Cauchy inequality,
\begin{align*}
n^{1/2}\sup_{f_2\in\mathcal{F}}\PP f_2&=n^{1/2}\sup_{m\in\mathcal{G}_1, \Delta\in\mathcal{G}_2}\EE \frac{\pi^*(\bX)-J(m(\bX))}{J(m(\bX))(1-J(m(\bX)))}\Delta(\bX)\\
&\lesssim n^{1/2}\sup_{m\in\mathcal{G}_1}\|\pi^*-J(m)\|_{P,2}\sup_{\Delta\in\mathcal{G}_2}\|\Delta\|_{P,2}\\
&\lesssim n^{1/2}\delta_1\delta_2\lesssim n^{1/2}(K^{1/2}/n^{1/2}+K^{-r_b})K^{-r_h}=o(1),
\end{align*}
where the last step follows from $r_h>1/2$ and the scaling assumption $n^{1/2}\lesssim K^{r_b+r_h}$ in this theorem. Next, we need to control the maximum of the empirical process $\sup_{f_2\in\mathcal{F}}\GG_n(f_2)$. Following the similar argument to that for $R_0$, we only need to upper bound the bracketing integral $J_{[~]}(\|F_2\|_{P,2}, \mathcal{F}, L_2(P))$. Since $J$ is bounded away from 0 and 1, we can set the envelop function to be $F_2:=C\delta_2$ for some constant $C>0$ and thus $\|F_2\|_{P,2}\lesssim \delta_2$. Define $\mathcal{F}_0=\{f_2: \|m-m^*\|_{P,2}\leq C, \|\Delta\|_{P,2}\leq 1\}$ for some constant $C>0$, $\mathcal{G}_{10}=\{m\in\mathcal{M}+m^*:  \|m\|_{P,2}\leq C\}$ and $\mathcal{G}_{20}=\{\Delta\in\mathcal{H}-\balpha_1^{*\top}\bh_1:  \|\Delta\|_{P,2}\leq 1\}$. Similarly, we have
\begin{align*}
\log N_{[~]}(\epsilon, \mathcal{F}, L_2(P))&\lesssim \log N_{[~]}(\epsilon/\delta_2, \mathcal{F}_0, L_2(P))\\
&\lesssim \log N_{[~]}(\epsilon/\delta_2, \mathcal{G}_{10}, L_2(P))+\log N_{[~]}(\epsilon/\delta_2, \mathcal{G}_{20}, L_2(P))\\
&\lesssim \log N_{[~]}(\epsilon/\delta_2, \mathcal{M}, L_2(P))+\log N_{[~]}(\epsilon/\delta_2, \mathcal{H}, L_2(P))\\
&\lesssim (\delta_2/\epsilon)^{1/k_1}+(\delta_2/\epsilon)^{1/k_2},
\end{align*}
where the second step follows from the boundness assumption on $J$ and its Lipschitz property, the third step is due to $\mathcal{G}_{10}-m^*\subset \mathcal{M}$ and $\mathcal{G}_{20}+\balpha_1^{*\top}\bh_1\subset \mathcal{H}$ and the last step is by the bracketing number condition in our assumption.  Since $2k_1>1$ and $2k_2>1$, it is easily seen that the bracketing integral $J_{[~]}(\|F_2\|_{P,2}, \mathcal{F}, L_2(P))=o(1)$. This shows that $\sup_{f_2\in\mathcal{F}}\GG_n(f_2)=o_p(1)$.  Thus, we conclude that $n^{1/2}R_2=o_p(1)$. By the similar argument, we can show that $n^{1/2}R_3=o_p(1)$. This completes the whole proof.
\end{proof}

\section{Discussion on the Results in Section \ref{secnon}}

Under the conditions in Theorem \ref{thmnon}, it is well known that the convergence rate for estimating $K(\bx)$ (and also $L(\bx)$, $\psi^*(\bx)$) in the $L_2(P)$ norm (i.e, $\int (\hat K(\bx)-K(\bx))^2P(d\bx)$) is $O_p(\kappa^{-2r_h}+\kappa/n)$; see \cite{newey1997convergence}. Thus, the optimal choice of $\kappa$ that minimizes the rate is $\kappa\asymp n^{1/(2r_h+1)}$. Assume that $r_b=r_h$. With $\kappa\asymp n^{1/(2r_h+1)}$, the conditions $\kappa=o(n^{1/3})$ and $n^{\frac{1}{2(r_b+r_h)}}=o(\kappa)$ always hold as long as $r_h>1$. Recall that from the previous discussion $r_h=s/d$, where $s$ is the smoothness parameter and $d$ is the dimension of $\bX$. Thus under very mild conditions $s>d$, we do not need to under-smooth the estimator.

\begin{remark}
  By the proof of Theorem \ref{thmnon}, we find that when
  $\kappa=o(n^{1/(2r_b+1)})$ and $\kappa=o(n^{1/(2r_h+1)})$ hold, the asymptotic
  bias of the estimator $\tilde{\mu}_{\tilde\bbeta}$ is of order
  $O_p(K^{-(r_b+r_h)})$, which is the product of the approximation
  errors for $\psi^*(\bx)$ and $K(\bx)$ (also $L(\bx)$). Thus, to make
  the bias of the estimator $\tilde{\mu}_{\tilde\bbeta}$
  asymptotically ignorable, we can require either $r_b$ or $r_h$
  sufficiently large (not necessarily both). This phenomenon can be
  viewed as the double robustness property in the nonparametric
  context, which holds for the kernel based doubly robust estimator
  \citep{rothe2013semiparametric} and the targeted maximum likelihood
  estimator \citep{benkeser2016doubly}.  In addition, our estimator
  has smaller asymptotic bias than the usual nonparametric method. For
  simplicity, assume $r_b=r_h=r$. The asymptotic bias of the IPTW
  estimator in \cite{hira:imbe:ridd:03} is generally of order
  $O_p(\kappa^{-r})$, whereas our estimator has a smaller bias of order
  $O_p(\kappa^{-2r})$. 
\end{remark}

\section{Estimation of ATT}\label{app_att}
We consider the estimation of the average treatment effect for the treated (ATT)
$$
\tau^*=\EE(Y_i(1)-Y_i(0)|T_i=1).
$$
Let
$\tau^*_1=\EE(Y_i(1)\mid T_i=1)$ and $\tau^*_0=\EE(Y_i(0)\mid
T_i=1)$. By the law of total probability,
$$
\tau^*_1=\EE(T_iY_i(1)\mid T_i=1)=\EE(T_iY_i(1))/\PP(T_i=1).
$$
Thus, a simple estimator of $\tau^*_1$ is
$$
\hat\tau_1=\frac{\sum_{i=1}^n T_iY_i}{\sum_{i=1}^n T_i}.
$$
To estimate $\tau_0^*$, we notice that
\begin{align*}
\tau_0^*&=\EE[\EE(Y_i(0)\mid T_i=1,\bX_i)\mid T_i=1]=\EE[\EE(Y_i(0)\mid \bX_i)\mid T_i=1]\\
&= \frac{\EE[T_i  \EE(Y_i(0)\mid \bX_i)]} {\PP(T_i=1)}
=\frac{\EE(\pi(\bbeta^{*\top}\bX_i)\EE(Y_i(0)\mid \bX_i))}{\PP(T_i=1)}\\
&=\frac{1}{\PP(T_i=1)}\EE\left\{\frac{\pi(\bbeta^{*\top}\bX_i)(1-T_i)Y_i(0)}{1-\pi(\bbeta^{*\top}\bX_i)}\right\}.
\end{align*}
Similar to the bias and variance calculation for the ATE, we can estimate $\bbeta$ by the solving the following estimating equations
$$
{n}^{-1}\sum_{i=1}^n \left(T_i-\frac{(1-T_i)\pi(\bbeta^\top\bX)}{1-\pi(\bbeta^\top\bX)}\right)\fb(\bX)=0.
$$
Then, we set $\hat\pi_i=\pi(\hat\bbeta^\top\bX_i)$ and estimate $\tau_0$ by
$$
\hat\tau_0=\frac{\sum_{i=1}^n(1-T_i)\hat r_iY_i}{\sum_{i=1}^n (1-T_i) \hat r_i},
$$
where $\hat r_i=\hat\pi_i/(1-\hat\pi_i)$. The final estimator of
the ATT is $\hat\tau=\hat\tau_1-\hat\tau_0$. Similar to the proof of the main results on ATE, we can show that when both models are correct, $n^{1/2}(\hat{\tau}-\tau^*)\rightarrow_d N(0,W)$, where
$$
W \ = \ p^{-2}\EE\left[\pi^*\EE(\epsilon_1^2\mid \bX)+\frac{\pi^{*2}}{1-\pi_i^*}\EE(\epsilon_0^2\mid \bX)+\pi^*(L(\bX_i)-\tau^*)^2\right].
$$
Here, $\epsilon_0=Y(0)-K(\bX)$, $\epsilon_1=Y(1)-K(\bX)-L(\bX)$ and $p=\PP(Y=1)$

\section{Derivation of (\ref{eqmis_cbps}) and  (\ref{eqmis_aipw})}\label{sec_mis_cbps_aipw}

In this appendix, we only provide a sketch of the proof of (\ref{eqmis_cbps}) and  (\ref{eqmis_aipw}), because the detail is very similar to the proof of Theorem~\ref{th:asymnormal_mis}. Recall that as in Section \ref{sec:model.misspecification}, $\bbeta^o$ which satisfies $\EE(\bar\bg_{\bbeta^o}(\bT,\bX))=0$ is the limiting value of $\hat\bbeta$ as in Lemma \ref{lemconsistency}.  In addition, denote $K^o(\bX_i)=\balpha^{*T}\bh_1(\bX_i)+\delta \bA_1\bh_1(\bX_i)$ and $L^o(\bX_i)=\bgamma^{*T}\bh_2(\bX_i)+\delta \bA_2\bh_2(\bX_i)$, where the vectors $\bA_1$ and $\bA_2$ are to be determined. We have the following decomposition 
\begin{align*}
\hat\mu_{\hat\bbeta}-\mu&=\frac{1}{n}\sum_{i=1}^n\Big[\frac{T_i}{\pi_{\beta^o}(\bX_i)}\{Y_i(1)-K^o(\bX_i)-L^o(\bX_i)\}-\frac{1-T_i}{1-\pi_{\beta^o}(\bX_i)}\{Y_i(0)-K^o(\bX_i)\}+L^o(\bX_i)-\mu\Big]\\
&+\frac{1}{n}\sum_{i=1}^n \Big\{\frac{T_i}{\pi_{\hat\beta}(\bX_i)}-\frac{T_i}{\pi_{\beta^o}(\bX_i)}\Big\}\{Y_i(1)-K^o(\bX_i)-L^o(\bX_i)\}\\
&-\frac{1}{n}\sum_{i=1}^n\Big\{\frac{1-T_i}{1-\pi_{\hat\beta}(\bX_i)}-\frac{1-T_i}{1-\pi_{\beta^o}(\bX_i)}\Big\}\{Y_i(0)-K^o(\bX_i)\}:=I_1+I_2+I_3.
\end{align*}
We first consider $I_2$. The mean value theorem implies
$$
I_2=-\frac{1}{n}\sum_{i=1}^n \frac{T_i}{\pi^2_{\tilde\beta}(\bX_i)}\frac{\partial\pi_{\tilde\beta}(\bX_i)}{\partial\beta}\{Y_i(1)-K^o(\bX_i)-L^o(\bX_i)\} (\hat\beta-\beta^o),
$$
where $\tilde\beta$ is an intermediate value between $\hat\beta$ and $\beta^o$. Under Assumptions similar to  \ref{reg_mis}, the dominated convergence theorem implies 
$$
-\frac{1}{n}\sum_{i=1}^n \frac{T_i}{\pi^2_{\tilde\beta}(\bX_i)}\frac{\partial\pi_{\tilde\beta}(\bX_i)}{\partial\beta}\{Y_i(1)-K^o(\bX_i)-L^o(\bX_i)\}=O_p(\delta).
$$
Similar to Lemma \ref{lemAN}, we can show that $\hat\beta-\beta^o=O_p(n^{-1/2})$. The Slutsky theorem yields $I_2=O_p(\delta n^{-1/2})$. The same argument implies that $I_3=O_p(\delta n^{-1/2})$.  Finally, we focus on $I_1$. Note that
$$
I_1-\frac{1}{n}\sum_{i=1}^n\Big[\frac{T_i}{\pi(\bX_i)}\{Y_i(1)-K(\bX_i)-L(\bX_i)\}-\frac{1-T_i}{1-\pi(\bX_i)}\{Y_i(0)-K(\bX_i)\}+L(\bX_i)-\mu\Big]=\frac{1}{n}\sum_{i=1}^n\Delta_i,
$$
where
\begin{align*}
\Delta_i&=\Big\{\frac{T_i}{\pi_{\beta^o}(\bX_i)}-\frac{T_i}{\pi(\bX_i)}\Big\}\{Y_i(1)-K^o(\bX_i)-L^o(\bX_i)\}\\
&-\Big\{\frac{1-T_i}{1-\pi_{\beta^o}(\bX_i)}-\frac{1-T_i}{1-\pi(\bX_i)}\Big\}\{Y_i(0)-K^o(\bX_i)\}\\
&-\frac{T_i}{\pi(\bX_i)}\{K^o(\bX_i)+L^o(\bX_i)-K(\bX_i)-L(\bX_i)\}\\
&+\frac{1-T_i}{1-\pi(\bX_i)}\{K^o(\bX_i)-K(\bX_i)\}+L^o(\bX_i)-L(\bX_i).
\end{align*}
The central limit theorem implies $n^{1/2}(\frac{1}{n}\sum_{i=1}^n\Delta_i-\EE\Delta_i)/sd(\Delta_i)\rightarrow N(0,1)$. In order to derive the order of $\frac{1}{n}\sum_{i=1}^n\Delta_i$, it suffices to compute the $\EE(\Delta_i)$ and $sd(\Delta_i)$. As in the derivation of (\ref{eqbetao}), after some algebra, we similarly obtain
$$
\beta^o-\beta^*=\xi \Tb^{-1}\Mb+O(\xi^2),
$$
where 
$$
\Mb=\left( \begin{array}{c}
	\EE(\frac{1}{1-\pi_{\beta^*}(\bX_i)}u_i^*\bh_1(\bX_i)) \\
	\EE(u_i^*\bh_2(\bX_i) \end{array} \right)
$$
and $\Tb=[\EE(\frac{1}{\pi_{\beta^*}(\bX_i)(1-\pi_{\beta^*}(\bX_i))}\frac{\partial\pi_{\beta^*}(\bX_i)}{\partial\beta}\bh_1^T(\bX_i)),\EE(\frac{1}{\pi_{\beta^*}(\bX_i)}\frac{\partial\pi_{\beta^*}(\bX_i)}{\partial\beta}\bh_2^T(\bX_i))]^T$. 

Denote $\tilde r_1(\bX_i)=r_1(\bX_i)-\bA_1\bh_1(\bX_i)$ and $\tilde r_2(\bX_i)=r_2(\bX_i)-\bA_2\bh_2(\bX_i)$. Note that
\begin{align*}
\EE(\Delta_i)&=\EE\Big\{\frac{\pi(\bX_i)}{\pi_{\beta^o}(\bX_i)}\delta(\tilde r_1(\bX_i)+\tilde r_2(\bX_i))-\frac{1-\pi(\bX_i)}{1-\pi_{\beta^o}(\bX_i)}\delta\tilde r_1(\bX_i)-\delta\tilde r_2(\bX_i)\Big\}\\
&=\EE\Big\{\{1+\xi u_i^*-\frac{1}{\pi_{\beta^*}(\bX_i)}\frac{\partial\pi_{\beta^*}(\bX_i)}{\partial\beta}(\beta^o-\beta^*)\}\delta(\tilde r_1(\bX_i)+\tilde r_2(\bX_i))\\
&-\{1-\frac{\pi_{\beta^*}(\bX_i)}{1-\pi_{\beta^*}(\bX_i)}\xi u_i^*+\frac{1}{1-\pi_{\beta^*}(\bX_i)}\frac{\partial\pi_{\beta^*}(\bX_i)}{\partial\beta}(\beta^o-\beta^*)\}\delta\tilde r_1(\bX_i)-\delta\tilde r_2(\bX_i))\Big\}+O(\xi^2\delta)\\
&=\xi\delta\EE\Big[\{\frac{1}{1-\pi_{\beta^*}(\bX_i)} u_i^*-\frac{1}{(1-\pi_{\beta^*}(\bX_i))\pi_{\beta^*}(\bX_i)}\frac{\partial\pi_{\beta^*}(\bX_i)}{\partial\beta}\Tb^{-1}\Mb\}\tilde r_1(\bX_i)\Big]\\
&+\xi\delta\EE\Big[\{ u_i^*-\frac{1}{\pi_{\beta^*}(\bX_i)}\frac{\partial\pi_{\beta^*}(\bX_i)}{\partial\beta}\Tb^{-1}\Mb\}\tilde r_2(\bX_i)\Big]+O(\xi^2\delta).
\end{align*}
Assume that at least one entry of $\EE\Big[\{\frac{1}{1-\pi_{\beta^*}(\bX_i)} u_i^*-\frac{1}{(1-\pi_{\beta^*}(\bX_i))\pi_{\beta^*}(\bX_i)}\frac{\partial\pi_{\beta^*}(\bX_i)}{\partial\beta}\Tb^{-1}\Mb\}\bh_1(\bX_i)\Big]$ is nonzero. Then, there exists $\bA_1$ such that 
\begin{align*}
&\bA_1\EE\Big[\{\frac{1}{1-\pi_{\beta^*}(\bX_i)} u_i^*-\frac{1}{(1-\pi_{\beta^*}(\bX_i))\pi_{\beta^*}(\bX_i)}\frac{\partial\pi_{\beta^*}(\bX_i)}{\partial\beta}\Tb^{-1}\Mb\}\bh_1(\bX_i)\Big]\\
&=\EE\Big[\{\frac{1}{1-\pi_{\beta^*}(\bX_i)} u_i^*-\frac{1}{(1-\pi_{\beta^*}(\bX_i))\pi_{\beta^*}(\bX_i)}\frac{\partial\pi_{\beta^*}(\bX_i)}{\partial\beta}\Tb^{-1}\Mb\} r_1(\bX_i)\Big],
\end{align*}
which implies 
$$
\EE\Big[\{\frac{1}{1-\pi_{\beta^*}(\bX_i)} u_i^*-\frac{1}{(1-\pi_{\beta^*}(\bX_i))\pi_{\beta^*}(\bX_i)}\frac{\partial\pi_{\beta^*}(\bX_i)}{\partial\beta}\Tb^{-1}\Mb\}\tilde r_1(\bX_i)\Big]=0.
$$
Similarly, by choosing a proper $\bA_2$, we have
$$
\EE\Big[\{ u_i^*-\frac{1}{\pi_{\beta^*}(\bX_i)}\frac{\partial\pi_{\beta^*}(\bX_i)}{\partial\beta}\Tb^{-1}\Mb\}\tilde r_2(\bX_i)\Big]=0.
$$
As a result, we obtain $\EE(\Delta_i)=O(\xi^2\delta)$. Finally, after some tedious calculation, we can show that $sd(\Delta_i)=O(\xi+\delta)$. This implies $\frac{1}{n}\sum_{i=1}^n\Delta_i=O_p(\xi^2\delta+\xi n^{-1/2}+\delta n^{-1/2})$. This completes the proof of (\ref{eqmis_cbps}). The proof of (\ref{eqmis_aipw}) follows from the similar argument and we omit the details. 

\section{Asymptotic Variance Formulas Used for Simulations}\label{asy_var_formulas}
In this appendix, we present the asymptotic variance formulas used for constructing the 95\% confidence intervals for calculating the coverage probabilities in the simulations in Section~\ref{sec:simulations}. In particular, for a generic estimator $\hat\mu$, the 95\% confidence interval is $(\hat\mu-1.96*\hat \sigma, \hat\mu+1.96*\hat \sigma)$, where $\hat\sigma^2$ is the estimate of the asymptotic variance of $\sqrt{n}(\hat\mu-\mu)$. 

For the {\sf True} estimator, the asymptotic variance formula is similar to the one given in Section~\ref{sec:model.misspecification} and is as follows:
\begin{eqnarray*}
& \Sigma_{\mu_0}  =   \Var\bigl(\mu_{\bbeta_0}(T_i, Y_i, \bX_i) \bigr) \ = \ \mathbb{E}\biggl(\frac{Y_i(1)^2}{\pi_{\bbeta_0}(\bX_i)} + \frac{Y_i(0)^2}{1-\pi_{\bbeta_0}(\bX_i)}-(\mathbb{E}(Y_i(1))-\mathbb{E}(Y_i(0)))^2\biggr).
\end{eqnarray*}

For the {\sf GLM} estimator, the asymptotic variance formula is as follows:
\begin{eqnarray*}
& \Sigma_{\textrm{GLM}}  =   \Sigma_{\mu_0} - \bH_y^\top \bI^{-1} \bH_y
\end{eqnarray*}
where $\Sigma_{\mu_0}$ is defined like before, $\bI$ is the Fisher Information Matrix, and 
\begin{eqnarray*}
\bH_y  & = &  -\EE\left(\frac{K(\bX_i)+(1-\pi_{\bbeta_0}(\bX_i))L(\bX_i)}{\pi_{\bbeta_0}(\bX_i) (1-\pi_{\bbeta_0}(\bX_i))}\cdot \frac{\partial \pi_{\bbeta_0}(\bX_i)}{\partial\bbeta}\right).
\end{eqnarray*}
Since the second term is positive definite, $ \Sigma_{GLM} < \Sigma_{\mu} $ and thus the variance decreases. 

The {\sf GAM} estimator achieves the semiparametric efficiency bound \citep{hira:imbe:ridd:03} and so we can use $V_{\textrm{opt}}$ given in (\ref{eqinforbound1}) as the asymptotic variance formula. The {\sf CBPS} estimator has the following asymptotic variance formula:
\begin{align*}
\Sigma_{\textrm{CBPS}} ~~ = ~~ \Sigma_{\mu_0} ~~ + ~~ & \bH_y^\top(\bH_{\fb}^\top\bOmega^{-1}\bH_{\fb})^{-1}\bH_y \\
- ~~ &2\bH_y^\top(\bH_{\fb}^\top\bOmega^{-1}\bH_{\fb})^{-1}\bH_{\fb}^\top\bOmega^{-1}\Cov(\mu_{\bbeta_0}(T_i, Y_i, \bX_i),\bg_{\bbeta_0}(T_i, \bX_i))
\end{align*}
where $\Sigma_{\mu_0}$ and $\bH_y$ are defined like before, and we have: 
\begin{eqnarray*}
\bH_{\fb}  & = &  -\EE \left(\frac{\fb(\bX_i)}{\pi_{\bbeta_0}(\bX_i)(1-\pi_{\bbeta_0}(\bX_i))}\left(\frac{\partial \pi_{\bbeta_0}(\bX_i)}{\partial \bbeta}\right)^{\top}\right)\\
\bOmega & = & \Var(\bg_{\bbeta_0}(T_i, \bX_i))\\
\bg_{\bbeta_0}(T_i,\bX_i) & =  & \ \left(\frac{T_i}{\pi_{\bbeta_0}(\bX_i)}-\frac{1-T_i}{1-\pi_{\bbeta_0}(\bX_i)} \right)\fb(\bX_i)\\
\mu_{\bbeta_0}(T_i, Y_i, \bX_i) & =  & \frac{T_iY_i}{\pi_{\bbeta_0}(\bX_i)}-\frac{(1-T_i)Y_i}{1-\pi_{\bbeta_0}(\bX_i)}.
\end{eqnarray*}

The asymptotic variance for the {\sf DR} estimator is automatically computed in the R package {\sf drtmle}
and the confidence interval was constructed accordingly. 

Finally, we note that when we estimate the asymptotic variances, we simply replace the quantities $\pi_{\beta_0}$ and $K(\bX)$ and $L(\bX)$ with their estimates and replace the expectation with the sample average. To save space, we do not repeat the formulas of the estimated variances.

\end{document}